%% file: lmcs.tex
\documentclass{lmcs} 
\pdfoutput=1
\usepackage[utf8]{inputenc}

\usepackage{lastpage}
\lmcsdoi{20}{1}{18}
\lmcsheading{}{\pageref{LastPage}}{}{}%
{Sep.~07,~2022}{Mar.~01,~2024}{}

\usepackage{coqdoc}


\usepackage{quiver}
\usepackage{hyperref}

\input{macros}

\usepackage{ebutf8}

\newenvironment{full}{\textcolor{red}{Long version}}{\textcolor{red}{End of long version}}
\newenvironment{short}{\textcolor{red}{Short version}}{\textcolor{red}{End of short version}}

\newenvironment{enumerati}{\begin{enumerate}[(i)]}{\end{enumerate}}

\usepackage{verbatim}
\usepackage{listings}

\theoremstyle{thmC}
\newtheorem{corC}[thm]{Corollary}

\begin{document}

\title[Variable binding and substitution for (nameless) dummies]{
  Variable binding and substitution \\ for (nameless) dummies}
\titlecomment{{\lsuper*}Extended abstract (FoSSaCS 2022)}

\author[A.~Hirschowitz]{André Hirschowitz\lmcsorcid{0000-0003-2523-1481}}[a]
\author[T.~Hirschowitz]{Tom Hirschowitz\lmcsorcid{0000-0002-7220-4067}}[b]
\author[A.~Lafont]{Ambroise Lafont\lmcsorcid{0000-0002-9299-641X}}[c]
\author[M.~Maggesi]{Marco Maggesi\lmcsorcid{0000-0003-4380-7691}}[d]

\address{Univ. Côte d'Azur, CNRS, LJAD, 06103, Nice, France}	

\address{Univ. Savoie Mont Blanc, CNRS, LAMA, 73000,
  Chambéry, France}	

\address{LIX, École polytechnique, Institut Polytechnique de Paris,
Palaiseau, France}	
\address{Università degli Studi di Firenze, Italy}	





\begin{abstract}
  \noindent By abstracting over well-known properties of De Bruijn's
  representation with nameless dummies, we design a new theory of syntax
  with variable binding and capture-avoiding substitution.  We propose
  it as a simpler alternative to Fiore, Plotkin, and Turi's approach,
  with which we establish a strong formal link. We also show that our
  theory easily incorporates simple types and equations between terms.
\end{abstract}

\maketitle

\section{Introduction}
In this paper we propose a new initial-algebra
semantics~\cite{goguen1974initial} for syntax and substitution in the
presence of variable binding, which gives a new perspective on the
status of the well-known De Bruijn encoding~\cite{DB}.

Given a so-called binding signature~\cite{BindingSignatures} (which we
suppose untyped in this introduction), De Bruijn's encoding provides
an explicit definition of the desired syntax; it consists of a
(single) set of terms, equipped with a suitable operation of
``substitution''.  The salient feature of De Bruijn's encoding is that
variables are represented by natural numbers, which he termed
``nameless dummies'', hence the title of the present paper.  The idea
is that any occurrence of $0$ refers to the binder just above it (in
the abstract syntax tree), if any, while $1$ refers to the next one
up, and so on.  E.g., $λx.λy.(x\ y)$ is represented by $λ.λ.(1\ 0)$.
See~\cite{fiore:presheaf,ShulmanDB} for more recent analyses.  This
encoding is generally considered ``good for machine implementations,
but not […]  for machine-assisted human
reasoning''~\cite{PittsAM:newaas} (see also
\cite{DBLP:conf/tphol/AydemirBFFPSVWWZ05,DBLP:journals/entcs/BerghoferU07}).

Our initial-algebra semantics provides an alternative to the above
\emph{explicit} definition, by offering an \emph{implicit} one:
\begin{itemize}
\item We design a category of ``models'' of the considered signature.
\item We define the desired syntax (up to unique isomorphism) as the
  initial object in this category.
\end{itemize}
One may then reason about syntax independently of any chosen initial
object, since initiality provides a convenient induction principle.

Of course, we have to prove that such an initial object exists, and
the natural witness in this proof is precisely De Bruijn's encoding.
It thus acquires the new status of initiality witness, and hence may
be forgotten, to some extent.

We know of two initial-algebra semantics for syntax with substitution
in the presence of variable binding.  A mainstream one is by Fiore et
al.~\cite{fiore:presheaf,DBLP:conf/lics/Fiore08}, while the second
one, which also handles linear syntax, is due to Power~\cite{Power}.  Both
approaches consider terms indexed by the number of (potential) free
variables. By contrast, ours involves a single (infinite and implicit)
context. It is thus simpler, at least in the sense that it can
naturally be implemented in a proof assistant without dependent types.
We demonstrate this by implementing our framework in HOL Light.  We
also provide a Coq implementation for comparison.

Let us emphasise that our initial-algebra semantics optimises the
usual layering into (1) syntax, (2) variable renaming, and (3)
substitution.  Indeed, we show that the second layer is unnecessary,
and directly give the implicit definition of syntax with substitution
in (unindexed) sets.

A consequence is that our mechanisations offer a very different
\alert{trusted computing base}\footnote{I.e., the part of the
  development that needs to be read in order to check that the
  definitions and statements are correct.} from what one usually gets
with an explicit definition.
\begin{itemize}
\item With an explicit definition, the trusted computing base
  typically consists of
  \begin{itemize}
  \item the inductive type defining the syntax,
  \item the recursive definition of renaming, and
  \item the recursive definition of substitution.
\end{itemize}
\item By contrast, in our mechanisations, the trusted computing base
  consists of
  \begin{itemize}
  \item the definition of the category of models, and
  \item the initiality statement.
  \end{itemize}
\end{itemize}
As the authors have experienced, the pros and cons can be discussed
\emph{ad libitum}. We refrain from doing so in this paper.

\subsection{Overview}
Let us now present our contribution in a bit more detail, for
which we should start by recalling binding signatures.
\begin{defi}
  A \alert{binding arity} is a sequence of natural numbers.  A
  \alert{binding signature} is a set $O$ (of "operations"), together with
  a map $O → ℕ^*$, which associates a binding arity to each operation.
\end{defi}
The idea is that an operation of binding arity $b = (n₁,…,nₚ)$ has $p$
arguments, with the $i$th argument binding $nᵢ$ variables, for all
$i ∈ \{1,…,p\}$.
\begin{exa}\label{ex:lamapp}
  In pure $λ$-calculus, the binding arity for application is $(0,0)$:
  it has two arguments, binding no variables.  Abstraction, on the
  other hand, has one argument which binds one variable.  Its binding
  arity thus is the singleton sequence $(1)$.
\end{exa}

We should now answer the question: where do operations of a given
binding arity live, and what are they?
To the first question, we answer that they live in a De Bruijn monad,
whose definition we now sketch.
\begin{defi}
  A \alert{De Bruijn monad} is a set $X$, equipped with
  \begin{itemize}
  \item a \alert{variables} map $v∶ ℕ → X$, and
  \item a \alert{substitution} map $s∶ X × X^ℕ→ X$, which takes an
    element $x ∈ X$ and an assignment $σ∶ ℕ → X$, and returns an
    element $s(x,σ)$, which we denote by $x[σ]$ when $s$ is clear from
    context,
  \end{itemize}
  satisfying three simple axioms (see Definition~\ref{def:dbmonad} below).
\end{defi}
\begin{rem}
  The use of the word ``monad'' is justified by the fact that De
  Bruijn monads are in fact relative
  monads~\cite{DBLP:journals/corr/AltenkirchCU14}, see
  Corollary~\ref{cor:skew} below.
\end{rem}

To the second question, what is an operation of a given binding arity
in a De Bruijn monad $(X,v,s)$, we answer as follows.
\begin{defi}
  An \alert{operation of binding arity $b = (n₁,…,nₚ)$} is a map
  $o∶ Xᵖ → X$ satisfying the following \alert{binding condition}:
  for all $e₁,…,eₚ ∈ X$, and $σ∶ ℕ → X$,
  \begin{equation}
    o(e₁,…,eₚ)[σ] = o(e₁[⇑^{n₁} σ],…,eₚ[⇑^{nₚ} σ])\rlap{,}
    \label{eq:subst:o}
  \end{equation}
  where $⇑$ is a unary operation defined on $X^ℕ$ by
    \begin{align}
      (⇑σ)(0) & = v(0) \label{eq:subst:0} \\
      (⇑σ)(n+1) & = σ(n)[p↦v(p+1)]\rlap{.} \label{eq:subst:S}
    \end{align}
\end{defi}
To explain the idea behind this definition, let us consider the simplest non-trivial binding operation,
abstraction in the pure $λ$-calculus, which has arity $(1)$.
The main idea of De Bruijn indices is that, under an abstraction:
\begin{enumerate}[(a)]
\item \label{item:bv} the bound variable is $0$, and
\item \label{item:ov} any outer variable $k$ should be referred to as
  $1+k$.
\end{enumerate}
By the binding condition, for any operation $λ∶ X → X$ of binding
arity $(1)$, we have $λ(e)[σ] = λ(e[⇑ σ])$.  By definition, $⇑σ$
leaves the variable $0$ unchanged, which complies with~\ref{item:bv}
above.  Furthermore, by definition, any reference $1+k$ to some outer
variable is mapped by $⇑σ$ to $σ(k)[p↦v(p+1)]$. That is, the intended
element $σ(k)$, whose free variables are shifted by one to comply
with~\ref{item:ov}.

From here, we straightforwardly define models of a given binding
signature $S$ to be De Bruijn monads equipped with operations of the
specified binding arities.  We call such models \alert{De Bruijn
  $S$-algebras}, and organise them into a category $S\DBAlg$.

Finally, we prove that $S\DBAlg$ admits an initial object
(Theorem~\ref{thm:initiality:II}).  For this, we follow (the standard
modern variant of) De Bruijn's construction:
\begin{itemize}
\item We extract from $S$ a first-order signature $|S|$, by mapping
  binding arities $(n₁,…,nₚ)$ to their lengths $p$, and construct the
  free $|S|$-algebra $\DB$ over the set $ℕ$ of variables in the usual,
  first-order way.
\item We prove that $\DB$ admits a unique substitution map satisfying
  both the binding conditions and the De Bruijn monad axioms.  This is
  not entirely trivial, because we cannot directly
  take~\eqref{eq:subst:o}--\eqref{eq:subst:S} as a recursive
  definition. Indeed, the recursive call in~\eqref{eq:subst:S} would
  not be decreasing, at least in any standard proof assistant's sense!
  We thus resort to the usual, two-phase construction:
  \begin{itemize}
  \item We first define a renaming map $\DB × ℕ^ℕ → \DB$, by
    adapting~\eqref{eq:subst:o}--\eqref{eq:subst:S} to the renaming
    case.
  \item We then define the substitution map
    by~\eqref{eq:subst:o}--\eqref{eq:subst:S}, except that we replace
    the problematic recursive call in~\eqref{eq:subst:S} by
    $σ(n)[p ↦ p+1]$, which is a renaming, hence non recursive.
  \end{itemize}
  We finally prove that this uniquely equips $\DB$ with De Bruijn
  $S$-algebra structure, and that the obtaind De Bruijn $S$-algebra is
  initial.
\end{itemize}

Once this initial-algebra semantics is in place, we investigate the
link with the above-mentioned mainstream framework of Fiore, Plotkin,
and Turi.  We find that both categories of models may include
pathological objects, in the sense that we do not see any loss in
ruling them out. When we do so, we obtain equivalent categories
(Theorem~\ref{thm:gap}).

Next, we devote two sections to investigating the status of binding
signatures and the binding conditions. Indeed, binding signatures are
combinatorial objects, and the binding conditions may seem somewhat
arbitrary.  We provide two categorical interpretations of binding
signatures and binding conditions.
\begin{itemize}
\item We first recast binding signatures within Borthelle et al.'s
  framework~\cite{BHL}, which is a generalisation of
  Fiore's~\cite{DBLP:conf/lics/Fiore08}.  After recalling the notion
  of \alert{structurally strong} endofunctor (on $𝐒𝐞𝐭$), and the
  category $Σ\Mon$ of models of such an endofunctor $Σ$, we show that
  any binding signature $S$ gives rise to such an endofunctor $Σ_S$,
  and exhibit an isomorphism $S\DBAlg ≅ Σ_S\Mon$ of categories over
  $𝐃𝐁𝐌𝐧𝐝$.
\item We then recast our initial-algebra semantics within the
  module-based approach to syntax with variable binding and
  substitution~\cite{hirscho:lam,DBLP:journals/iandc/HirschowitzM10}.
  For this, we need to adapt the notion of parametric module over
  monads to De Bruijn monads, thus introducing \alert{parametric De
    Bruijn modules}.  We further define the category $M\MAlg$ (for
  ``modular algebras'') of models of any such parametric De Bruijn
  module $M$.  Finally, we show that any binding signature $S$ gives
  rise to a parametric De Bruijn module $M_S$, and exhibit an
  isomorphism $S\DBAlg ≅ M_S\MAlg$ of categories over $𝐃𝐁𝐌𝐧𝐝$.
\end{itemize}
Our next two contributions extend the initial-algebra semantics
in two different directions.

\begin{itemize}
\item We first propose a simply-typed generalisation, which is
  parameterised over a given set of types. We adopt a standard
  simply-typed variant of binding signatures~\cite{FioreHur}, and
  prove a corresponding initiality result
  (Theorem~\ref{thm:typedinitiality}).  The strength-based and
  module-based recastings that we just mentioned could be extended to
  this setting, but we refrain from doing so for simplicity.
\item Then, we consider equations. We introduce a notion of \alert{De
    Bruijn equational theory}, and prove a corresponding initiality
  result (Theorem~\ref{thm:eqn}), whose witness is a straightforward
  quotient of De Bruijn's encoding.
\end{itemize}
Finally, in~§\ref{s:mechanised}, we provide two mechanised versions of
our framework: the first one is in Coq, while the second one is in HOL
Light, a proof assistant which does not support dependent types, thus
illustrating the simplicity of our theory.







\subsection{Plan of the paper}
In~§\ref{s:dbmonads}, we introduce De Bruijn monads, De Bruijn
$S$-algebras, and the De Bruijn $S$-algebra $\DB$. We furthermore
prove (Theorem~\ref{thm:initiality:I}) that $\DB$ admits a unique
substitution map satisfying the binding conditions with the desired
behaviour on variables.
In~§\ref{s:elementary}, we organise De Bruijn monads as a category,
which we prove equivalent to categories of relative monads and of
monoids.  For any binding signature $S$, we then organise De Bruijn
$S$-algebras into a category $S\DBAlg$, wherein we prove that $\DB$ is
an initial object.
In~§\ref{s:presheaves}, we establish the announced link with
the presheaf-based approach.
In~§\ref{s:strength} and~\ref{s:module-based}, we introduce our
interpretations of binding signatures and binding conditions in terms
of structurally strong endofunctors and modules, respectively.
We enrich the framework with simple types in~§\ref{s:types}, and with
equations in~§\ref{s:equations}.
In~§\ref{s:mechanised}, we briefly describe our mechanisations in HOL
Light and Coq.  Finally, we conclude in~§\ref{s:conclu}.


\subsection{Related work}
\paragraph{Abstract frameworks for variable binding}
We have already mentioned the tight link with the presheaf-based
approach~\cite{fiore:presheaf}.  This link could probably be extended
to variants such
as~\cite{hirscho:lam,DBLP:journals/iandc/HirschowitzM10,DBLP:conf/fscd/ArkorM21,FioreSzamozvancev}.

In recent work, Allais et al.~\cite{DBLP:journals/pacmpl/AllaisA0MM18}
introduce a universe of syntaxes, which essentially corresponds to a
simply-typed version of binding signatures. Their framework is
designed to facilitate the definition of so-called \alert{traversals},
i.e., functions defined by structural induction, ``traversing'' their
argument.  In a similar spirit, let us mention the recent work of
Gheri and Popescu~\cite{DBLP:journals/jar/GheriP20}, which presents a
theory of syntax with binding, mechanised in Isabelle/HOL.  Potential
links between these frameworks and our approach remain unclear to us
at the time of writing.

The categories of "intersectional" objects obtained
in~§\ref{s:presheaves} are technically very close to nominal
sets~\cite{PittsAM:newaas}: finite supports appear in the ``action-based''
presentation of nominal sets (and in our~§\ref{ss:wbdb}), while
pullback preservation appears in their sheaf-based presentation (and
in our~§\ref{ss:wbfpt}). And indeed, any intersectional presheaf yields
a nominal set, and so does any finitary De Bruijn monad.  However,
these links are not entirely satisfactory, because they do not account
for substitution. The reason is that the only categorical theory of
substitution that we know of for nominal sets, by Power~\cite{Power},
is operadic rather than monadic, so we do not immediately see how to
state a correspondence.

Finally, Pitts~\cite{10.1145/3571210} recently introduced semantics
for the locally nameless approach to syntax, where bound variables are
De Bruijn indices and free variables are chosen in a fixed infinite
set of atoms.  In some sense, his locally nameless sets are the
counterpart of our finitary De Bruijn monads, in the untyped case.
Beyond the difference between the locally nameless approach and the
crude De Bruijn encoding we focus on, while only single-variable
renamings are available in locally nameless sets, simultaneous
substitution is built-in in De Bruijn monads.  This enables us to
define a notion of model (for a binding signature) with explicit
compatibility conditions about substitution, resulting in a recursion
principle which is compatible with substitution.

\paragraph{Proof assistant libraries}
Allais et al.~\cite{DBLP:journals/pacmpl/AllaisA0MM18} and Gheri and
Popescu~\cite{DBLP:journals/jar/GheriP20} mechanise their approach in
Agda and Isabelle/HOL, respectively. In the same spirit, the
presheaf-based approach was recently
formalised~\cite{FioreSzamozvancev}.

De Bruijn representation benefits from well-developed proof assistant
libraries, in particular
Autosubst~\cite{DBLP:conf/itp/SchaferTS15,DBLP:conf/cpp/StarkSK19}.
Such libraries are somewhat complementary to our work.  Their main
goal is to automate part of the reasoning about substitution in the
proof assistant, while we provide an initial-algebra semantics.  In
particular, it could be useful to adapt the decision procedure of
Autosubst to our Coq library.

\subsection{General notation}
We denote by $A^* = ∑_{n ∈ ℕ} Aⁿ$ the set of finite sequences of
elements of $A$, for any set $A$.  In any category $𝐂$, we tend to
write $[C,D]$ for the hom-set $𝐂(C,D)$ between any two objects $C$ and
$D$.  Finally, for any endofunctor $F$, $F\alg$ denotes the usual
category of $F$-algebras and morphisms between them, and
$\mu F = \mu X. F(X)$ will be its least fixed point.  Finally, $𝐂𝐀𝐓$
denotes the large category of locally small categories.

\section{De Bruijn monads}\label{s:dbmonads}
In this section, we start by introducing De Bruijn monads in an
untyped setting.  Then, we define assignment lifting, the binding
conditions, and the models of a binding signature $S$ in De Bruijn
monads, De Bruijn $S$-algebras. Finally, we construct the term De
Bruijn $S$-algebra $\DB$.

\subsection{Definition of De Bruijn monads}
We start by fixing some terminology and notation, and then give the
definition.

\begin{defi}
  Given a set $X$, an \alert{$X$-assignment} is a map $ℕ → X$. We
  sometimes merely use ``assignment'' when $X$ is clear from context.
\end{defi}

\begin{nota}\label{not:subst} Consider any map $s∶ X × Y^ℕ → Z$.
  \begin{itemize}
  \item For all $x ∈ X$ and $g∶ ℕ → Y$, we write $x[g]ₛ$ for $s(x,g)$,
    or $x[g]_X$ when $s$ is clear from context, or even $x[g]$ when
    $s$ and $X$ are clear from context.
  \item Furthermore, $s$ gives rise to the map
    \[
      \begin{array}{rcl}
      X^ℕ × Y^ℕ & → & Z^ℕ \\
      (f,g) & ↦ & n ↦ s(f(n),g).
      \end{array}
    \]
    We use similar notation for this map, i.e., $f[g](n) ≔ f(n)[g]ₛ$.
\end{itemize}
\end{nota}


\begin{defi}\label{def:dbmonad}
  A \alert{De Bruijn monad} is a set $X$, equipped with
  \begin{itemize}
  \item a \alert{substitution} map $s∶ X × X^ℕ→ X$, which
    takes an element $x ∈ X$ and an assignment $f∶ ℕ → X$, and returns
    an element $x[f]$, and

  \item a \alert{variables} map $v∶ ℕ → X$,
  \end{itemize}
  satisfying, for all $x ∈ X$, and $f,g∶ ℕ → X$:
  \begin{itemize}
  \item \alert{associativity}: $x[f][g] = x[f[g]]$,
  \item \alert{left unitality}: $v(n)[f] = f(n)$, and
  \item \alert{right unitality}: $x[v] = x$.
  \end{itemize}
\end{defi}

  \begin{exa}\label{ex:nat}
    The set $ℕ$ itself is clearly a De Bruijn monad, with variables
    given by the identity and substitution $ℕ × ℕ^ℕ → ℕ$ given by
    evaluation. This is in fact the initial De Bruijn monad, as should
    be clear from the development below.
  \end{exa}

  \begin{exa}\label{ex:lam}
    The set $Λ ≔ μ X. ℕ + X+ X^2$ of $λ$-terms forms a De Bruijn monad
    with well-known structure, which we now recall for
    completeness. Elements of $Λ$ are generated by the following
    grammar, where $n$ ranges over $ℕ$.
    $$e \Coloneqq n ｜ λ(e) ｜ e\ e$$
    The variables map $ℕ → Λ$ sends any $n$ to itself, i.e., the leaf
    labelled $n$.  For substitution, we want it to satisfy the
    following mutually recursive equations:
    \begin{align*}
      v(n)[σ] & = σ(n) \\
      (e₁\ e₂)[σ] & = e₁[σ]\ e₂[σ] \\
      λ(e)[σ] & = λ(e[⇑σ]) \\[.5em]
      (⇑σ)(0) &= v(0) \\
      (⇑σ)(n+1) &= σ(n)[v ∘ \scc]\rlap{,}
    \end{align*}
    where $\scc∶ ℕ → ℕ$ denotes the successor map.  However, the very
    last recursive call to substitution is not clearly decreasing in
    any way, so we cannot take this as a definition.  Instead, we 
    take it as a specification, and prove that there exist unique
    substitution and lifting maps satisfying the above equations.

    For this, we use a standard technique, based on the observation
    that the problematic recursive call ($σ(n)[v ∘ \scc]$) does not
    involve a general assignment but the mere renaming $v ∘ \scc$.  We
    replace this recursive call with $σ(n) \{ \scc \}$, where
    ${-}\{{-}\}∶ Λ × ℕ^ℕ → Λ$ denotes a \alert{renaming} map, easily
    defined by recursion as follows:
    \begin{align*}
      v(n)\{ f \} &= v(f(n)) \\
      (e₁\ e₂)\{f\} & = e₁\{f\}\ e₂\{f\} \\
      λ(e)\{f\} & = λ(e\{ {↑}f \})\rlap{,} \\[.5em]
      \llap{where \hspace*{1cm}} ({↑}f)(0) &= 0 \\
      ({↑}f)(n+1) &= f(n)+1.      
    \end{align*}

    (Because $f$ is a mere renaming, the definition of $↑$ is not
    recursive.)

    It is then straightforward to prove that the original equations
    are (uniquely) satisfied.

    In~Example~\ref{ex:univlambda}, as an application of
    Theorem~\ref{thm:initiality:II}, we will characterise the obtained
    De Bruijn monad by a universal property. In fact, the set
    $Λ ≔ μ X. ℕ + X+ X^2$ has infinitely many De Bruijn monad
    structures, as many as there are binding arities with underlying
    endofunctors $X ↦ X$ and $X ↦ X²$, in the sense defined below.
    But only one of these structures models $λ$-calculus substitution.
  \end{exa}

\subsection{Lifting assignments}
In preparation for introducing the binding conditions, given a De
Bruijn monad $M$, we now define an operation called \alert{lifting} on
its set of assignments $ℕ → M$.  It is convenient to stress that only
part of the structure of a De Bruijn monad is needed for this
definition.

\begin{defi} \label{def:prime}
  Consider any set $M$, equipped
  with maps $s∶ M × M^ℕ → M$ and $v∶ ℕ → M$.
  For any assignment
  $σ∶ ℕ → M$, we define the assignment $⇑σ∶ ℕ → M$ by
  \[
  \begin{array}[t]{rcl}
           (⇑σ)(0) & = & v(0) \\
        (⇑σ)(n+1) & = & σ(n)[↑]\rlap{,}
    \end{array}
  \]
    \noindent  where $↑∶ ℕ → X$ maps any $n$ to $v(n+1)$. %
  \end{defi}
    \begin{rem}
      Both $⇑$ and $↑$ depend on $M$ and (part of) $(s,v)$. Here, and
      in other similar situations below, we abuse notation and omit such
      dependencies for readability.
    \end{rem}
Of course we may iterate lifting:
  \begin{defi}\label{def:deriv} Let $⇑^{0}A = A$, and
    $⇑^{n+1}A =\ ⇑(⇑^{n}A)$.
 \end{defi}

\subsection{Binding arities and binding conditions}
Our treatment of binding arities reflects the separation between the
first-order part of the arity, namely its length, which concerns the
syntax, and the binding information, namely the binding numbers, which
concerns the compatibility with substitution.

\begin{defi}\hfill
  \begin{itemize}
  \item A \alert{first-order arity} is a natural number.
  \item A \alert{binding arity} is a sequence $(n₁, …, nₚ)$ of natural
  numbers, i.e., an element of $ℕ^*$.
\item The first-order arity $|a|$ associated with a binding arity
  $a = (n₁,…,nₚ)$ is its length $p$.
\end{itemize}
\end{defi}

Let us now axiomatise what we call an operation of a given binding arity.

\begin{defi}
  Let $a = (n₁,…,nₚ)$ be any binding arity, $M$ be any set,
  $s∶ M × M^ℕ → M$, and $v∶ ℕ → M$ be any maps.  An operation \alert{of binding
    arity $a$} is a map $o∶ Mᵖ → M$ satisfying the following
  \alert{$a$-binding condition} w.r.t.\ $(s,v)$:
  \begin{equation}
    ∀ σ∶ ℕ → M,
    x₁,…,xₚ ∈ M, \quad
    o(x₁,…,xₚ)[σ] = o(x₁[⇑^{n₁}σ],…,xₚ[⇑^{nₚ}σ]).
    \label{eq:binding}
\end{equation}
\end{defi}
\begin{rem}
  Let us emphasise the dependency of this definition on $v$ and $s$ --
  which is hidden in the notation for substitution and lifting.
\end{rem}

\subsection{Binding signatures and algebras}
In this section, we recall the standard notions of first-order (resp.\
binding) signatures, and adapt the definition of algebras to our De
Bruijn context.  Let us first briefly recall the former.
\begin{defi}
  A \alert{first-order signature} consists of a set $O$ of
  \alert{operations}, equipped with an \alert{arity} map $\ar∶ O → ℕ$.
\end{defi}
\begin{defi}
  For any first-order signature $S ≔ (O,\ar)$, an \alert{$S$-algebra}
  is a set $X$, together with, for each operation $o ∈ O$, a map
  $o_X∶ X^{\ar(o)} → X$.
\end{defi}

Let us now generalise this to binding signatures.
  \begin{defi}\hfill
    \begin{itemize}
    \item A \alert{binding signature}~\cite{BindingSignatures} consists of a set
      $O$ of \alert{operations}, equipped with an \alert{arity} map
      $\ar∶ O → ℕ^*$.
      Intuitively, the arity of an operation specifies the number of bound
      variables in each argument.
    \item The first-order signature $|S|$ associated with a binding
      signature $S ≔ (O,\ar)$ is $|S| ≔ (O,|\ar|)$, where
      $|\ar|∶ O → ℕ$ maps any $o ∈ O$ to the length $|\ar(o)|$ of
      $\ar(o)$.
  \end{itemize}
  \end{defi}
  \begin{exa}\label{ex:siglambda}
    As we saw in Example~\ref{ex:lamapp}, the binding signature for
    $λ$-calculus has two operations, abstraction and application, of
    respective arities $(1)$ and $(0,0)$. The associated first-order
    signature has two operations of respective arities $1$ and $2$.
  \end{exa}

  Let us now present the notion of De Bruijn $S$-algebra:
  \begin{defi}
    For any binding signature $S ≔ (O, \ar)$, a \alert{De Bruijn
      $S$-algebra} is a De Bruijn monad $(X,s,v)$ equipped with an
    operation of binding arity $\ar(o)$, for all $o ∈ O$.
\end{defi}

In order to state our characterisation of the term model, we associate
to any binding signature an endofunctor on sets, as follows.


\begin{defi}\label{def:SigmaS}
  The endofunctor $Σ_S$ associated to a binding signature $S=(O,\ar)$
  is defined by $Σ_S(X) = ∑_{o ∈ O} X^{|\ar(o)|}$.
\end{defi}
\begin{rem}
  The induced endofunctor merely depends on the underlying first-order
  signature.  
\end{rem}
\begin{defi}
  For any binding signature $S= (O,\ar)$ and $Σ_S$-algebra
  $a∶ Σ_S(X) → X$, we call the composite
  $X^{|\ar(o)|} \xto{inₒ} Σ_S(X) \xto{a} X$ the \alert{interpretation}
  of $o$ in $X$.
\end{defi}


\begin{rem}
As is well known, for any binding signature, the initial $(ℕ + \Sigma_S)$-algebra is the desired syntax; it has as carrier
the least fixed point $μA.ℕ+Σ_{S}(A)$.
\end{rem}

The following theorem defines the term model of a binding signature.
\begin{thm}\label{thm:initiality:I}
  Consider any binding signature $S= (O,\ar)$, and let ${\DB}$ denote the
  initial $(ℕ+Σ_S)$-algebra, with structure maps $v∶ ℕ → {\DB}$ and
  $a∶   Σ_S({\DB}) → {\DB}$. Then,
  \begin{enumerati}
  \item \label{item:subst} There exists a unique map
    $s∶ {\DB} × {\DB}^ℕ → {\DB}$ such that
    \begin{itemize}
    \item  for all $n ∈ ℕ$ and $f∶ ℕ → {\DB}$, $s(v(n),f) = f(n)$, and
    \item the interpretation of each $o ∈ O$ in $\DB$ satisfies the
      $\ar(o)$-binding condition w.r.t.\ $(s,v)$.
    \end{itemize}
  \item This map turns $({\DB},v,s,a)$ into a De Bruijn $S$-algebra.
  \end{enumerati}
  %
\end{thm}
\begin{proof}
  We have proved the result in both HOL Light~\cite{DBHOL} and
  Coq~\cite{DBCoq}, see~§\ref{s:mechanised}.
\end{proof}
\begin{rem}
  Point~\ref{item:subst} may be viewed as an abstract form of
  recursive definition for substitution in the term
  model.
  The theorem thus allows us to construct the term model of a signature in
  two steps: first the underlying set, constructed as the inductive
  datatype $μZ.ℕ+Σ_S(Z)$, and then substitution, defined by the
  binding conditions viewed as recursive equations.
\end{rem}

\section{Initial-algebra semantics of binding signatures in De Bruijn
monads}\label{s:elementary}
In this section, for any binding signature $S$, we organise De Bruijn
$S$-algebras into a category, $S\DBAlg$, and prove that the term De
Bruijn $S$-algebra $\DB$ is initial therein.

\subsection{A category of De Bruijn monads}
Let us start by organising general De Bruijn monads into a category:
  \begin{defi}
    A morphism $(X,s,v) → (Y,t,w)$ between De Bruijn monads is a set-map
    $f∶ X → Y$ commuting with substitution and variables, in the sense that
    for all $x ∈ X$ and $σ∶ ℕ → X$ we have
    $f(x[σ]) = f(x)[f∘σ]$ and $f∘v = w$.
  \end{defi}
  \begin{rem}
    More explicitly, the first axiom says: $f(s(x,σ)) = t(f(x),f∘σ)$.
  \end{rem}

  \begin{nota}
    De Bruijn monads and morphisms between them form a category, which
    we denote by $𝐃𝐁𝐌𝐧𝐝$.
  \end{nota}

    \subsection{De Bruijn monads as relative monads and as monoids}\label{ss:relativemonads}
    In this subsection, we briefly mention an alternative presentation
    of De Bruijn monads for the categorically-minded reader, in terms
    of \alert{relative} monads.  Namely, we show that they are monads
    relative to the functor $1 → 𝐒𝐞𝐭$ picking $ℕ$.  Then, following
    Altenkirch et al.~\cite{DBLP:journals/corr/AltenkirchCU14}, we
    explain a companion presentation in terms of monoids in $𝐒𝐞𝐭$, for
    a suitable \alert{skew monoidal}
    structure~\cite{DBLP:journals/corr/AltenkirchCU14,Szlachanyi}.

    \begin{rem}
      Altenkirch et al.\ have similarly shown that Fiore, Plotkin, and
      Turi's approach may be understood in terms of monads relative to
      the canonical embedding from finite sets into sets (and hence
      also in terms of monoids in a corresponding monoidal category).
    \end{rem}
  
    Let us first briefly recall relative monads, which were introduced
    by Altenkirch et al.~\cite{DBLP:journals/corr/AltenkirchCU14}.
    \begin{defi}
      For any set $𝐄$, category $𝐂$, and map $J∶ 𝐄 → 𝐨𝐛(𝐂)$, a \alert{$J$-relative monad}, or
      \alert{monad relative to $J$}, consists of
      \begin{itemize}
      \item an object mapping $T∶ 𝐄 → 𝐨𝐛(𝐂)$, together with
      \item \alert{unit} morphisms $η_X∶ J(X) → T(X)$, for all $X ∈ 𝐄$, and
      \item for each morphism $f∶ J(X) → T(Y)$, an \alert{extension}
        $f^†∶ T(X) → T(Y)$,
      \end{itemize}
      such that the following diagrams commute for all $X,Y,Z ∈ 𝐄$,
      $f∶ J(X) → T(Y)$, and $g∶ J(Y) → T(Z)$.
      \begin{mathpar}
\begin{tikzpicture}[every node/.style={inner sep=5pt,outer sep=0pt,anchor=base,text height=1.2ex, text depth=0.25ex}] 
\node (0) at (6em, -3.5773809523809526em) {$J(X)$} ; 
\node (1) at (10.761904761904763em, -3.5773809523809526em) {$T(X)$} ; 
\node (2) at (8.380952380952381em, -5.958333333333333em) {$T(Y)$} ; 
\path 
(0) to[->, ] node[coordinate](3){} (1) 
(1) to[->, ] node[coordinate](4){} (2) 
(0) to[->, ] node[coordinate](5){} (2) 
; 
\path[->] 
(0) edge["$\scriptstyle \eta_X$", pos=0.5, ->, ] (1) 
(1) edge["$\scriptstyle f^\dagger$", pos=0.5, ->, ] (2) 
(0) edge["$\scriptstyle f$"', pos=0.5, ->, ] (2) 
; 
\end{tikzpicture}

\and
\begin{tikzpicture}[every node/.style={inner sep=5pt,outer sep=0pt,anchor=base,text height=1.2ex, text depth=0.25ex}] 
\node (0) at (5.9523809523809526em, -3.5714285714285716em) {$T(X)$} ; 
\node (1) at (13.095238095238095em, -3.5714285714285716em) {$T(X)$} ; 
\path 
(0) to[->, bend right={-17.1887338539247}, ] node[coordinate](2){} (1) 
(0) to[cell=0, -,bend right={11.459155902616464}, ] node[coordinate](3){} (1) 
; 
\path[->] 
(0) edge["$\scriptstyle \eta_X^\dagger$", pos=0.5, bend right={-17.1887338539247}, ] (1) 
(0) edge["$\scriptstyle $", pos=0.5, cell=0, -,bend right={11.459155902616464}, ] (1) 
; 
\end{tikzpicture}


\begin{tikzpicture}[every node/.style={inner sep=5pt,outer sep=0pt,anchor=base,text height=1.2ex, text depth=0.25ex}] 
\node (0) at (5.523809523809524em, -3.1964285714285716em) {$T(X)$} ; 
\node (1) at (10.285714285714286em, -3.1964285714285716em) {$T(Y)$} ; 
\node (2) at (7.904761904761905em, -5.5773809523809526em) {$T(Z)$} ; 
\path 
(0) to[->, ] node[coordinate](3){} (1) 
(1) to[->, ] node[coordinate](4){} (2) 
(0) to[->, ] node[coordinate](5){} (2) 
; 
\path[->] 
(0) edge["$\scriptstyle f^\dagger$", pos=0.5, ->, ] (1) 
(1) edge["$\scriptstyle g^\dagger$", pos=0.5, ->, ] (2) 
(0) edge["$\scriptstyle (g^\dagger \circ f)^\dagger$"', pos=0.5, ->, ] (2) 
; 
\end{tikzpicture}

\end{mathpar}

A morphism $(T,η,(-)^†) → (T',η',(-)^{†'})$ of $J$-relative monads
consists of morphisms $α_X∶ T(X) → T'(X)$, for all $X ∈ 𝐄$,
making
the following diagrams commute
\begin{center}
\begin{tikzpicture}[every node/.style={inner sep=5pt,outer sep=0pt,anchor=base,text height=1.2ex, text depth=0.25ex}] 
\node (0) at (13.095238095238095em, -1.1904761904761905em) {$J(X)$} ; 
\node (1) at (10.714285714285714em, -3.5714285714285716em) {$T(X)$} ; 
\node (2) at (15.476190476190476em, -3.5714285714285716em) {$T'(X)$} ; 
\path 
(0) to[->, ] node[coordinate](3){} (1) 
(1) to[->, ] node[coordinate](4){} (2) 
(0) to[->, ] node[coordinate](5){} (2) 
; 
\path[->] 
(0) edge["$\scriptstyle \eta_X$"', pos=0.5, ->, ] (1) 
(1) edge["$\scriptstyle \alpha_X$"', pos=0.5, ->, ] (2) 
(0) edge["$\scriptstyle \eta'_X$", pos=0.5, ->, ] (2) 
; 
\end{tikzpicture}
\hfil
\begin{tikzpicture}[every node/.style={inner sep=5pt,outer sep=0pt,anchor=base,text height=1.2ex, text depth=0.25ex}] 
\node (0) at (5.9523809523809526em, -3.5714285714285716em) {$T(X)$} ; 
\node (1) at (13.095238095238095em, -3.5714285714285716em) {$T(Y)$} ; 
\node (2) at (13.095238095238095em, -5.9523809523809526em) {$T'(Y)$} ; 
\node (3) at (5.9523809523809526em, -5.9523809523809526em) {$T'(X)$} ; 
\path 
(0) to[->, ] node[coordinate](4){} (1) 
(1) to[->, ] node[coordinate](5){} (2) 
(0) to[->, ] node[coordinate](6){} (3) 
(3) to[->, ] node[coordinate](7){} (2) 
; 
\path[->] 
(0) edge["$\scriptstyle f^\dagger$", pos=0.5, ->, ] (1) 
(1) edge["$\scriptstyle \alpha_Y$", pos=0.5, ->, ] (2) 
(0) edge["$\scriptstyle \alpha_X$"', pos=0.5, ->, ] (3) 
(3) edge["$\scriptstyle (\alpha_Y \circ f)^{\dagger'}$"', pos=0.5, ->, ] (2) 
; 
\end{tikzpicture}
\end{center}
 for all $X,Y ∈ 𝐄$, and $f∶ J(X) → T(X)$.
 Monads relative to $J$ and morphisms between them form a category.
\end{defi}
    \begin{rem}
      This definition is slightly different from, but equivalent to
      the original.
    \end{rem}
    
\begin{prop}
  The category $𝐃𝐁𝐌𝐧𝐝$ is canonically isomorphic to the category of
  monads relative to the map $1 → 𝐒𝐞𝐭$ picking $ℕ$.
\end{prop}
\begin{rem}
  Canonicity here means that the isomorphism lies over the canonical
  isomorphism $[1,𝐒𝐞𝐭] ≅ 𝐒𝐞𝐭$.
\end{rem}
\begin{proof}
  By mere definition unfolding:
  \begin{itemize}
  \item An object mapping $T∶ 1 → 𝐒𝐞𝐭$ amounts to a choice of object $X$ in $𝐒𝐞𝐭$.
  \item A unit $η∶ ℕ → X$ amounts to a choice of variables map.
  \item The assignment of an extension $f^†∶ X → X$ to each $f∶ ℕ → X$
    amounts to a map $X^ℕ → X^X$, which is equivalent by uncurrying to
    a choice of substitution map $X × X^ℕ → X$. Notationally, $f^†(x)$
    thus corresponds to $x[f]$.
  \end{itemize}
  We then check that the axioms match:
  \begin{itemize}
  \item Right unitality $x[η] = x$ corresponds to $η^†(x) = x$, i.e., $η^† = \id_X$.
  \item Left unitality $η(n)[f] = f(n)$ corresponds to $f^†(η(n)) = f(n)$, i.e.,
    $f^† ∘ η = f$.
  \item For associativity $x[f][g] = x[f[g]]$, by definition
    $f[g](n) = f(n)[g]$, so $f[g]$ corresponds to $n ↦ g^†(f(n))$,
    i.e., $g^† ∘ f$, and the axiom becomes $g^† ∘ f^† = (g^† ∘ f)^†$,
    as desired. \qedhere
  \end{itemize}
\end{proof}

Following Altenkirch et al., let us now give a further alternative
characterisation of the category $𝐃𝐁𝐌𝐧𝐝$ of De Bruijn monads in terms
of skew monoidal categories, which we now recall.
\begin{defi}[{Szlachányi~\cite{Szlachanyi}}]
  A \alert{skew monoidal category} is a category $𝐂$ equipped with a \alert{tensor product} functor
  ${⊗}∶ 𝐂² → 𝐂$, written in infix notation, and a \alert{unit} object $I ∈ 𝐂$, together with
  \begin{itemize}
  \item an \alert{associator} natural transformation
    $α_{X,Y,Z}∶ (X ⊗ Y) ⊗ Z → X ⊗ (Y ⊗ Z)$,
  \item a \alert{right unitor} natural transformation
    $ρ_X∶ X → X ⊗ I$, and
  \item a \alert{left unitor} natural transformation
    $λ_X∶ I ⊗ X → X$,
  \end{itemize}
  satisfying the following coherence conditions.
  \begin{mathpar}
\begin{tikzpicture}[every node/.style={inner sep=5pt,outer sep=0pt,anchor=base,text height=1.2ex, text depth=0.25ex}] 
\node (0) at (5.9523809523809526em, -3.5714285714285716em) {$I$} ; 
\node (1) at (10.714285714285714em, -1.1904761904761905em) {$I \otimes I$} ; 
\node (2) at (15.476190476190476em, -3.5714285714285716em) {$I$} ; 
\path 
(0) to[->, ] node[coordinate](3){} (1) 
(1) to[->, ] node[coordinate](4){} (2) 
(0) to[cell=0, -,] node[coordinate](5){} (2) 
; 
\path[->] 
(0) edge["$\scriptstyle \rho_I$", pos=0.5, ->, ] (1) 
(1) edge["$\scriptstyle \lambda_I$", pos=0.5, ->, ] (2) 
(0) edge["$\scriptstyle $", pos=0.5, cell=0, -,] (2) 
; 
\end{tikzpicture}
\and 
\begin{tikzpicture}[every node/.style={inner sep=5pt,outer sep=0pt,anchor=base,text height=1.2ex, text depth=0.25ex}] 
\node (0) at (8.333333333333334em, -8.333333333333334em) {$X \otimes Y$} ; 
\node (1) at (11.666666666666666em, -5em) {$(X \otimes I) \otimes Y$} ; 
\node (2) at (21.666666666666668em, -5em) {$X \otimes (I \otimes Y)$} ; 
\node (3) at (25em, -8.333333333333334em) {$X \otimes Y$} ; 
\path 
(0) to[->, ] node[coordinate](4){} (1) 
(1) to[->, ] node[coordinate](5){} (2) 
(2) to[->, ] node[coordinate](6){} (3) 
(0) to[cell=0, -,] node[coordinate](7){} (3) 
; 
\path[->] 
(0) edge["$\scriptstyle \rho_X \otimes Y$", pos=0.30000000000000004, ->, ] (1) 
(1) edge["$\scriptstyle \alpha_{X,I,Y}$", pos=0.5, ->, ] (2) 
(2) edge["$\scriptstyle X \otimes \lambda_Y$", pos=0.7, ->, ] (3) 
(0) edge["$\scriptstyle $", pos=0.5, cell=0, -,] (3) 
; 
\end{tikzpicture}
\and 
\begin{tikzpicture}[every node/.style={inner sep=5pt,outer sep=0pt,anchor=base,text height=1.2ex, text depth=0.25ex}] 
\node (0) at (5.9523809523809526em, -8.333333333333334em) {$((A \otimes B) \otimes C) \otimes D$} ; 
\node (1) at (5.9523809523809526em, -10.714285714285714em) {$(A \otimes B) \otimes (C \otimes D)$} ; 
\node (2) at (13.095238095238095em, -13.095238095238095em) {$A \otimes (B \otimes (C \otimes D))$} ; 
\node (3) at (20.238095238095237em, -8.333333333333334em) {$(A \otimes (B \otimes C)) \otimes D$} ; 
\node (4) at (20.238095238095237em, -10.714285714285714em) {$A \otimes ((B \otimes C) \otimes D)$} ; 
\path 
(0) to[->, ] node[coordinate](5){} (1) 
(1) to[->, ] node[coordinate](6){} (2) 
(0) to[->, ] node[coordinate](7){} (3) 
(3) to[->, ] node[coordinate](8){} (4) 
(4) to[->, ] node[coordinate](9){} (2) 
; 
\path[->] 
(0) edge["$\scriptstyle \alpha_{A\otimes B, C, D}$"', pos=0.5, ->, ] (1) 
(1) edge["$\scriptstyle \alpha_{A,B,C \otimes D}$"', pos=0.10000000000000003, ->, ] (2) 
(0) edge["$\scriptstyle \alpha_{A,B,C} \otimes D$", pos=0.5, ->, ] (3) 
(3) edge["$\scriptstyle \alpha_{A, B \otimes C, D}$", pos=0.5, ->, ] (4) 
(4) edge["$\scriptstyle A \otimes \alpha_{B,C,D}$", pos=0.10000000000000003, ->, ] (2) 
; 
\end{tikzpicture}
\and
\begin{tikzpicture}[every node/.style={inner sep=5pt,outer sep=0pt,anchor=base,text height=1.2ex, text depth=0.25ex}] 
\node (0) at (8.333333333333334em, -5.9523809523809526em) {$(I \otimes X) \otimes Y$} ; 
\node (1) at (17.857142857142858em, -5.9523809523809526em) {$I \otimes (X \otimes Y)$} ; 
\node (2) at (13.095238095238095em, -8.333333333333334em) {$X \otimes Y$} ; 
\path 
(0) to[->, ] node[coordinate](3){} (1) 
(1) to[->, ] node[coordinate](4){} (2) 
(0) to[->, ] node[coordinate](5){} (2) 
; 
\path[->] 
(0) edge["$\scriptstyle \alpha_{I,X,Y}$", pos=0.5, ->, ] (1) 
(1) edge["$\scriptstyle \lambda_{X \otimes Y}$", pos=0.5, ->, ] (2) 
(0) edge["$\scriptstyle \lambda_X \otimes Y$"', pos=0.5, ->, ] (2) 
; 
\end{tikzpicture}
\and 
\begin{tikzpicture}[every node/.style={inner sep=5pt,outer sep=0pt,anchor=base,text height=1.2ex, text depth=0.25ex}] 
\node (0) at (10.714285714285714em, -3.5714285714285716em) {$X \otimes Y$} ; 
\node (1) at (5.9523809523809526em, -5.9523809523809526em) {$(X \otimes Y)  \otimes I$} ; 
\node (2) at (15.476190476190476em, -5.9523809523809526em) {$X \otimes (Y  \otimes I)$} ; 
\path 
(0) to[->, ] node[coordinate](3){} (1) 
(1) to[->, ] node[coordinate](4){} (2) 
(0) to[->, ] node[coordinate](5){} (2) 
; 
\path[->] 
(0) edge["$\scriptstyle \rho_{X\otimes Y}$"', pos=0.5, ->, ] (1) 
(1) edge["$\scriptstyle \alpha_{X,Y,I}$"', pos=0.5, ->, ] (2) 
(0) edge["$\scriptstyle X \otimes \rho_Y$", pos=0.5, ->, ] (2) 
; 
\end{tikzpicture}
  \end{mathpar}
\end{defi}

We will now show that De Bruijn monads are equivalently monoids for
some suitable skew monoidal category structure on $𝐒𝐞𝐭$.  For this,
we introduce the following terminology.
\begin{nota}\label{not:Jrel}
  For a functor $J∶ 𝐄 → 𝐂$, ``$J$-relative monad'' means monad
  relative to the object mapping $𝐨𝐛(𝐄) → 𝐨𝐛(𝐂)$ of $J$.
\end{nota}

\begin{prop}
  \label{prop:skew}
  For any small category $𝐄$, cocomplete category $𝐂$, and functor $J∶ 𝐄 → 𝐂$,
  the functor category $[𝐄,𝐂]$ is skew monoidal, with tensor
  $G ⊗ F = ∑_J(G) ∘ F$ and unit $I = J$, where $∑_J$ denotes left Kan
  extension along $J$.  Furthermore, monoids in $[𝐄,𝐂]$ are in
  one-to-one correspondence with monads relative to $J$.
\end{prop}
\begin{proof}
  By~\cite[Theorems~4 and~5]{DBLP:journals/corr/AltenkirchCU14}.
\end{proof}

Applying this to the functor $J∶ 1 → 𝐒𝐞𝐭$ picking $ℕ$ and transferring
across the isomorphism $[1,𝐒𝐞𝐭] ≅ 𝐒𝐞𝐭$, we obtain a skew monoidal
structure on sets, and Proposition~\ref{prop:skew} gives:
\begin{cor}\label{cor:skew}
  The tensor product $X⊗Y ≔ X \times Y^ℕ$ extends to a skew
  monoidal structure on $𝐒𝐞𝐭$, with:
  \begin{itemize}
  \item unit $ℕ$,
  \item right unitor $ρ_X∶ X → X ⊗ ℕ$ given by
    $ρ_X(x) = (x, \id_ℕ)$,
  \item left unitor $λ_X∶ ℕ ⊗ X → X$ given by
    evaluation $λ_X(n,σ) = σ(n)$, and
  \item associator $α_{X,Y,Z}∶ (X ⊗ Y) ⊗ Z → X ⊗ (Y ⊗ Z)$ given by
    $α_{X,Y,Z} ((x,υ),ζ) = (x, (n↦ (υ(n),ζ)))$.
  \end{itemize}

  Furthermore, $𝐃𝐁𝐌𝐧𝐝$ is precisely the category of monoids therein.
\end{cor}
\begin{proof}
  By the standard formula for left Kan extension, we have
\begin{align*}
  ∑_J(X)(U) & ≅ ∫^{⋆ ∈ 1} 𝐒𝐞𝐭(J(⋆),U) × X(⋆) \\
            & ≅ 𝐒𝐞𝐭(ℕ,U) × X(⋆) \\
            & ≅ U^ℕ × X(⋆). \qedhere
\end{align*}
\end{proof}


\begin{rem}\label{rem:BNNmonoid}
  By Notation~\ref{not:Jrel}, if two functors $J∶ 𝐄 → 𝐂$ and
  $J'∶ 𝐄' → 𝐂$ have the same object mapping up to isomorphism (hence
  in particular $𝐨𝐛(𝐄) ≅ 𝐨𝐛(𝐄')$), then $J$-relative monads are
  isomorphic to $J'$-relative monads, and both are isomorphic to
  monoids in $[𝐄,𝐂]$, resp.\ in $[𝐄',𝐂]$ (under the assumptions of
  Proposition~\ref{prop:skew}).
  
  In particular, the functor $1 → 𝐒𝐞𝐭$ picking $ℕ$ factors as
  $$1 \xto{I} 𝐁[ℕ,ℕ] \xto{K} 𝐒𝐞𝐭\rlap{,}$$
  where $𝐁[ℕ,ℕ]$ denotes the full subcategory spanned by $ℕ$. Since
  the object mapping of $K$ is the same as that of $1 → 𝐒𝐞𝐭$, De
  Bruijn monads are equivalently monoids in the category
  $[𝐁[ℕ,ℕ], 𝐒𝐞𝐭]$.  Remarkably, unlike $𝐒𝐞𝐭$ with the skew monoidal
  structure of Corollary~\ref{cor:skew}, $[𝐁[ℕ,ℕ], 𝐒𝐞𝐭]$ is in fact
  monoidal.
\end{rem}

\subsection{Categories of De Bruijn algebras}
In this section, for any binding signature $S$,
we organise De Bruijn $S$-algebras into a category $S\DBAlg$.

Let us start by recalling the category of $S$-algebras for a
first-order $S$:
\begin{defi}
  For any first-order signature $S$, a morphism $X → Y$ of
  $S$-algebras is a map between underlying sets commuting with
  operations, in the sense that for each $o ∈ O$, letting
  $p ≔ \ar(o)$, we have $f(o_X(x₁,…,xₚ)) = o_Y(f(x₁),…,f(xₚ))$.

  We denote by  $S\alg$ the category of $S$-algebras and morphisms
  between them.
\end{defi}

We now exploit this to define morphisms between De Bruijn
$S$-algebras:
\begin{defi}
  For any binding signature $S$, a morphism of De Bruijn $S$-algebras
  is a map $f∶ X → Y$ between underlying sets, which is a morphism
  both of De Bruijn monads and of $|S|$-algebras.  We denote by
  $S\DBAlg$ the category of De Bruijn $S$-algebras and morphisms
  between them.
\end{defi}

\begin{thm}\label{thm:initiality:II}
  Consider any binding signature $S= (O,\ar)$, and let ${\DB}$ denote the
  initial $(ℕ+Σ_S)$-algebra. Then, the De Bruijn $S$-algebra structure
  of Theorem~\ref{thm:initiality:I} on ${\DB}$ makes it initial in
  $S\DBAlg$.
\end{thm}
  \begin{proof}
  We have proved the result in both HOL Light~\cite{DBHOL} and
  Coq~\cite{DBCoq}, see~§\ref{s:mechanised}.
  \end{proof}

  \begin{exa}\label{ex:univlambda}
    For the binding signature of $λ$-calculus
    (Example~\ref{ex:siglambda}), the carrier of the initial model is
    $μ Z. ℕ + Z + Z^2$. 
  \end{exa}

\section{Relation to presheaf-based models}\label{s:presheaves}
The classical initial-algebra semantics introduced in
\cite{fiore:presheaf,DBLP:conf/lics/Fiore08} associates in particular
to each binding signature $S$ a category, say $Φ_S\Mon$ of models,
while we have proposed in~§\ref{s:elementary} an alternative category
of models $S\DBAlg$.  In this section, we are interested in comparing
both categories of models.

In fact, we find that both may include pathological models, in the sense that
we do not see any loss in ruling them out.  And when we do so, we
obtain equivalent categories.

  \subsection{Trimming down presheaf-based models}\label{ss:wbfpt}
  First of all, in this subsection, let us recall the mainstream
  approach we want to relate to, and exclude some pathological objects from
  it.  

  \subsubsection{Presheaf-based models}
  We start by recalling the presheaf-based approach.  The ambient
  category is the category of functors $[𝔽,𝐒𝐞𝐭]$, where $𝔽$ denotes
  the category of finite ordinals, and all maps between them.

  \begin{defi}
    Let $[𝐒𝐞𝐭,𝐒𝐞𝐭]_f$ denote the full subcategory of $[𝐒𝐞𝐭,𝐒𝐞𝐭]_f$
    spanning \alert{finitary} functors, i.e., those preserving
    directed colimits ($=$ colimits of directed posets).
  \end{defi}
  \begin{prop}\label{prop:finsets}
    The restriction functor $[𝐒𝐞𝐭,𝐒𝐞𝐭]_f → [𝔽,𝐒𝐞𝐭]$ is an equivalence.
  \end{prop}
  \begin{proof}
    The category of sets is $ω$-accessible, so by~\cite[Theorem~2.26,
    (i) $⇔$ (ii)]{Adamek} and~\cite[Remark~2.26(1)]{Adamek}, it is a
    free cocompletion of its full subcategory of finitely presentable
    objects under directed colimits.  Equivalently, it is a free
    cocompletion of $𝔽$ under directed colimits.  Thus, by taking
    $\mathcal{B}$ to be $𝐒𝐞𝐭$ in \cite[Definition~2.25]{Adamek}, we
    obtain that the restriction functor $[𝐒𝐞𝐭,𝐒𝐞𝐭]_f → [𝔽, 𝐒𝐞𝐭]$ is an
    equivalence.
  \end{proof}

  \begin{defi}\label{def:monFSet}
    Let $(⊗,I)$ denote the monoidal structure on $[𝔽,𝐒𝐞𝐭]$ inherited
    from the composition monoidal structure on $[𝐒𝐞𝐭,𝐒𝐞𝐭]_f$ through
    the equivalence of Proposition~\ref{prop:finsets}.
  \end{defi}

  By construction, monoids in $[𝔽,𝐒𝐞𝐭]$ are thus equivalent to
  finitary monads on sets.

  The idea is then to interpret binding signatures $S$ as endofunctors
  $Φ_S$ on $[𝔽,𝐒𝐞𝐭]$, and to define models as monoids equipped with
  $Φ_S$-algebra structure, satisfying a suitable compatibility
  condition.

  The definition of $Φ_S$ relies on an operation called derivation:
  \begin{defi}[Endofunctor associated to a binding signature]\label{def:PhiS}\hfill
    \begin{itemize}
    \item Let the \alert{derivative} $X'$ of any functor $X∶ 𝔽 → 𝐒𝐞𝐭$
      be defined by $X'(n) = X(n+1)$.
  \item Furthermore, let $X^{(0)} = X$, and $X^{(n+1)} = (X^{(n)})'$.
    \item For any binding arity $a = (n₁,…,nₚ)$, let
      $Φₐ(X) = X^{(n₁)}×…×X^{(nₚ)}$.
    \item For any binding signature $S = (O,\ar)$, let
      $Φ_S = ∑_{o ∈ O} Φ_{\ar(o)}$.
  \end{itemize}
\end{defi}
  \begin{prop}
    Through the equivalence with finitary functors, derivation becomes
    $F'(A) = F(A+1)$, for any finitary $F∶ 𝐒𝐞𝐭 → 𝐒𝐞𝐭$ and $A ∈ 𝐒𝐞𝐭$.
  \end{prop}
\begin{exa}
  On the binding signature for $λ$-calculus, say $S_λ$, which we saw
  in Example~\ref{ex:siglambda}, we get
  \[ Φ_{S_λ}(X)(n) = X(n)² + X(n+1).\]
\end{exa}

Next, we want to express the relevant compatibility condition between
algebra and monoid structure. For this, let us briefly recall the
notion of pointed strength, see
\cite{fiore:presheaf,DBLP:conf/lics/Fiore08} for details.
  \begin{defi}
    A \alert{pointed strength} on an endofunctor $F∶ 𝐂 → 𝐂$ on a
    monoidal category $(𝐂,⊗,I,α,λ,ρ)$ is a family of morphisms
    $st_{C,(D,v)}∶ F(C)⊗D → F (C⊗D)$, natural in $C ∈ 𝐂$ and
    $(D,v∶ I → D) ∈ I/𝐂$, the coslice category under the tensor unit
    $I$, making the following diagrams commute,
      \begin{center}
    \diag{%
      \& F(A) \\
      F(A)⊗I \& \& F (A⊗I) %
    }{%
      (m-1-2) edge[labelal={ρ_{F(A)}}] (m-2-1) %
      edge[labelar={F(ρ_A)}] (m-2-3) %
      (m-2-1) edge[labelb={st_{A,(I,\id)}}] (m-2-3) %
    }
    \diag(0.6,2.2){
      (F(A)⊗X)⊗Y \& F(A⊗X)⊗Y \& F ((A⊗X)⊗Y) \\
      F(A)⊗(X⊗Y) \& \& F (A⊗(X⊗Y))
    }{%
      (m-1-1) edge[labela={\scalebox{1.03}{$st_{A,(X,v_X)}⊗Y$}}] (m-1-2) %
      edge[labell={\scalebox{1.03}{$α_{F(A),X,Y}$}}] (m-2-1) %
      (m-1-2) edge[labela={\scalebox{1.03}{$st_{A⊗X,(Y,v_Y)}$}}] (m-1-3) %
      (m-1-3) edge[labelr={\scalebox{1.03}{$F (α_{A,X,Y})$}}] (m-2-3) %
      (m-2-1) edge[labelb={\scalebox{1.03}{$st_{A,(X⊗Y,v_{X⊗Y})}$}}] (m-2-3) %
    }
  \end{center}
  for all objects $A$, $X$, $Y$, and morphisms $v_X∶ I → X$ and
  $v_Y∶ I → Y$, where $v_{X⊗Y}$ denotes the composite
  $$I \xto{ρ_I} I ⊗ I \xto{v_X ⊗ v_Y} X ⊗ Y.$$
\end{defi}
The next step is to observe that binding signatures generate pointed
strong endofunctors.
\begin{defi}
  The derivation endofunctor $X ↦ X'$ on $[𝔽,𝐒𝐞𝐭]$ has a pointed
  strength, defined through the equivalence with finitary functors by
  $$G (F(X)+1) \xto{G (F(X)+v₁)}G (F (X) + F (1)) \xto{G[F(in₁),F(in₂)]}
  G (F (X+1)).$$
\end{defi}
Product, coproduct, and composition of endofunctors lift
to pointed strong endofunctors, which yields:
\begin{corC}[{\cite{fiore:presheaf,DBLP:conf/lics/Fiore08}}]
  For all binding signatures $S$, $Φ_S$ is canonically pointed strong.
\end{corC}

At last, we arrive at the definition of models.
\begin{defi}
  \label{def:F-monoid}
  For any pointed strong endofunctor $F$ on a monoidal category\linebreak{}$(𝐂,⊗,I,α,λ,ρ)$, an \alert{$F$-monoid}
  is an object $X$ equipped with $F$-algebra and monoid structure, say
  $a∶ F(X) → X$, $s∶ X⊗X→X$, and $v∶ I → X$, such that the following
  pentagon commutes.
  \begin{center}
    \diag{%
      F(X)⊗X \& F (X⊗X) \& F(X) \\
      X⊗X \& \& X %
    }{%
      (m-1-1) edge[labela={st_{X,(X,v)}}] (m-1-2) %
      edge[labell={a⊗X}] (m-2-1) %
      (m-1-2) edge[labela={F(s)}] (m-1-3) %
      (m-2-1) edge[labelb={s}] (m-2-3) %
      (m-1-3) edge[labelr={a}] (m-2-3) %
    }
  \end{center}
  A morphism of $F$-monoids is a morphism in $𝐂$ which is a morphism
  both of $F$-algebras and of monoids.  We let $F\Mon$ denote the
  category of $F$-monoids and morphisms between them.
\end{defi}

\begin{exa}
  For the binding signature $S_λ$ of Example~\ref{ex:siglambda}, a
  $Φ_{S_λ}$-monoid is an object $X$, equipped with maps $X' → X$ and
  $X²→X$, and with compatible monoid structure.  Compatibility describes
  how substitution should be pushed down through abstractions and
  applications.
\end{exa}

\subsubsection{Intersectional presheaves}
  The pathology we want to rule out only concerns the underlying
  functor of a model, so we just have to define well-behaved functors in
  $[𝔽,𝐒𝐞𝐭]$.

  Well-behavedness for a functor $T∶ 𝔽 → 𝐒𝐞𝐭$ is about getting closed
  terms right.  More precisely, for some finite sets $m$ and $n$, an
  element of $T(m+n)$ which both exists in $T(m)$ and $T(n)$ should
  also exist in $T(∅)$, and uniquely so. This says exactly that $T$
  should preserve the pullback
  \begin{equation}
    \label{eq:pbk-intersection}
    \Diag{%
      \stdpbk %
    }{%
      ∅ \& n \\
      m \& m+n\rlap{.} %
    }{%
      (m-1-1) edge[labela={}] (m-1-2) %
      edge[labell={}] (m-2-1) %
      (m-2-1) edge[labelb={}] (m-2-2) %
      (m-1-2) edge[labelr={}] (m-2-2) %
    }
    \end{equation}
  Taking \alert{intersection} to mean pullback of two monomorphisms,
  the following known result shows that all non-empty intersections
  are automatically preserved.
  \begin{propC}[{\cite[Proposition~2.2]{PresentationOfSetFunctors}}]
    \label{prop:trnkova}
    All endofunctors of sets preserve non-empty intersections.
  \end{propC}
  Thus, by Proposition~\ref{prop:finsets}, all functors $𝔽 → 𝐒𝐞𝐭$
  preserve non-empty intersections, and we have:
  \begin{cor}
    A functor $𝔽→𝐒𝐞𝐭$ preserves (binary) intersections iff it
    preserves empty (binary) intersections.
  \end{cor}
  \begin{lem}
    A functor $T$ from $𝐒𝐞𝐭$ (or $𝔽$) to $𝐒𝐞𝐭$ preserves empty binary intersections if and only
    if it preserves the following pullback.
    \begin{center}
      \Diag{%
        \stdpbk %
      }{%
        0 \& 1 \\
        1 \& 2
      }{%
        (m-1-1) edge[labela={}] (m-1-2) %
        edge[labell={}] (m-2-1) %
        (m-2-1) edge[labelb={0}] (m-2-2) %
        (m-1-2) edge[labelr={1}] (m-2-2) %
      }%
    \end{center}
  \end{lem}
  \begin{proof}
    By Proposition~\ref{prop:finsets}, it is enough to reason on an endofunctor
    $T$ on $𝐒𝐞𝐭$.
    If $T$ preserves empty binary intersections, then it preserves the above
    pullback as a particular case.
    Conversely, assume that it preserves the above pullback. Then, the following
    diagram is an equaliser.
    \[
      \begin{tikzcd}
        T(0)\ar[r] &
        T(1) \ar[r,shift left=.75ex,"T(0)"]
        \ar[r,shift right=.75ex,swap,"T(1)"]
        &
        T(2).
      \end{tikzcd}
    \]
    Therefore, $T$ coincides with its so-called Trnková closure and
    thus by~\cite[Corollary~VII.2]{CoproductsMonadsSet}, it preserves
    finite intersections.
  \end{proof}
  \begin{defi}\hfill
    \label{d:wb-FSet}
    \begin{itemize}
    \item A functor $𝔽 → 𝐒𝐞𝐭$ is \alert{intersectional} iff it preserves
      binary intersections, or equivalently empty binary
      intersections.  Let $[𝔽,𝐒𝐞𝐭]_{\its}$ denote the full subcategory
      spanned by intersectional functors.
    \item A monoid in $[𝔽,𝐒𝐞𝐭]$, (resp., for any binding signature
      $S$, an object of $Φ_S\Mon$) is \alert{intersectional} iff the
      underlying functor is.  Let $Φ_S\Mon_{\its}$ denote the full
      subcategory spanned by intersectional objects.
  \end{itemize}
  \begin{exa}
    As an example of a non intersectional finitary monad, first consider the monad $L$
    of $λ$-calculus, so that $L(X)$ is set of $λ$-terms taking free variables
    in $X$. This monad is intersectional, but now consider the monad $L'$ agreeing
    with $L$ on any non-empty set, and such that $L'(∅) = ∅$. Then, $L'$ is not intersectional.
  \end{exa}
  \end{defi}

  The important result for comparing the presheaf-based approach with
  ours is the following.
  \begin{prop}
    The subcategory $Φ_S\Mon_{\its}$ includes the initial object.
  \end{prop}
  \begin{proof}
    Roughly, closed terms are isomorphic to terms in two free
    variables that use neither the first, nor the second.
  \end{proof}

  Let us conclude this subsection with the following observation, that
  for a wide class of signatures all models are in fact well behaved.
  \begin{prop}
    If the initial object ${\DB'}$ of $Φ_S\Mon$ has at least one closed term
    (i.e., ${\DB'}(∅) ≠ ∅$), then $Φ_S\Mon_{\its} = Φ_S\Mon$.
  \end{prop}
  \begin{proof}
    If $T$ is a $Φ_S$-monoid, then by initiality there is a morphism
    ${\DB'} → T$, and in particular a map ${\DB'}(∅) → T(∅)$. Since ${\DB'}(∅)$ is
    non-empty by assumption, $T(∅)$ cannot be empty. The result then
    follows from~\cite[Proposition~VII.7]{CoproductsMonadsSet}: a
    monad $T$ on $𝐒𝐞𝐭$ either preserves the initial object, or is
    intersectional.
  \end{proof}
  \begin{rem}
    The binding signatures for which the initial model has at least
    one closed term are those specifying at least a constant or an
    operation binding (at least) one variable in each argument.
  \end{rem}
  \subsection{Trimming down De Bruijn monads}\label{ss:wbdb}
  Let us now turn to well-behaved De Bruijn algebras.  Here
  well-behavedness is about finitariness.  However, it may not be
  immediately clear how to define finitariness of a De Bruijn monad.

  \begin{defi}\label{def:support}
    A De Bruijn monad $(X,s,v)$ is finitary iff each
    of its elements $x ∈ X$ has a (finite) support $Nₓ ∈ ℕ$, in the sense that
    for all $f∶ ℕ → ℕ$ fixing the first $Nₓ$
    numbers, the corresponding renaming $v∘f$ fixes $x$, i.e., $x[v∘f] = x$.
  \end{defi}
  \begin{exa}
    By Proposition~\ref{prop:initwb} below, the initial $S$-algebra is
    finitary, for any binding signature $S$.  For an example of infinitary De
    Bruijn monad, consider the greatest fixed point $νA. ℕ + Σ_S(A)$, for any $S$
    with at least one operation with more than one argument.  E.g., if
    $S$ has an operation of binding arity $(0,0)$, like application in
    $λ$-calculus, then the term $v(0)\ (v(1)\ (v(2)\ …))$ does not
    have finite support.
  \end{exa}
  \begin{prop}
    \label{prop:support-eqs}
    Let $x∈ X$ be an element of a De Bruijn monad $(X,s,v)$.
    The following are equivalent:
    \begin{enumerate}
    \item $x$ has support $N$;
      \label{it-support1}
    \item given any pair of assignments $f₁,f₂∶ ℕ → X$ which coincide on the first $N$ numbers,
      $x[f₁] = x[f₂]$.
      \label{it-support2}
    \item for any assignment $f∶ ℕ → X$ fixing the first $N$ variables
      (in the sense that $f(n) = v(n)$ for any $n < N$), $x[f] = x$;
      \label{it-support3}
    \end{enumerate}
  \end{prop}
  \begin{proof}

    \eqref{it-support1} $ ⇒ $ \eqref{it-support2} Suppose given two
    assignments $f₁,f₂∶ ℕ → X$ such that $f₁(n)=f₂(n)$ for any $n<
    N$. Let us fix some bijection $s∶ ℕ → ℕ + ℕ$ such
    that$s(n)=in₁(n)$ for all $n < N$.  For $i ∈ \{1,2\}$, let
    $hᵢ∶ ℕ → ℕ$ fix the first $N$ numbers and map any $n ≥ N$ to
    $s^{-1}(inᵢ(n))$, where $in₁,in₂∶ ℕ → ℕ + ℕ$ are the coproduct
    injections. Furthermore, let $u∶ ℕ → X$ map any $n < N$ to $f₁(n)=f₂(n)$
    and any $n≥N$ to $[f₁,f₂](s(n))$, where $[f₁,f₂]∶ ℕ + ℕ → X$
    denotes the copairing of $f₁$ and $f₂$. We then have $fᵢ = u ∘ hᵢ$, for
    $i ∈ \{1, 2\}$. Indeed, for each $i$:
    \begin{itemize}
    \item For any $n<N$, we have $u(hᵢ(n))=u(n)=fᵢ(n)$ by definition.
    \item For any $n≥N$, we have $hᵢ(n)=s^{-1}(inᵢ(n))$, hence
      $s(hᵢ(n))=inᵢ(n)$. But we know that $hᵢ(n)≥N$, because otherwise
      we would have $s(hᵢ(n))=in₁(k)$ for some $k < N$, which does not
      hold. We thus obtain
      $u(hᵢ(n))=[f₁,f₂](s(hᵢ(n))=[f₁,f₂](inᵢ(n))=fᵢ(n).$
    \end{itemize}
    Thus, $x[fᵢ] = x[u∘hᵢ]=x[v∘hᵢ][u]$. Since $x$ has support $N$,
    $x[v∘hᵢ] = x$. Hence, $x[f₁] = x[u]=x[f₂]$.

    \eqref{it-support2} $ ⇒ $ \eqref{it-support3} 
    Suppose given an assignment $f∶ ℕ → X$ fixing the first $N$ variables.
    Then, $f$ coincides with $v$ on the first $N$ variables. Thus,
    $x[f]=x[v]=x$.

    \eqref{it-support3} $ ⇒ $ \eqref{it-support1} 
    Let $f∶ ℕ → ℕ$ fixing the first $N$ numbers. Then, $v∘f$, as an assignment,
    also does. Thus, $x[v∘f] = x$.
  \end{proof}

  \begin{defi}
    \label{def:wb-db-alg}
    For any binding signature $S$, let $S\DBAlg_{\fin}$ denote the full
    subcategory spanning De Bruijn $S$-algebras whose underlying De
    Bruijn monad is finitary.
  \end{defi}

  \begin{prop}\label{prop:initwb}
    The subcategory $S\DBAlg_{\fin}$ includes the initial object.
  \end{prop}
  \begin{proof}
   One can define by induction the greatest free variable $N$ of a term $x$ (or
   $0$ if $x$ is closed). Then, $x$ has support $N+1$.
  \end{proof}
  \subsection{Bridging the gap}
  We may at last state the relationship between initial-algebra
  semantics of binding signatures in presheaves and in De Bruijn monads:
  \begin{thm} \label{thm:gap} Consider any binding signature $S$.
    The subcategories $Φ_S\Mon_{\its}$ and $S\DBAlg_{\fin}$ are
    equivalent.
  \end{thm}

  \begin{rem}
    The moral of this is that, if one removes pathological objects from both
    $Φ_S\Mon$ and $S\DBAlg$, then one obtains equivalent categories,
    which both retain the initial object.  Thus, up to equivalence, the two approaches to
    initial-algebra semantics of binding signatures differ only
    marginally.

    Restricting attention to well-behaved objects, we may thus benefit
    from the strengths of both approaches. Typically, in De Bruijn
    monads, free variables need to be computed explicitly, while
    presheaves come with intrinsic scoping, as terms are indexed by
    sets of potential free variables. Conversely, in some settings,
    observational equivalence may relate programs with different sets
    of free variables~\cite{DBLP:books/daglib/0004377}.  In such
    cases, it is useful to have all terms collected in one single set.
    This needs to be computed (and involves non-trivial quotienting)
    in presheaves, while it is direct in De Bruijn monads.
  \end{rem}

  The remainder of this section is devoted to sketching the proof of
  Theorem~\ref{thm:gap}, and may be skipped on a first reading as
  it relies on the module-based interpretation of the binding
  conditions described later in~§\ref{s:module-based}.

  We start by proving that both De Bruijn monads and finitary monads
  are monoids in monoidal, full subcategories of $[𝐒𝐞𝐭,𝐒𝐞𝐭]$.  Let us first treat
  the easy case of finitary monads:
      \begin{lem}\label{lem:endos:fpt}
        The category $𝐌𝐨𝐧[𝔽,𝐒𝐞𝐭]$ of monoids in $[𝔽,𝐒𝐞𝐭]$ for the
        monoidal structure of Definition~\ref{def:monFSet}, is
        equivalent to the category $𝐌𝐨𝐧[𝐒𝐞𝐭,𝐒𝐞𝐭]_f$ of monoids in
        $[𝐒𝐞𝐭,𝐒𝐞𝐭]_f$.
      \end{lem}
      \begin{proof}
        By definition of the monoidal structure on $[𝔽,𝐒𝐞𝐭]$.
      \end{proof}
      Now for De Bruijn monads:
    \begin{defi}
      Let $[𝐒𝐞𝐭,𝐒𝐞𝐭]_{ℵ₁,\its₀}$ denote the full subcategory spanned by
      $ℵ₁$-accessible endofunctors which preserve empty intersections.
      \end{defi}

      \begin{lem}\label{lem:endos:db}
        Evaluation at $ℕ$ induces an equivalence between
        the category\linebreak{}$𝐌𝐨𝐧[𝐒𝐞𝐭,𝐒𝐞𝐭]_{ℵ₁,\its₀}$ of monoids in
        $[𝐒𝐞𝐭,𝐒𝐞𝐭]_{ℵ₁,\its₀}$ and the category $𝐃𝐁𝐌𝐧𝐝$ of De Bruijn monads.
      \end{lem}
      Note that any monad $T$ on $𝐒𝐞𝐭$ induces a De Bruijn monad $T(ℕ)$ by
      restricting the monadic bind and unit. This induces a functor whose
      restriction to
      $𝐌𝐨𝐧[𝐒𝐞𝐭,𝐒𝐞𝐭]_{ℵ₁,\its₀}$ underlies the above claimed equivalence.
      \begin{proof}[Proof sketch]
        De Bruijn monads are equivalently monads relative to the
        embedding\linebreak{}$𝐁[ℕ,ℕ] → 𝐒𝐞𝐭$ of the full subcategory on $ℕ ∈ 𝐒𝐞𝐭$.
        Now, presheaves on a category are equivalent to presheaves on
        its Cauchy completion, and we prove that the Cauchy completion
        of $𝐁[ℕ,ℕ]$, i.e., the category of idempotent maps $ℕ → ℕ$, is
        equivalent to the full subcategory $\bar{𝔽}⁺$ of $𝐒𝐞𝐭$ spanned
        by non-empty, finite ordinals and $ℕ$. De Bruijn monads are
        thus equivalent to monads relative to the embedding
        $J⁺∶ \bar{𝔽}⁺ ↪ 𝐒𝐞𝐭$.  Now, because the embedding is full,
        functors $\bar{𝔽}⁺ ↪ 𝐒𝐞𝐭$ are equivalent to functors
        $𝐒𝐞𝐭 → 𝐒𝐞𝐭$ which preserve the initial object and are
        $ℵ₁$-accessible. Letting $[𝐒𝐞𝐭,𝐒𝐞𝐭]_{ℵ₁,0}$ denote the
        category of such functors, we thus obtain an equivalence
        $$[𝐁[ℕ,ℕ],𝐒𝐞𝐭] ≃ [𝐒𝐞𝐭,𝐒𝐞𝐭]_{ℵ₁,0}\rlap{,}$$
        which is monoidal. We thus obtain an equivalence
        $$𝐃𝐁𝐌𝐧𝐝 ≃ 𝐌𝐨𝐧[𝐒𝐞𝐭,𝐒𝐞𝐭]_{ℵ₁,0}\rlap{.}$$
        It remains to make the link with $𝐌𝐨𝐧[𝐒𝐞𝐭,𝐒𝐞𝐭]_{ℵ₁,\its₀}$.
        At this point, there is a difficulty.
        Indeed, the functor
        $𝐩𝐢𝐧∶ [𝐒𝐞𝐭,𝐒𝐞𝐭]_{ℵ₁,\its₀} → [𝐒𝐞𝐭,𝐒𝐞𝐭]_{ℵ₁,0}$ defined by
        $$
        𝐩𝐢𝐧(F)(X) = \left \{ \begin{array}{rcl}
                                   ∅ & \mbox{if $X=∅$} \\
                                   F(X) & \mbox{otherwise} %
                                 \end{array}
                               \right .$$ is an equivalence preserving the
                               identity endofunctor, but it is
                               however not monoidal: e.g., letting
                               $F=∅$ and $G=1$, we have
        $$𝐩𝐢𝐧(G ∘ F)(1) = G (F (1)) = G (∅) = 1$$
        while
        $$(𝐩𝐢𝐧(G) ∘ 𝐩𝐢𝐧(F))(1) =
        𝐩𝐢𝐧(G)(F(1)) = 𝐩𝐢𝐧(G)(∅) = ∅.$$
        Still, we have:
        \begin{lem}
          For any monoid $F ∈  [𝐒𝐞𝐭,𝐒𝐞𝐭]_{ℵ₁,\its₀}$ and any $G∈ [𝐒𝐞𝐭,𝐒𝐞𝐭]_{ℵ₁,\its₀}$, we have
          $𝐩𝐢𝐧(G ∘ F) = 𝐩𝐢𝐧(G) ∘ 𝐩𝐢𝐧(F)$.
        \end{lem}
        \begin{proof}
          We prove the more general fact that, if $F(X) ≠ ∅$ at
          any $X ≠ ∅$, then $𝐩𝐢𝐧(G) ∘ 𝐩𝐢𝐧(F) = 𝐩𝐢𝐧(G∘F)$:
          \begin{itemize}
          \item at $X=∅$, we have
            $$(𝐩𝐢𝐧(G) ∘ 𝐩𝐢𝐧(F))(∅) = ∅ = 𝐩𝐢𝐧 (G ∘ F)(∅)\rlap{,}$$
          \item at any $X ≠ ∅$, we have 
            $$(𝐩𝐢𝐧(G) ∘ 𝐩𝐢𝐧(F))(X) =
            𝐩𝐢𝐧(G)(F(X)) = G (F (X)) = 𝐩𝐢𝐧(G ∘ F)(X).$$
          \end{itemize}
          But at any $X ≠ ∅$, any monoid $F$ is equipped with a unit
          component $X → F(X)$, so $F(X) ≠ ∅$, hence the result.
        \end{proof}

          We thus obtain an equivalence
          $$𝐌𝐨𝐧[𝐒𝐞𝐭,𝐒𝐞𝐭]_{ℵ₁,\its₀} ≃ 𝐌𝐨𝐧[𝐒𝐞𝐭,𝐒𝐞𝐭]_{ℵ₁,0}$$
          over $𝐩𝐢𝐧$, hence the result.
      \end{proof}

      We then characterise well-behavedness in both contexts, as follows.
      \begin{defi}
        Let $[𝐒𝐞𝐭,𝐒𝐞𝐭]_{f,\its₀}$ denote the full subcategory of
        $[𝐒𝐞𝐭,𝐒𝐞𝐭]$ spanned by finitary endofunctors preserving empty
        intersections.
      \end{defi}
      \begin{lem}
        The following squares commute,
        \begin{center}
          \diag(0.6,0.4){%
            {𝐌𝐨𝐧[𝐒𝐞𝐭,𝐒𝐞𝐭]_{f,\its₀}} \& {𝐌𝐨𝐧[𝐒𝐞𝐭,𝐒𝐞𝐭]_{ℵ₁,\its₀}} \\
            𝐃𝐁𝐌𝐧𝐝_{\fin} \& 𝐃𝐁𝐌𝐧𝐝 %
          }{%
            (m-1-1) edge[labela={}] (m-1-2) %
            edge[labell={\scalebox{1.03}{$≃$}}] (m-2-1) %
            (m-2-1) edge[labelb={}] (m-2-2) %
            (m-1-2) edge[labelr={\scalebox{1.03}{$≃$}}] (m-2-2) %
          }\quad
          \diag(0.6,0.4){%
            {𝐌𝐨𝐧[𝐒𝐞𝐭,𝐒𝐞𝐭]_{f,\its₀}} \& {𝐌𝐨𝐧[𝐒𝐞𝐭,𝐒𝐞𝐭]_f} \\
            {𝐌𝐨𝐧[𝔽,𝐒𝐞𝐭]_{\its}} \& {𝐌𝐨𝐧[𝔽,𝐒𝐞𝐭]} %
          }{%
            (m-1-1) edge[labela={}] (m-1-2) %
            edge[labell={\scalebox{1.03}{$≃$}}] (m-2-1) %
            (m-2-1) edge[labelb={}] (m-2-2) %
            (m-1-2) edge[labelr={\scalebox{1.03}{$≃$}}] (m-2-2) %
          }        \end{center}
        and all vertical functors are equivalences.
      \end{lem}
      \begin{proof}
        For De Bruijn monads, well-behavedness is finitarity.
        For presheaves, well-behavedness is preservation of empty
        intersections.
      \end{proof}
      
      \begin{cor}
        We obtain a chain of equivalences
        \begin{equation}
        𝐃𝐁𝐌𝐧𝐝_{\fin} ≃ 𝐌𝐨𝐧[𝐒𝐞𝐭,𝐒𝐞𝐭]_{f,\its₀} ≃ 𝐌𝐨𝐧[𝔽,𝐒𝐞𝐭]_{\its}.\label{eq:wb:equivs}
      \end{equation}
    \end{cor}

      The point is now to prove that this chain of equivalences lifts
      to one between $S\DBAlg_{\fin}$ and $Φ_S\Mon_{\its}$, for any
      binding signature $S$.

      For this, we adopt the viewpoint of modules over
      monads~\cite{DBLP:journals/iandc/HirschowitzM10}
      (see~§\ref{s:module-based} below for the module-based interpretation of
      the binding conditions).  Let
      $S = (O,\ar)$ denote any binding signature.  We first introduce the
      analogue of the endofunctor $Φ_S$ induced by a binding signature
      (Definition~\ref{def:PhiS}) in the context of $[𝐒𝐞𝐭,𝐒𝐞𝐭]$:
      \begin{defi}
        We define $ℱ_S∶ [𝐒𝐞𝐭,𝐒𝐞𝐭] → [𝐒𝐞𝐭,𝐒𝐞𝐭]$ by
        $$ℱ_S(F)(X) = ∑_{o ∈ O} ∏_{i ∈ |\ar(o)|} F(X + \ar(o)ᵢ).$$
      \end{defi}
      We then show that this functor restricts to the relevant subcategories.
      \begin{prop}
        For any
        $𝐂 ∈ \{ [𝐒𝐞𝐭,𝐒𝐞𝐭]_{ℵ₁,\its₀}, [𝐒𝐞𝐭,𝐒𝐞𝐭]_f,[𝐒𝐞𝐭,𝐒𝐞𝐭]_{f,\its₀} \}$,
        the functor $ℱ_S$ restricts to a functor $ℱ^𝐂_S∶ 𝐂 → 𝐂$ making the following
        square commute.
        \begin{center}
          \diag{%
            𝐂 \& 𝐂 \\
            {[𝐒𝐞𝐭,𝐒𝐞𝐭]} \& {[𝐒𝐞𝐭,𝐒𝐞𝐭]}
          }{%
            (m-1-1) edge[labela={ℱ^𝐂_S}] (m-1-2) %
            edge[into,labell={}] (m-2-1) %
            (m-2-1) edge[labelb={ℱ_S}] (m-2-2) %
            (m-1-2) edge[into,labelr={}] (m-2-2) %
          }
        \end{center}
        Furthermore, for any monoid $T ∈ 𝐂$, $ℱ^𝐂_S(T)$ forms a $T$-module, with action given at
        each $o ∈ O$ with $\ar(o) = (n₁,…,nₚ)$ by
        $$∏_{i ∈ p} T(T(X) + nᵢ)
        → ∏_{i ∈ p} T(T(X + nᵢ))
        \xto{∏ᵢ μ} ∏_{i ∈ p} T(X + nᵢ).$$
      \end{prop}
      \begin{proof}
        For the various restrictions of $ℱ_S$, one checks that each of
        the conditions (finitarity, $ℵ₁$-accessibility, preservation
        of empty intesections) is closed under coproducts, finite
        products, and shift, i.e., any $T ↦ T({-}+n)$.  For the second
        statement, it holds in $[𝐒𝐞𝐭,𝐒𝐞𝐭]$, and all considered
        subcategories are full.
      \end{proof}
      \begin{defi}
        For any
        $𝐂 ∈ \{ [𝐒𝐞𝐭,𝐒𝐞𝐭]_{ℵ₁,\its₀}, [𝐒𝐞𝐭,𝐒𝐞𝐭]_f,[𝐒𝐞𝐭,𝐒𝐞𝐭]_{f,\its₀} \}$,
        an \alert{$S^𝐂$-algebra} is a monoid $T$ in $𝐂$ equipped with a
        module morphism $ℱ^𝐂_S(T) → T$. A morphism of $S^𝐂$-algebras
        is a monoid morphism commuting with action. We denote by
        $S^𝐂\alg$ the category of $S^𝐂$-algebras.
      \end{defi}
      We next prove that in the case of $[𝐒𝐞𝐭,𝐒𝐞𝐭]_f$ and
      $[𝐒𝐞𝐭,𝐒𝐞𝐭]_{ℵ₁,\its₀}$ this interpretation of $S$ corresponds to
      its interpretations in presheaves and De Bruijn monads through
      the equivalences of Lemmas~\ref{lem:endos:fpt}
      and~\ref{lem:endos:db}.
      
      \begin{lem}
        We have commuting squares
        \begin{equation}
          \diag{%
            S^{[𝐒𝐞𝐭,𝐒𝐞𝐭]_f}\alg \& Φ_S\Mon \\
            {𝐌𝐨𝐧[𝐒𝐞𝐭,𝐒𝐞𝐭]_f} \& {𝐌𝐨𝐧[𝔽,𝐒𝐞𝐭]} %
          }{%
            (m-1-1) edge[labela={\scalebox{1.03}{$≃$}}] (m-1-2) %
            edge[labell={}] (m-2-1) %
            (m-2-1) edge[labelb={\scalebox{1.03}{$≃$}}] (m-2-2) %
            (m-1-2) edge[labelr={}] (m-2-2) %
          }
          \qquad
          \diag{%
            S^{[𝐒𝐞𝐭,𝐒𝐞𝐭]_{ℵ₁,\its₀}}\alg \& S\DBAlg \\
            {𝐌𝐨𝐧[𝐒𝐞𝐭,𝐒𝐞𝐭]_{ℵ₁,\its₀}} \& {𝐃𝐁𝐌𝐧𝐝} %
          }{%
            (m-1-1) edge[labela={\scalebox{1.03}{$≃$}}] (m-1-2) %
            edge[labell={}] (m-2-1) %
            (m-2-1) edge[labelb={\scalebox{1.03}{$≃$}}] (m-2-2) %
            (m-1-2) edge[labelr={}] (m-2-2) %
          }
          \label{eq:S:equivs}
        \end{equation}
      \end{lem}
      \begin{proof}
        The first square is easy. The second is a tedious
        verification that the binding conditions correspond
        to the definition of module morphisms
        $ℱ^{[𝐒𝐞𝐭,𝐒𝐞𝐭]_{ℵ₁,\its₀}}_S(T) → T$.
      \end{proof}

      Finally, we show that the restrictions of $Φ_S\Mon$ and
      $S\DBAlg$ to well-behaved objects are equivalent to
      $S^{[𝐒𝐞𝐭,𝐒𝐞𝐭]_{f,\its₀}}\alg$.

      Indeed, by definition, we have pullback squares
    \begin{equation}
      \Diag(.5,.5){
        \pullbackk{m-2-1}{m-1-1}{m-1-2}{draw,-,shorten >=3pt} %
      }{%
        Φ_S\Mon_{\its} \& Φ_S\Mon \\
        𝐌𝐨𝐧[𝔽,𝐒𝐞𝐭]_{\its} \& 𝐌𝐨𝐧[𝔽,𝐒𝐞𝐭] %
      }{%
        (m-1-1) edge[labela={}] (m-1-2) %
        edge[labell={}] (m-2-1) %
        (m-2-1) edge[labelb={}] (m-2-2) %
        (m-1-2) edge[labelr={}] (m-2-2) %
      }
      \qquad
      \Diag(.5,.5){ %
        \pullbackk{m-2-1}{m-1-1}{m-1-2}{draw,-,shorten >=4pt} %
      }{%
        S\DBAlg_{\fin} \& S\DBAlg \\
        𝐃𝐁𝐌𝐧𝐝_{\fin} \& 𝐃𝐁𝐌𝐧𝐝 
      }{%
        (m-1-1) edge[labela={}] (m-1-2) %
        edge[labell={}] (m-2-1) %
        (m-2-1) edge[into,labelb={}] (m-2-2) %
        (m-1-2) edge[labelr={}] (m-2-2) %
      }
      \label{eq:pbks}
    \end{equation}
    so by the equivalences~\eqref{eq:S:equivs} and~\eqref{eq:wb:equivs} the theorem follows
    from the next result.
      \begin{prop}
        We have the following pullback squares.
        \begin{center}
          \Diag{%
            \pullbackk{m-2-2}{m-1-2}{m-1-1}{draw,-,shorten >=3pt} %
            \pullbackk{m-2-2}{m-1-2}{m-1-3}{draw,-,shorten >=3pt} %
          }{%
            S^{[𝐒𝐞𝐭,𝐒𝐞𝐭]_f}\alg \&  S^{[𝐒𝐞𝐭,𝐒𝐞𝐭]_{f,\its₀}}\alg \& S^{[𝐒𝐞𝐭,𝐒𝐞𝐭]_{ℵ₁,\its₀}}\alg  \\
            {𝐌𝐨𝐧[𝐒𝐞𝐭,𝐒𝐞𝐭]_f} \&  {𝐌𝐨𝐧[𝐒𝐞𝐭,𝐒𝐞𝐭]_{f,\its₀}} \& {𝐌𝐨𝐧[𝐒𝐞𝐭,𝐒𝐞𝐭]_{ℵ₁,\its₀}}  %
          }{%
            (m-1-1) edge[<-,labela={}] (m-1-2) %
            edge[labell={}] (m-2-1) %
            (m-2-1) edge[<-,labelb={}] (m-2-2) %
            (m-1-2) edge[labelr={}] (m-2-2) %
            (m-1-2) edge[labela={}] (m-1-3) %
            edge[labell={}] (m-2-2) %
            (m-2-2) edge[labelb={}] (m-2-3) %
            (m-1-3) edge[labelr={}] (m-2-3) %
          }%
          \end{center}
      \end{prop}

%
%

  \section{Strength-based interpretation of the binding conditions}\label{s:strength}
  In the previous section, we have compared the category $S\DBAlg$ of
  models of a binding signature $S$ in De Bruijn monads with the usual
  category of $Φ_S$-monoids~\cite{fiore:presheaf}.  In fact, the
  latter approach is much more general, in the sense that it does not
  only work for binding signatures but for so-called \alert{pointed
    strong} endofunctors~\cite{fiore:presheaf}, and in fact also for
  the more general \alert{structurally strong} endofunctors introduced
  by Borthelle et al.~\cite{BHL}.

  In this section, we show that De Bruijn algebras also generalise
  from binding signatures to structurally strong endofunctors, in the
  following sense.  To any binding signature $S$, we associate such an
  endofunctor, say $Σ_S$, such that $Σ_S\Mon ≅ S\DBAlg$, where
  $Σ_S\Mon$ is as defined for any structurally strong endofunctor by
  Borthelle et al.
  
  This way we give a categorical status to binding signatures, as
  particular structurally strong endofunctors on $𝐒𝐞𝐭$.
  \begin{rem}
    We do not (yet) prove existence of an initial $Σ$-De Bruijn
    algebras for any larger class of endofunctors than those of the
    form $Σ_S$.
  \end{rem}

  \begin{rem}
    We resort to structurally strong endofunctors because pointed
    strong endofunctors live on monoidal
    categories~\cite{fiore:presheaf,DBLP:conf/lics/Fiore08}, while we
    have seen in Corollary~\ref{cor:skew} that our tensor product
    merely equips $𝐒𝐞𝐭$ with skew monoidal structure.  (The very
    purpose of structurally strong endofunctors is to generalise
    pointed strong endofunctors to the skew monoidal case.)  Following
    up on Remark~\ref{rem:BNNmonoid}, we could
    equivalently work with the monoidal category $[𝐁[ℕ,ℕ],𝐒𝐞𝐭]$, in
    which the machinery of pointed strong endofunctors applies.
\end{rem}

\begin{rem}
  In fact, the isomorphism $Σ_S\Mon ≅ S\DBAlg$ is almost an
  equality, since the only difference lies in the difference between a
  family $(M^{|\ar(o)|} → M)_{o ∈ O}$ of operations and its cotupling
  $∑_{o ∈ O} M^{|\ar(o)|} → M$: one could easily adjust the
  presentation to get an exact match.
\end{rem}

  The starting point is that the endofunctor $Σ_S$ associated to any
  given binding signature $S$ may be equipped with a family of maps
  \[𝐝𝐛𝐬_S∶ Σ_S(X) ⊗ Y → Σ_S(X⊗Y)\] that will be used to specify how
  substitution commutes with the operations of $S$.  However, in order
  for such a map to be well-defined for binding operations, we need to
  assume that $Y$ features variables and renaming, i.e., that it is a
  \alert{pointed $ℕ$-module}. Moreover, this map should satisfy some
  compatibility laws. These definitions and conditions are detailed
  in~§\ref{ss:sstrength}, where we furthermore recall structurally strong
  endofunctors.  In~§\ref{ss:ssinterp}, we interpret binding
  signatures as such endofunctors, we recall the category of
  $Σ$-monoids, for any structurally strong endofunctor $Σ$, and
  establish the announced isomorphism of categories.

  \subsection{Structural strengths}  \label{ss:sstrength}
  We start by introducing a notion of set equipped with variables and
  renamings, in Definition~\ref{d:pointed-mod} below.
  Recalling from Example~\ref{ex:nat} that $ℕ$ forms a De Bruijn monad,
  we have:
  \begin{defi}
    \label{def:N-module}
    An \alert{$ℕ$-module} is a set $X$ equipped with an action of the
    monoid $ℕ^ℕ$, namely a map $r∶ X×ℕ^ℕ=X⊗ℕ → X$, making the
    following diagrams commute.
    \begin{center}
\begin{tikzcd}[ampersand replacement=\&]
	{(X ⊗ ℕ) ⊗ ℕ} \&\& {X ⊗ (ℕ ⊗ ℕ)} \\
	{X ⊗ ℕ} \&\& {X ⊗ ℕ} \\
	\& X
	\arrow["{α_{X,ℕ,ℕ}}", from=1-1, to=1-3]
	\arrow["{X ⊗ λ_ℕ}", from=1-3, to=2-3]
	\arrow["r", from=2-3, to=3-2]
	\arrow["{r ⊗ ℕ }"', from=1-1, to=2-1]
	\arrow["r"', from=2-1, to=3-2]
      \end{tikzcd}
      \quad
\begin{tikzcd}[ampersand replacement=\&]
	X \&\& {X ⊗ ℕ} \\
	\& X
	\arrow["r", from=1-3, to=2-2]
	\arrow["{ρ_X}", from=1-1, to=1-3]
	\arrow[Rightarrow, no head, from=1-1, to=2-2]
\end{tikzcd}      
\end{center}
A \alert{morphism of $ℕ$-modules} $(X,r) → (Y,s)$ is a map $f∶ X → Y$
between underlying sets commuting with action, i.e., making the
following square commute.
\begin{center}
  \diag{
    X ⊗ ℕ \& Y ⊗ ℕ \\
    X \& Y %
  }{ (m-1-1) edge[labela={f ⊗ ℕ}] (m-1-2) %
    edge[labell={r}] (m-2-1) %
    (m-2-1) edge[labelb={f}] (m-2-2) %
    (m-1-2) edge[labelr={s}] (m-2-2) %
    }
  \end{center}
  Finally, $ℕ$-modules and morphisms between them form a category,
  which we denote by $ℕ\Mod$.
  \end{defi}
  \begin{nota}
    We generally denote $r(x,f)$ by $x[f]ᵣ$, or merely $x[f]$ when
    $r$ is clear from context.
  \end{nota}
  \begin{exa}\label{ex:dbmnd:Nmodule}
    Any De Bruijn monad $(X,s,v)$ (in particular $(ℕ,λ,\id)$ itself) has a canonical
      structure of $ℕ$-module given by $r(x,f)=x[v∘f]ₛ$.
  \end{exa}
  \begin{defi}
    \label{d:pointed-mod}\hfill
    \begin{itemize}
    \item A \alert{pointed $ℕ$-module} is an $ℕ$-module $(X, r)$,
      equipped with a map $v∶ ℕ → X$ which is a morphism of
      $ℕ$-modules, i.e., such that the following square commutes.
      \begin{center}
        \diag{%
          ℕ ⊗ ℕ \& X ⊗ ℕ \\
          ℕ \& X 
        }{%
          (m-1-1) edge[labela={v ⊗ ℕ}] (m-1-2) %
          edge[labell={λ_ℕ}] (m-2-1) %
          (m-2-1) edge[labelb={v}] (m-2-2) %
          (m-1-2) edge[labelr={r}] (m-2-2) %
        }
      \end{center}
    \item A morphism of pointed $ℕ$-modules $(X,r,v) → (Y,s,w)$ is a morphism of $ℕ$-modules
      $f∶ (X,r)→(Y,s)$ commuting with point, i.e., such that
      the following triangle commutes.
      \begin{center}
        \diag{%
          \& ℕ \\
          X \& \& Y %
        }{%
          (m-1-2) edge[labelal={v}] (m-2-1) %
	edge[labelar={w}] (m-2-3) (m-2-1) edge[labelb={f}] (m-2-3) %
        }
      \end{center}
  \item 
    Let $ℕ\Mod_ℕ$ denote the category of pointed $ℕ$-modules.
  \end{itemize}
  \end{defi}
  \begin{rem}
    Equivalently, $ℕ\Mod_ℕ$ is the coslice category $ℕ/(ℕ\Mod)$.
  \end{rem}
  \begin{exa}\label{ex:dbmnd:module}
    The canonical $ℕ$-module structure of any De Bruijn monad
    $(X,s,v)$ (in particular $(ℕ,λ,\id)$ itself), described in
    Example~\ref{ex:dbmnd:Nmodule}, is in fact pointed, by the map
    $v∶ ℕ → X$.
  \end{exa}

  We now define a tensor product on (pointed) $ℕ$-modules,
  following~\cite[(8.1)]{Warpings}.

  \begin{defi}
    Given an $ℕ$-module $(X,r)$ and a set $Y$, let $X ⊠ Y$ denote the
    following coequaliser in $𝐒𝐞𝐭$.
    \begin{equation}
      \hspace*{-2ex} \diag|baseline=(m-2-1.base)|(0,1){%
        \& X ⊗ (ℕ ⊗ Y) \\
        (X ⊗ ℕ) ⊗ Y \&             \& X ⊗ Y \& X ⊠ Y
      }{%
        (m-2-1) edge[labelal={α_{X,ℕ,Y}}] (m-1-2) %
        edge[bend right=10,labelb={r ⊗ Y}] (m-2-3) %
        (m-1-2) edge[labelar={X⊗λ_Y}] (m-2-3) %
        (m-2-3) edge[onto,labela={κ_{X,Y}}] (m-2-4)
      }
      \label{eq:btens}
    \end{equation}
  \end{defi}
  \begin{nota}
    Concretely, $X⊠Y$ is the set of equivalence classes of pairs
    $(x,υ) ∈ X × Y^ℕ$, modulo the equation $(x[ρ],υ) = (x,υ∘ρ)$, for
    any $ρ∶ ℕ → ℕ$. We denote such an equivalence class by $x⦇υ⦈$, and
    extend the notation to $σ⦇υ⦈$, for any assignment $σ∶ ℕ → X$,
    i.e., $σ⦇υ⦈(n) = σ(n)⦇υ⦈$ for all $n$.
  \end{nota}
  \begin{propC}[{\cite[Theorem~8.1]{Warpings}}]
    When $Y$ is equipped with $ℕ$-module structure, $X⊠Y$ admits a
    canonical $ℕ$-module structure, such that $x ⦇ υ ⦈ [ f ] = x ⦇ υ [ f ] ⦈$.
    This makes the category $ℕ\Mod$
    of $ℕ$-modules into a skew monoidal category (with unit $ℕ$, and
    invertible right unitor). Furthermore, the forgetful functor is
    monoidal, and creates monoids in the sense that monoids are the
    same in $ℕ\Mod$ and in $𝐒𝐞𝐭$.
  \end{propC}
  \begin{proof}
    To apply~\cite[Theorem~8.1]{Warpings}, we need to prove that
    tensoring on the right in $𝐒𝐞𝐭$ preserves reflexive coequalisers,
    which holds by interchange of colimits since $X ⊗ Y = X × Y^ℕ$ is
    the $Y^ℕ$-fold coproduct of $X$ with itself.
  \end{proof}
  In fact, this extends to $ℕ\Mod_ℕ$:
  \begin{prop}
    Given pointed $ℕ$-modules $(X,r,v)$ and $(Y,s,w)$, the $ℕ$-module
    $X⊠Y$ is canonically pointed by the map
    $$ℕ \xto{ρ_ℕ} ℕ ⊗ ℕ \xto{v ⊗ w} X ⊗ Y \xonto{κ_{X,Y}} X ⊠ Y.$$
  \end{prop}
  \begin{proof}
    This result was proved and formalised in a general skew monoidal setting
    in~\cite{BHL},
    see~\cite[\lstinline|IModules.PtIModule_tensor|]{BHLcode}.
  \end{proof}

  Now that we have defined the tensor product of pointed $ℕ$-modules, we
  may introduce structural strengths.
  \begin{defiC}[{\cite[Definition~2.11]{BHL}}]
    A \alert{structural strength} on an endofunctor $Σ∶ 𝐒𝐞𝐭 → 𝐒𝐞𝐭$ is
    a natural transformation $st_{X,Y}∶ Σ(X)⊗Y → Σ(X⊗Y)$, where $X$ is
    any set and $Y$ is a pointed $ℕ$-module, making the following
    diagrams commute,
      \begin{center}
    \diag{%
      \& Σ(A) \\
      Σ(A)⊗ℕ \& \& Σ (A⊗ℕ) %
    }{%
      (m-1-2) edge[labelal={ρ_{Σ(A)}}] (m-2-1) %
      edge[labelar={Σ(ρ_A)}] (m-2-3) %
      (m-2-1) edge[labelb={st_{A,ℕ}}] (m-2-3) %
    }
    \diag(.8,2.2){%
      (Σ(A)⊗X)⊗Y \& Σ(A⊗X)⊗Y \& Σ ((A⊗X)⊗Y) \\
      Σ(A)⊗(X⊠Y) \& \& Σ (A⊗(X⊠Y))
    }{%
      (m-1-1) edge[labela={\scalebox{1.03}{$st_{A,X}⊗Y$}}] (m-1-2) %
      edge[labell={\scalebox{1.03}{$α'_{Σ(A),X,Y}$}}] (m-2-1) %
      (m-1-2) edge[labela={\scalebox{1.03}{$st_{A⊗X,Y}$}}] (m-1-3) %
      (m-1-3) edge[labelr={\scalebox{1.03}{$Σ (α'_{A,X,Y})$}}] (m-2-3) %
      (m-2-1) edge[labelb={\scalebox{1.03}{$st_{A,X⊠Y}$}}] (m-2-3) %
    }
  \end{center}
  where $α'_{A,X,Y}$ is
  $(A⊗κ_{X,Y}) ∘ α_{A,X,Y}$, for any $A$.
  \end{defiC}
  \begin{rem}
    In examples, the first axiom will entail that the ``identity''
    assignment should cross operations unchanged.  In terms of De
    Bruijn monads, the ``identity'' assignment is merely the variables
    map $v$, so, e.g., in the setting of Example~\ref{ex:lam}, the
    axiom boils down to the fact that lifting $v$ yields $v$ again:
    $⇑v = v$.  The second axiom will entail the substitution lemma
    $x[σ][ϕ] = x[σ[ϕ]]$, where we recall that by definition
    $σ[ϕ](n) = σ(n)[ϕ]$. E.g., if crossing a given unary operation $o$
    maps assignments $σ$ to $σ'$, the axiom says that
    $(σ[ϕ])' = σ'[ϕ']$.  This is just what is needed for a proof by
    induction to go through, as in
    $$
    \begin{array}{rcll}
      o(x)[σ][ϕ] & = & o(x[σ'][ϕ']) & \mbox{by the assumed binding condition} \\
                 & = &  o(x[σ'[ϕ']]) & \mbox{by induction hypothesis} \\
                 & = &  o(x[(σ[ϕ])']) & \mbox{by the second axiom} \\
                 & = & o(x)[σ[ϕ]] & \mbox{by the binding condition again.}
    \end{array}$$
    The technique extends to operations with more complex arities.
  \end{rem}

  \subsection{De Bruijn algebras as \texorpdfstring{$Σ$}{Sigma}-monoids}\label{ss:ssinterp}
  Let us now interpret binding signatures as structurally strong
  endofunctors, and show that the corresponding category of models
  coincides with De Bruijn algebras.
  
  We can readily equip the endofunctor $Σ_S$ associated to any binding
  signature $S$ (Definition~\ref{def:SigmaS}) with a structural
  strength $𝐝𝐛𝐬_S$, which we call the \alert{De Bruijn} strength
  prescribed by $S$ on $Σ_S$.
  \begin{rem}
    Let us recall that by
    definition, $Σ_S$ ignores the binding information in $S$: we have
    $$Σ_S(X) = ∑_{o ∈ O} X^{pₒ}\rlap{,}$$ where $\ar(o) = (nᵒ₁,…,nᵒ_{pₒ})$ for
    all $o ∈ O$.
  \end{rem}
  In order to define $𝐝𝐛𝐬_S$, we start by adapting the definition of
  assignment lifting (Definition~\ref{def:prime}) to pointed
  $ℕ$-modules.
  \begin{defi}
    Let $(Y,r,v)$ be a pointed $ℕ$-module.
    For any assignment
    $σ∶ ℕ → Y$, we define the assignment $⇑σ∶ ℕ → Y$ by
    \[
      \begin{array}[t]{rcl}
        (⇑σ)(0) & = & v(0) \\
        (⇑σ)(n+1) & = & σ(n)[↑]\rlap{,}
      \end{array}
    \]
    
    \noindent where $↑∶ ℕ → ℕ$ denotes the successor map.  Let
    $⇑^{0}A = A$, and $⇑^{n+1}A = ⇑(⇑^{n}A)$.
  \end{defi}
  
  Using this, let us now define the De Bruijn strength of the identity
  functor.  We will then iterate the process to show that each
  iterated lifting also equips the identity functor with a structural
  strength.  Finally, we will use this as a basis for equipping the
  endofunctor $Σ_S$ associated with any binding signature $S$, with a
  structural strength.
  \begin{defi}
    The \alert{first De Bruijn strength} of the identity functor is the map
    \begin{displaymath}
      \begin{array}{rcl}
        𝐝𝐛𝐬_{\id,X,Y}∶ X⊗Y & → & X⊗Y \\
        (x,σ) & ↦ & (x,⇑ σ)\rlap{,} %
      \end{array}
    \end{displaymath}
    defined for all $X ∈ 𝐒𝐞𝐭$ and $Y ∈ ℕ\Mod_ℕ$.
  \end{defi}
  \begin{prop}
    The first De Bruijn strength is a structural strength on the
    identity functor.
  \end{prop}
  \begin{proof}
    We first check commutation with the right unitor, in this case
    \begin{center}
\begin{tikzcd}[ampersand replacement=\&]
	\& A \\
	{A⊗ℕ} \&\& {A⊗ℕ\rlap{.}}
	\arrow["{ρ_A}"', from=1-2, to=2-1]
	\arrow["{𝐝𝐛𝐬_{\id,A,ℕ}}"', from=2-1, to=2-3]
	\arrow["{ρ_A}", from=1-2, to=2-3]
      \end{tikzcd}
    \end{center}
    This triangle commutes because any $a ∈ A$ is mapped by $ρ_A$ to
    $(a,\id)$, and then by $𝐝𝐛𝐬_{\id,A,ℕ}$ to $(a,⇑\id)$. But
    $⇑\id = \id$, hence the result.

    For commutation with the associator,
    \begin{center}
\begin{tikzcd}[ampersand replacement=\&]
	{(A ⊗ X) ⊗ Y} \&\& {(A ⊗ X) ⊗ Y} \&\& {(A ⊗ X) ⊗ Y} \\
	{A ⊗ (X ⊗ Y)} \&\&\&\& {A ⊗ (X ⊗ Y)} \\
	{A ⊗ (X ⊠ Y)} \&\&\&\& {A ⊗ (X ⊠ Y)}
	\arrow["{\scalebox{1.2}{$α_{A,X,Y}$}}"', from=1-1, to=2-1]
	\arrow["{\scalebox{1.2}{$A ⊗ κ_{X,Y}$}}"', from=2-1, to=3-1]
	\arrow["{\scalebox{1.2}{$𝐝𝐛𝐬_{\id,A,X} ⊗ Y$}}", from=1-1, to=1-3]
	\arrow["{\scalebox{1.2}{$𝐝𝐛𝐬_{\id,A⊗X,Y}$}}", from=1-3, to=1-5]
	\arrow["{\scalebox{1.2}{$α_{A,X,Y}$}}", from=1-5, to=2-5]
	\arrow["{\scalebox{1.2}{$A ⊗κ_{X,Y}$}}", from=2-5, to=3-5]
	\arrow["{\scalebox{1.2}{$𝐝𝐛𝐬_{\id,A,X⊠Y}$}}"', from=3-1, to=3-5]
      \end{tikzcd}
    \end{center}
    we observe that any triple $(a,σ,υ) ∈ (A ⊗ X) ⊗ Y$
    is mapped by the bottom left composite to
    $$(a, ⇑(σ⦇υ⦈))\rlap{,}$$
    where we define $σ⦇υ⦈(n) ≔ σ(n)⦇υ⦈$.
    Furthermore, $(a,σ,υ)$ is mapped by the top right composite to
    $$(a, (⇑σ)⦇⇑υ⦈)\rlap{,}$$
    so we are left with the task of proving 
    $$⇑(σ⦇υ⦈) = (⇑σ)⦇⇑υ⦈.$$
    Let $u∶ ℕ → X$ and $v : ℕ → Y$ denote the points of $X$ and $Y$.
    We proceed by case analysis:
    \begin{itemize}
    \item At $0$, we have:
      $$
      \begin{array}{rcl}
        ⇑(σ⦇υ⦈)(0) & = & v(0) \\
                   & = & (⇑υ)(0) \\
                   & = & u(0)⦇⇑υ⦈ \\
                   & = & (⇑σ)(0) ⦇⇑υ⦈  \\
                   & = & (⇑σ)⦇⇑υ⦈(0)\rlap{.}
      \end{array}$$
    \item At any $n+1$, we have
      $$
      \begin{array}{rcl}
        ⇑(σ⦇υ⦈)(n+1) & = & σ⦇υ⦈(n)[↑] \\
                     & = & σ(n)⦇υ⦈[↑] \\
                     & = & σ(n)⦇υ[↑]⦈\rlap{,}
      \end{array}$$
    where by definition $υ[↑](n) = υ(n)[↑]$.
      $$
      \begin{array}{rcl}
        (⇑σ)⦇⇑υ⦈(n+1) & = & (⇑σ)(n+1)⦇⇑υ⦈ \\
                      & = & σ(n)[↑]⦇⇑υ⦈ \\
                      & = & σ(n)⦇(⇑υ)∘↑⦈.
      \end{array}$$
      But we have
      $(⇑υ)∘↑\ = υ[↑]$ since,
      for all $p ∈ ℕ$, we have

      \noindent \hfill $((⇑υ)∘↑)(p) = (⇑υ) (p+1) = υ(p)[↑] = υ[↑](p).$ \qedhere
    \end{itemize}
  \end{proof}

  Furthermore, we have:
  \begin{prop}\label{prop:compost}
    Structurally strong endofunctors compose, in the sense that if $F$
    and $G$ are structurally strong endofunctors, then so is $G∘F$,
    with structural strength given by the composite
    \begin{equation}
    G (F (X)) ⊗Y \xto{} G (F (X)⊗Y) \xto{} G (F (X⊗Y)).\label{eq:composite:strength}
  \end{equation}
  \end{prop}
  \begin{proof}
    For the first axiom, we have
    \begin{center}
      \begin{tikzcd}[ampersand replacement=\&]
	\& {G(F(A))} \\
	{G(F(A))⊗ℕ} \& {G(F(A)⊗ℕ)} \& {G(F(A ⊗ ℕ))\rlap{.}}
	\arrow["{ρ_{G(F(A))}}"', from=1-2, to=2-1]
	\arrow["{G(F(ρ_A))}", from=1-2, to=2-3]
	\arrow["{stᴳ_{F(A),ℕ}}"', from=2-1, to=2-2]
	\arrow["{G(st^F_{A,ℕ})}"', from=2-2, to=2-3]
	\arrow["{G(ρ_{F(A)})}", from=1-2, to=2-2]
      \end{tikzcd}
    \end{center}
    The second axiom holds by chasing the following diagram.

    \noindent \hfill
\begin{tikzcd}[ampersand replacement=\&,baseline=(\tikzcdmatrixname-5-1.base)]
	{G(F(A))⊗X ⊗Y} \&\& {G(F(A)⊗X) ⊗Y} \&\& {G(F(A⊗X)) ⊗Y} \\
	\&\& {G(F(A)⊗X ⊗Y)} \&\& {G(F(A⊗X) ⊗Y)} \\
	{G(F(A))⊗(X ⊗Y)} \&\& {G(F(A)⊗(X ⊗Y))} \&\& {G(F(A⊗X⊗Y))} \\
	\&\&\&\& {G(F(A⊗(X⊗Y)))} \\
	{G(F(A))⊗(X ⊠Y)} \&\& {G(F(A)⊗(X ⊠Y))} \&\& {G(F(A⊗(X⊠Y)))}
	\arrow["{\scalebox{1.2}{$stᴳ_{F(A),X} ⊗Y$}}", from=1-1, to=1-3]
	\arrow["{\scalebox{1.2}{$G(st^F_{A,X})⊗Y$}}", from=1-3, to=1-5]
	\arrow["{\scalebox{1.2}{$stᴳ_{F(A⊗X),Y}$}}", from=1-5, to=2-5]
	\arrow["{\scalebox{1.2}{$G(st^F_{A⊗X,Y})$}}", from=2-5, to=3-5]
	\arrow["{\scalebox{1.2}{$G(F(α))$}}", from=3-5, to=4-5]
	\arrow["{\scalebox{1.2}{$G(F(A⊗κ_{X,Y}))$}}", from=4-5, to=5-5]
	\arrow["{\scalebox{1.2}{$α$}}"'right, from=1-1, to=3-1, xshift=-10pt]
	\arrow["{\scalebox{1.2}{$G(F(A))⊗κ_{X, Y}$}}"'right, from=3-1, to=5-1, xshift=-10pt]
	\arrow["{\scalebox{1.2}{$stᴳ_{F(A),X⊠Y}$}}"', from=5-1, to=5-3]
	\arrow["{\scalebox{1.2}{$G(st^F_{A,X⊠Y})$}}"', from=5-3, to=5-5]
	\arrow["{\scalebox{1.2}{$stᴳ_{F(A)⊗X,Y}$}}"', from=1-3, to=2-3]
	\arrow["{\scalebox{1.2}{$G(st^F_{A,X}⊗Y)$}}"', from=2-3, to=2-5]
	\arrow["{\scalebox{1.2}{$G(α)$}}"', from=2-3, to=3-3]
	\arrow["{\scalebox{1.2}{$G(F(A)⊗κ_{X, Y})$}}"', from=3-3, to=5-3]
\end{tikzcd}
  \end{proof}
  Combining the last two results, any $\idⁿ = \id$ is structurally
  strong, with the following strength, obtained by inductively
  unfolding~\eqref{eq:composite:strength}:
  \begin{defi}\label{def:dbsn}
    Let the \alert{$n$th De Bruijn strength} of the identity functor,
    $𝐝𝐛𝐬ⁿ$, be defined by $𝐝𝐛𝐬ⁿ_{\id,X,Y}(x,σ) = (x,⇑^{n}σ)$.
  \end{defi}
  
  In summary:
  \begin{prop}\label{prop:dbsn}
    Each $𝐝𝐛𝐬ⁿ$ is a structural strength on the identity functor.
  \end{prop}

  Let us now extend this to general binding arities:
  \begin{prop}
    \label{prop:strong-endo-product}
    Given structurally strong endofunctors $(F,st^F)$ and $(G,stᴳ)$,
    the pointwise product $F×G$ admits the structural strength defined at any $X ∈ 𝐂$
    and $Y ∈ ℕ\Mod_ℕ$ by the composite
    \begin{equation}
    (F(X) × G(X)) ⊗ Y \xto{⟨π₁ ⊗ Y, π₂ ⊗ Y⟩} (F(X) ⊗ Y) × (G(X) ⊗ Y)
    \xto{st^F_{X,Y} × stᴳ_{X,Y}} F (X ⊗ Y) × G (X ⊗ Y).\label{eq:strength:prod}
  \end{equation}
  \end{prop}
  \begin{proof}
    The first axiom holds by chasing the following diagram.
    \begin{center}
    \hspace*{-5pt}
      \begin{tikzcd}[ampersand replacement=\&]
	\&\& {F(X) × G(X)} \\
	{(F(X) × G(X)) ⊗ ℕ\!} \&\& {\!(F(X) ⊗ ℕ) × (G(X) ⊗ ℕ)\!} \&\& {\!F (X ⊗ ℕ) × G (X ⊗ ℕ)}
	\arrow["{⟨π₁ ⊗ ℕ, π₂ ⊗ ℕ⟩}"', from=2-1, to=2-3]
	\arrow["{st^F_{X,ℕ} × stᴳ_{X,ℕ}}"', from=2-3, to=2-5]
	\arrow["{\rho_{F(X) × G(X)}}"', from=1-3, to=2-1]
	\arrow["{F(\rho_X) × G(\rho_X)}", from=1-3, to=2-5]
	\arrow["{ρ_{F(X)} × ρ_{G(X)}}", from=1-3, to=2-3]
      \end{tikzcd}
    \end{center}      
    For the second axiom, we need to prove that the following diagram
    commutes.
    \begin{center}
\ajustedroit{
\begin{tikzcd}[ampersand replacement=\&]
	{(F(X) × G(X)) ⊗ Y⊗ Z} \&\& {((F(X) ⊗ Y) × (G(X) ⊗ Y))⊗Z} \\
	{(F(X) × G(X)) ⊗ (Y⊗ Z)} \&\& {(F (X ⊗ Y) × G (X ⊗ Y))⊗Z} \\
	\&\& {(F (X ⊗ Y) ⊗Z)× (G (X ⊗ Y)⊗Z)} \\
	{(F(X) × G(X)) ⊗ (Y⊠ Z)} \&\& {F (X ⊗ Y ⊗Z)× G (X ⊗ Y⊗Z)} \\
	\&\& {F (X ⊗ (Y ⊗Z))× G (X ⊗ (Y⊗Z))} \\
	{(F(X)  ⊗ (Y⊠ Z)) × (G(X) ⊗ (Y⊠ Z))} \&\& {F (X ⊗ (Y ⊠Z))× G (X ⊗ (Y⊠Z))}
	\arrow["{\scalebox{1.2}{$⟨π₁ ⊗ Y, π₂ ⊗ Y⟩ ⊗ Z$}}", from=1-1, to=1-3]
	\arrow["{\scalebox{1.2}{$(st^F_{X,Y} × stᴳ_{X,Y})⊗Z$}}", from=1-3, to=2-3]
	\arrow["{\scalebox{1.2}{$⟨π₁ ⊗ Z, π₂ ⊗ Z⟩$}}", from=2-3, to=3-3]
	\arrow["{\scalebox{1.2}{$st^F_{X⊗Y,Z} × stᴳ_{X⊗Y,Z}$}}", from=3-3, to=4-3]
	\arrow["{\scalebox{1.2}{$F(α_{X,Y,Z}) × G(α_{X,Y,Z})$}}", from=4-3, to=5-3]
	\arrow["{\scalebox{1.2}{$F (X ⊗ κ_{Y,Z})× G (X ⊗ κ_{Y,Z})$}}", from=5-3, to=6-3]
	\arrow["{\scalebox{1.2}{$α$}}"', from=1-1, to=2-1]
	\arrow["{\scalebox{1.2}{$(F(X) × G(X)) ⊗ κ_{Y, Z}$}}"', from=2-1, to=4-1]
	\arrow["{\scalebox{1.2}{$⟨π₁ ⊗ (Y⊠Z), π₂  ⊗ (Y⊠Z)⟩$}}"', from=4-1, to=6-1]
	\arrow["{\scalebox{1.2}{$st^F_{X,Y⊠Z} × stᴳ_{X,Y⊠Z}$}}"', from=6-1, to=6-3]
\end{tikzcd}}
    \end{center}
    For this, since the target is a product, we proceed componentwise,
    and by symmetry it suffices to check the first:

    \noindent \hfill
\ajustedroit[.9]{
\begin{tikzcd}[ampersand replacement=\&,baseline=(\tikzcdmatrixname-5-1.base)]
	{(F(X) × G(X)) ⊗ Y⊗ Z} \&\& {F(X) ⊗ Y⊗Z} \&\& {F (X ⊗ Y) ⊗Z} \\
	\&\&\&\& {F (X ⊗ Y ⊗Z)} \\
	{(F(X) × G(X)) ⊗ (Y⊗ Z)} \&\& {F(X) ⊗ (Y⊗Z)} \&\& {F (X ⊗ (Y ⊗Z))} \\
	\\
	{(F(X) × G(X)) ⊗ (Y⊠ Z)} \&\& {F(X)  ⊗ (Y⊠ Z)} \&\& {F (X ⊗ (Y ⊠Z))\rlap{.}}
	\arrow["{\scalebox{1.2}{$π₁ ⊗ Y⊗ Z$}}", from=1-1, to=1-3]
	\arrow["{\scalebox{1.2}{$st^F_{X,Y} ⊗Z$}}", from=1-3, to=1-5]
	\arrow["{\scalebox{1.2}{$st^F_{X⊗Y,Z}$}}", from=1-5, to=2-5]
	\arrow["{\scalebox{1.2}{$F(α_{X,Y,Z})$}}", from=2-5, to=3-5]
	\arrow["{\scalebox{1.2}{$F (X ⊗ κ_{Y,Z})$}}", from=3-5, to=5-5]
	\arrow["{\scalebox{1.2}{$α$}}"', from=1-1, to=3-1]
	\arrow["{\scalebox{1.2}{$(F(X) × G(X)) ⊗ κ_{Y, Z}$}}"', from=3-1, to=5-1]
	\arrow["{\scalebox{1.2}{$π₁ ⊗ (Y⊠Z)$}}"', from=5-1, to=5-3]
	\arrow["{\scalebox{1.2}{$st^F_{X,Y⊠Z} $}}"', from=5-3, to=5-5]
	\arrow["{\scalebox{1.2}{$α$}}", from=1-3, to=3-3]
	\arrow["{\scalebox{1.2}{$π₁ ⊗ (Y⊗ Z)$}}", from=3-1, to=3-3]
	\arrow["{\scalebox{1.2}{$F(X) ⊗ κ_{Y,Z}$}}", from=3-3, to=5-3]
\end{tikzcd}} \qedhere
  \end{proof}

  \begin{cor}\label{cor:dbsa}
    For any binding arity $a = (n₁,…,nₚ)$, the family
    \begin{displaymath}
      \begin{array}{rcl}
        𝐝𝐛𝐬_{a,X,Y}∶ Xᵖ⊗Y & → & (X⊗Y)ᵖ \\
        ((x₁,…,xₚ),σ) & ↦ & ((x₁,⇑^{n₁}σ),…,(xₚ,⇑^{nₚ}σ))\rlap{,} %
      \end{array}
    \end{displaymath}
    for all sets $X$ and pointed $ℕ$-modules $Y$, defines a structural
    strength on the endofunctor $X ↦ Xᵖ$, which we call the \alert{De
      Bruijn strength} $𝐝𝐛𝐬ₐ$ of $a$.
  \end{cor}
  \begin{proof}
    By inductively unfolding~\eqref{eq:strength:prod} using
    Proposition~\ref{prop:dbsn}.    
  \end{proof}

  As promised, let us now express the binding condition in terms
  of strengths:
  \begin{prop}\label{prop:bindingiffpenta:a}
    For any binding arity $a = (n₁,…,nₚ)$ and De Bruijn monad
    $(M,s,v)$, a map $o∶ Mᵖ → M$ satisfies the $a$-binding condition
    w.r.t.\ $(s,v)$ iff the following pentagon commutes.
    \begin{equation}
      \diag(1,1.3){%
        Xᵖ ⊗ X \&  (X⊗X)ᵖ \& Xᵖ \\
        X⊗X \& 
        \& X
      }{%
        (m-1-1) edge[labela={𝐝𝐛𝐬_{a,X,X}}] (m-1-2) %
       edge[labell={o⊗X}] (m-2-1) %
       (m-1-2) edge[labela={sᵖ}] (m-1-3) %
       (m-2-1) edge[labelb={s}] (m-2-3) %
       (m-1-3) edge[labelr={o}] (m-2-3) %
      }
      \label{eq:pentagon:a}
    \end{equation}    
  \end{prop}
  \begin{proof}
    The bottom left composite maps any tuple $((x₁,…,xₚ),σ)$
    to $o(x₁,…,xₚ)[σ]$, while the top right one
    maps it first to
    $$((x₁, ⇑^{n₁} σ),…,(xₚ, ⇑^{nₚ} σ))\rlap{,}$$
    then to
    $$(x₁[⇑^{n₁} σ],…,xₚ[⇑^{nₚ} σ])\rlap{,}$$
    and finally to
    $$o(x₁[⇑^{n₁} σ],…,xₚ[⇑^{nₚ} σ])\rlap{,}$$
    as desired.
  \end{proof}
  
  At last, let us now define the De Bruijn strength of the endofunctor
  $Σ_S$ induced by an arbitrary binding signature $S$.  For this, just
  as we have shown that structurally strong endofunctors are closed
  under products (Proposition~\ref{prop:strong-endo-product}), we start by
  showing that they are closed under coproducts.
  \begin{prop}
    \label{prop:coprod-strong}
  Given structurally strong endofunctors
    $(Fᵢ,stⁱ)_{i ∈ I}$, the pointwise coproduct $∑ᵢ Fᵢ$ admits the
    structural strength defined at any $X ∈ 𝐂$ and $Y ∈ I\Mod_I$ by
    the composite
    $$\Big (∑ᵢFᵢ(X) \Big )⊗Y ≅ ∑ᵢ(Fᵢ(X)⊗Y) \xto{∑ᵢ stⁱ_{X,Y}} ∑ᵢ Fᵢ(X⊗Y).$$
  \end{prop}
  \begin{proof}
    The first axiom holds by chasing the following diagram.
    \begin{center}
\begin{tikzcd}[ampersand replacement=\&]
	\&\& {∑ᵢFᵢ(X)} \\
	{(∑ᵢFᵢ(X))⊗\mathbb{N}} \&\& {∑ᵢ(Fᵢ(X)⊗\mathbb{N})} \&\& {∑ᵢ Fᵢ(X⊗\mathbb{N})}
	\arrow["{\scalebox{1.2}{$\cong$}}"', from=2-1, to=2-3]
	\arrow["{\scalebox{1.2}{$∑ᵢ stⁱ_{X,\mathbb{N}}$}}"', from=2-3, to=2-5]
	\arrow["{\scalebox{1.2}{$ρ_{∑ᵢFᵢ(X)}$}}"', from=1-3, to=2-1]
	\arrow["{\scalebox{1.2}{$∑ᵢFᵢ(ρ_X)$}}", from=1-3, to=2-5]
	\arrow["{\scalebox{1.2}{$∑ᵢρ_{Fᵢ(X)}$}}", from=1-3, to=2-3]
\end{tikzcd}
    \end{center}
    The second axiom holds by chasing the diagram in Figure~\ref{fig:big-diag}.
    \end{proof}
    \qedhere
    \begin{figure}
      \centering
    \ajustedroit[.9]{
\begin{tikzcd}[ampersand replacement=\&]
	{(∑ᵢFᵢ(X))⊗Y ⊗ Z} \&\& {∑ᵢ(Fᵢ(X)⊗Y) ⊗Z} \\
	{(∑ᵢFᵢ(X))⊗(Y ⊗ Z)} \\
	\& {∑ᵢ (Fᵢ(X)⊗Y ⊗ Z)} \& {∑ᵢ Fᵢ(X⊗Y) ⊗ Z} \\
	\\
	\& {∑ᵢ (Fᵢ(X)⊗(Y ⊗ Z))} \& {∑ᵢ (Fᵢ(X⊗Y) ⊗ Z)} \\
	{(∑ᵢFᵢ(X))⊗(Y ⊠ Z)} \&\& {∑ᵢ Fᵢ(X⊗Y ⊗ Z)} \\
	\&\& {∑ᵢ Fᵢ(X⊗(Y ⊗ Z))} \\
	{∑ᵢ(Fᵢ(X)⊗(Y ⊠ Z))} \&\& {∑ᵢ Fᵢ(X⊗(Y ⊠ Z))}
	\arrow["{\scalebox{1.2}{$[inᵢ ⊗ Y]ᵢ^{-1} ⊗ Z$}}", from=1-1, to=1-3]
	\arrow["{\scalebox{1.2}{$∑ᵢ stⁱ_{X,Y} ⊗ Z$}}", from=1-3, to=3-3]
	\arrow["{\scalebox{1.2}{$[inᵢ ⊗ Z]ᵢ^{-1} $}}", from=3-3, to=5-3]
	\arrow["{\scalebox{1.2}{$∑ᵢ stⁱ_{X⊗Y,Z}$}}", from=5-3, to=6-3]
	\arrow["{\scalebox{1.2}{$∑ᵢ Fᵢ(α_{X,Y,Z})$}}", from=6-3, to=7-3]
	\arrow["{\scalebox{1.2}{$∑ᵢ Fᵢ(X⊗κ_{Y,  Z})$}}", from=7-3, to=8-3]
	\arrow["{\scalebox{1.2}{$α $}}"', from=1-1, to=2-1]
	\arrow["{\scalebox{1.2}{$(∑ᵢFᵢ(X))⊗κ_{Y, Z}$}}"', from=2-1, to=6-1]
	\arrow["{\scalebox{1.2}{$[inᵢ ⊗ (Y ⊠Z)]ᵢ^{-1} $}}"', from=6-1, to=8-1]
	\arrow["{\scalebox{1.2}{$∑ᵢ stⁱ_{X,Y ⊠ Z}$}}"', from=8-1, to=8-3]
	\arrow["{\scalebox{1.2}{$[inᵢ ⊗ Z]ᵢ^{-1} $}}"'{pos=0.8}, from=1-3, to=3-2]
	\arrow["{\scalebox{1.2}{$∑ᵢ (stⁱ_{X,Y} ⊗ Z)$}}", from=3-2, to=5-3, xshift=-5pt]
	\arrow["{\scalebox{1.2}{$∑ᵢ α_{F_i(X),Y,Z}$}}", from=3-2, to=5-2, yshift=-5pt]
	\arrow["{\scalebox{1.2}{$∑ᵢ (Fᵢ(X)⊗κ_{Y, Z})$}}", from=5-2, to=8-1]
	\arrow["{\scalebox{1.2}{$[inᵢ ⊗ (Y ⊗ Z)]ᵢ^{-1} $}}"'{pos=0.7}, from=2-1, to=5-2]
      \end{tikzcd}
    }
\caption{Diagram chasing for the proof of Proposition~\ref{prop:coprod-strong}.}
      \label{fig:big-diag}
    \end{figure}

    Let us finally put things together:
  \begin{cor}
    For any binding signature $S = (O,\ar)$, the endofunctor $Σ_S$ induced by $S$
    admits as structural strength the composite
$$(∑ₒ X^{|\ar(o)|})⊗Y ≅ ∑ₒ (X^{|\ar(o)|}⊗Y) \xto{∑ₒ 𝐝𝐛𝐬_{\ar(o),X,Y}} ∑ₒ (X⊗Y)^{|\ar(o)|}\rlap{,}$$
    or more concretely
    $$\begin{array}[t]{rcl}
        Σ_S(X)⊗Y & → & Σ_S(X⊗Y) \\
        ((o,(x₁,…,x_{pₒ})),σ) & ↦ & (o,((x₁,⇑^{n₁}σ),…,(x_{pₒ},⇑^{n_{pₒ}}σ)))
                                    \mbox{,}
     \end{array}$$
     \noindent for all sets $X$ and pointed $ℕ$-modules $Y$,
     where $\ar(o) = (n₁,…,n_{pₒ})$.

     We call this the \alert{De Bruijn strength} $𝐝𝐛𝐬_S$ of $Σ_S$.
  \end{cor}
  \begin{proof}
    Recalling that, by Definition~\ref{def:SigmaS}, we have
    $Σ_S(X) = ∑_{o ∈ O} X^{|\ar(o)|}$, $Σ_S$ is a coproduct of
    functors $X ↦ X^{ | \ar(o)|}$ with structural strengths
    $𝐝𝐛𝐬_{\ar(o)}$ by Corollary~\ref{cor:dbsa}, hence admits the given
    structural strength by Proposition~\ref{prop:coprod-strong}.
  \end{proof}

  In order to relate the initial-algebra semantics
  of~§\ref{s:elementary} to the strength-based approach
  of~\cite{fiore:presheaf,DBLP:conf/lics/Fiore08}, let us recall the
  definition of models, following the generalisation to the skew
  monoidal setting~\cite{BHL}.
  \begin{defi}
    Given an endofunctor $Σ$ with structural strength $st$, a
    \alert{$Σ$-monoid} is an object $X$, equipped with monoid and
    $Σ$-algebra structures, say $s∶ X⊗X → X$, $v∶ ℕ → X$, and
    $a∶ Σ(X) → X$, making the following pentagon commute.
    \begin{equation}
      \diag{%
        Σ(X) ⊗ X \& Σ (X⊗X) \&  Σ(X) \\
        X⊗X \&      \& X
      }{%
        (m-1-1) edge[labela={st_{X,X}}] (m-1-2) %
       edge[labell={a⊗X}] (m-2-1) %
       (m-1-2) edge[labela={Σ(s)}] (m-1-3) %
       (m-2-1) edge[labelb={s}] (m-2-3) %
       (m-1-3) edge[labelr={a}] (m-2-3) %
      }
      \label{eq:pentagon}
    \end{equation}
    A morphism of $Σ$-monoids is a map which is both a monoid and a $Σ$-algebra morphism.

    Let $Σ\Mon$ denote the category of $Σ$-monoids and morphisms
    between them.    
  \end{defi}

  We may at last relate the initial-algebra semantics
  of~§\ref{s:elementary} with the strength-based approach:
  \begin{prop}
    For any binding signature $S= (O,\ar)$ and De Bruijn monad
    $(M,s,v)$ equipped with a map $o_M∶ Mᵖ → M$ for all $o ∈ O$ with
    $\ar(o) = (n₁,…,nₚ)$, the following are equivalent:
    \begin{enumerati}
    \item each map $o_M∶ Mᵖ → M$ satisfies the \alert{$a$-binding
        condition} w.r.t.\ $(s,v)$;
    \item the corresponding map $Σ_SM → M$ renders the
      pentagon~\eqref{eq:pentagon} (with $Σ ≔ Σ_S$ and $st ≔ 𝐝𝐛𝐬_S$)
      commutative.
    \end{enumerati}
  \end{prop}
  \begin{proof}
    By universal property of coproduct and distributivity, the
    pentagon~\eqref{eq:pentagon} commutes iff each corresponding
    pentagon~\eqref{eq:pentagon:a} does, which holds iff each $o$
    satisfies the $\ar(o)$-binding condition w.r.t.\ $(s,v)$ by
    Proposition~\ref{prop:bindingiffpenta:a}.
  \end{proof}
  
  \begin{cor}\label{cor:equivsigmamon}
    For any binding signature $S$, we have an isomorphism
    $Σ_S\Mon ≅ S\DBAlg$ of categories over $𝐃𝐁𝐌𝐧𝐝$.
  \end{cor}
  This readily entails the following (bundled) reformulation of
  Theorems~\ref{thm:initiality:I} and~\ref{thm:initiality:II}.
  \begin{cor}\label{cor:initiality:str}
  Consider any binding signature $S= (O,\ar)$, and let ${\DB}$ denote the
  initial $(ℕ+Σ_S)$-algebra, with structure maps $v∶ ℕ → {\DB}$ and
  $a∶   Σ_S({\DB}) → {\DB}$. Then:
  \begin{enumerati}
  \item \label{item:exists-subst} There exists a unique substitution map
    $s∶ {\DB} ⊗ {\DB} → {\DB}$ such that
    \begin{itemize}
    \item the map $ℕ⊗{\DB} \xrightarrow{v⊗{\DB}} {\DB}⊗{\DB} \xrightarrow{s} {\DB}$ coincides
      with the left unit of the skew monoidal structure $(n,f)↦ f(n)$, and
    \item 
      the pentagon~\eqref{eq:pentagon} (with
      $Σ ≔ Σ_S$) commutes.
    \end{itemize}
  \item This substitution map turns $({\DB},v,s,a)$ into a $Σ_S$-monoid.
  \item This $Σ_S$-monoid is initial in $Σ_S\Mon$.
  \end{enumerati}
\end{cor}
\begin{proof}
  Let $𝐌𝐨𝐧(𝐒𝐞𝐭)$ denote the category of monoids in $𝐒𝐞𝐭$ for the skew
  monoidal structure.  We have an equality $𝐌𝐨𝐧(𝐒𝐞𝐭) = 𝐃𝐁𝐌𝐧𝐝$ of
  categories, and the algebra structure $Σ_S({\DB}) → {\DB}$ is merely the
  cotupling of the maps $o_{\DB}$ of Theorem~\ref{thm:initiality:I}.  This
  correspondence translates one statement into the other.
\end{proof}
  \begin{rem}
    This result hints at a potential push-button proof of
    Theorems~\ref{thm:initiality:I} and~\ref{thm:initiality:II} (and
    Corollary~\ref{cor:initiality:str}). Indeed, it is almost an
    instance of~\cite[Theorem~2.15]{BHL}: the latter is stated for
    general skew monoidal categories instead of merely $𝐒𝐞𝐭$, but does
    not directly apply in the present setting, because it assumes that
    the tensor product is finitary in the second argument. However, we
    expect the generalisation consisting in replacing this finitarity
    condition with $α$-accessibility to
    be straightforward.
  \end{rem}

  \section{Module-based interpretation of the binding conditions}\label{s:module-based}
  In the previous section, we have shown that the construction of De
  Bruijn algebras generalises from binding signatures to structurally
  strong endofunctors, thus yielding a categorical status for binding
  signatures and a categorical interpretation of the binding
  conditions.

  In this section, we give binding signatures an alternative
  categorical status, as \alert{parametric modules} over De Bruijn
  monads, together with a corresponding categorical interpretation of
  the binding conditions. For this, we merely adapt to De Bruijn
  monads the treatment proposed for mere monads by Hirschowitz and
  Maggesi~\cite{hirscho:lam,DBLP:journals/iandc/HirschowitzM10}.


  Compared to the original setting~\cite{hirscho:lam}, a peculiarity
  is that \alert{derivation} of modules, the module-theoretic
  incarnation of variable binding, preserves the underlying object.
  In other words, it only affects substitution.

  In~§\ref{ss:modules}, we introduce modules over a De Bruijn
  monad. In~§\ref{ss:derivation}, we define module derivation.  We
  then introduce parametric De Bruijn modules in~§\ref{ss:parametric},
  and show how any binding signature $S$ yields such a module $M_S$.
  Finally, in~§\ref{ss:modulebinding}, we define the category $M\MAlg$
  of modular algebras of a parametric De Bruijn module $M$, and show
  that they provide an alternative categorical interpretation of the
  binding conditions by exhibiting an isomorphism $S\DBAlg ≅ M_S\MAlg$
  of categories over $𝐃𝐁𝐌𝐧𝐝$.

  \subsection{Modules over De Bruijn monads and first-order signatures}\label{ss:modules}
  There is a general notion of module over a monoid in a monoidal (or
  skew monoidal) category;
  we just give the instance we are concerned with.  Intuitively, if
  $X$ is a De Bruijn monad, an $X$-module is a set that admits
  substitution of variables by elements of $X$:
  \begin{defi}\label{def:module}
    For any De Bruijn monad $(X,s,v)$, an \alert{$X$-module} is a set
    $A$ equipped with a \alert{substitution} map, or \alert{action}, 
    $$r∶ A \times X^ℕ → A$$
    subject to the following condition, where we we use
    Notation~\ref{not:subst}:
    \begin{center}
      for all $a ∈ A$ and $f,g ∈ X^ℕ$, we have \qquad $a[f][g] = a [f[g]]$
      \qquad and \qquad $a[v] = a$.
    \end{center}
  \end{defi}
  \begin{rem}
    Please note that the first equation involves both
    substitution maps:\linebreak{}$a[f]_A[g]_A = a [f[g]_X]_A$.
  \end{rem}
  \begin{rem}
    The definition is consistent with the definition of $ℕ$-modules
    (Definition~\ref{def:N-module}), viewing $ℕ$ as a De Bruijn monad
    as in Example~\ref{ex:nat}.
  \end{rem}
  \begin{rem}
    Equivalently, an $X$-module is an algebra for the monad $-⊗X$,
    using the skew monoidal structure of Corollary~\ref{cor:skew}.
    Indeed, the equations amount to commutation of the following
    diagrams,
    \begin{center}
      \diag{%
        (A ⊗ X) ⊗ X \& \& A ⊗ (X ⊗ X) \\
        A ⊗ X \& \& A ⊗ X \\
        \& A
      }{%
        (m-1-1) edge[labela={α_{A,X,X}}] (m-1-3) %
       edge[labell={r ⊗ X}] (m-2-1) %
       (m-1-3) edge[labelr={A ⊗ s}] (m-2-3) %
       (m-2-1) edge[labelbl={r}] (m-3-2) %
       (m-2-3) edge[labelbr={r}] (m-3-2) %
     }
     \quad
     \diag{%
       A \& A ⊗ ℕ \& A ⊗ X \\
         \&       \& A\rlap{.} %
       }{%
         (m-1-1) edge[labela={ρ_A}] (m-1-2) %
         edge[bend right=10,identity] (m-2-3) %
         (m-1-2) edge[labela={A ⊗ v}] (m-1-3) %
         (m-1-3) edge[labelr={r}] (m-2-3) %
     }
    \end{center}
    which are exactly the equations for $({-}⊗X)$-algebras.
  \end{rem}
  Let us now introduce a few basic constructions of modules:
  \begin{defi} Consider any De Bruijn monad $X$.
   \begin{itemize}
   \item The \alert{tautological} $X$-module is $X$ itself, with action
     $X × X^ℕ → X$ given by substitution.
   \item Given $X$-modules $U$ and $V$, their \alert{binary
       product} is $U × V$, with action given by
     \begin{align*}
       U × V × X^ℕ &→ U × V \\
       (u,v,σ) &↦ (u[σ]_{U}, v[σ]_{V}).
     \end{align*}
     This extends straightforwardly to small products.
   \item Given $X$-modules $U$ and $V$, their
     \alert{coproduct} is $U+V$, with action
     defined by case analysis:
     $$\begin{array}{rcl}
       (U + V) × X^ℕ & → & U + V \\
       (in₁(u),σ) & ↦ & in₁(u[σ]_{U}) \\
       (in₂(v),σ) & ↦ & in₂(v[σ]_{V}).
       \end{array}$$
       This extends straightforwardly to small coproducts.
 \end{itemize}
  \end{defi}

  \subsection{Derivation of substitution for modules}\label{ss:derivation}
  In this subsection, we explain module derivation.  This operation
  does not change the carrier of the module, hence it acts on the
  substitution map only. In fact, it acts via the second argument of
  substitution, namely the assignment, as in~§\ref{s:elementary}
  and~§\ref{s:strength}.

  \begin{defi}
    Given a De Bruijn monad $X$,
   the \textbf{derivative} $A^{(1)}$ of an $X$-module $A$
    has the same carrier as $A$, with action given by
   \begin{displaymath}
     \begin{array}{rcl}
       A\times X^ℕ & → &  A \\        
       (a,σ) & ↦ &
                   a[⇑σ]\rlap{,}
     \end{array}
   \end{displaymath}
   where $⇑σ$ is as in Definition~\ref{def:prime}:
  \[
  \begin{array}[t]{rcl}
           (⇑σ)(0) & = & v(0) \\
        (⇑σ)(n+1) & = & σ(n)[↑]\rlap{,} 
    \end{array}
  \]   
  \end{defi}

Of course we may iterate this operation:
  \begin{defi}\label{def:derivmod} Let $A^{(0)} = A$, and
    $A^{(n+1)} = (A^{(n)})^{(1)}$.
 \end{defi}

 \subsection{Binding signatures as parametric modules}\label{ss:parametric}
 In order to interpret binding signatures, we now introduce a
 parametric version of modules.  For this, we construct a category
 $𝐃𝐁𝐌𝐨𝐝$ whose objects are pairs of a De Bruijn monad and a module
 over it, and then define parametric modules as sections of the
 forgetful functor $𝐃𝐁𝐌𝐨𝐝 → 𝐃𝐁𝐌𝐧𝐝$.

 \begin{defi}
   Let $𝐃𝐁𝐌𝐨𝐝$ denote the category with
   \begin{itemize}
   \item as objects all pairs $(X,(U,a))$ of a De Bruijn monad $X$ and
     an $X$-module $(U,a)$, with $a∶ U × X^ℕ → U$, and
   \item as morphisms $(X,(U,a)) → (Y,(V,b))$ all pairs $(f,g)$,
     where $f∶ X → Y$ is a De Bruijn monad morphism, and $g∶ U → V$
     is a map making the following diagram commute,
     \begin{center}
       \diag{%
         U × X^ℕ  \& V × Y^ℕ \\
         U \& V
       }{%
         (m-1-1) edge[labela={g × f^ℕ}] (m-1-2) %
         edge[labell={a}] (m-2-1) %
         (m-2-1) edge[labelb={g}] (m-2-2) %
         (m-1-2) edge[labelr={b}] (m-2-2) %
       }
     \end{center}
     or equivalently,
     $g(u[σ]_U) =  g(u)[f∘σ]_V$, for all
     $u ∈ U$ and $σ∶ ℕ → X$.
   \end{itemize}
   The \alert{forgetful functor} $𝒰∶ 𝐃𝐁𝐌𝐨𝐝 → 𝐃𝐁𝐌𝐧𝐝$ maps any
   $(X,(U,a))$ to $X$, and any $(f,g)$ to $f$.
 \end{defi}

 We now introduce parametric modules:
 \begin{defi}
   A \alert{parametric De Bruijn module} is a section of the forgetful
   functor $𝒰∶ 𝐃𝐁𝐌𝐨𝐝 → 𝐃𝐁𝐌𝐧𝐝$, i.e., a functor $M∶ 𝐃𝐁𝐌𝐧𝐝 → 𝐃𝐁𝐌𝐨𝐝$ such
   that $𝒰 ∘ M = \id_{𝐃𝐁𝐌𝐧𝐝}$.
 \end{defi}


 Binding signatures naturally induce parametric De Bruijn modules:
 \begin{defi}\hfill
   \begin{itemize}
   \item The \alert{tautological} parametric De Bruijn module,
     denoted by $θ$, maps
     any De Bruijn monad $X$ to itself, with action $X × X^ℕ → X$
     given by substitution.
   \item The \alert{derivative} $U^{(1)}$ of a parametric De
     Bruijn module $U$ is defined to map any De Bruijn monad $X$
     to $U(X)^{(1)}$, and any morphism $f∶ X → Y$ to $(f, U(f))$.
     This works because the following square commutes.
     \begin{center}
       \diag{%
         X^ℕ \& Y^ℕ \\
         X^ℕ \& Y^ℕ 
       }{%
         (m-1-1) edge[labela={f^ℕ}] (m-1-2) %
         edge[labell={⇑}] (m-2-1) %
         (m-2-1) edge[labelb={f^ℕ}] (m-2-2) %
         (m-1-2) edge[labelr={⇑}] (m-2-2) %
       }
     \end{center}
     Indeed, letting $v_X$ and $v_Y$ denote the respective variables
     maps of $X$ and $Y$, we show by case analysis on $n ∈ ℕ$ that
     for all $σ∶ ℕ → X$, we have $f^ℕ (⇑σ) (n) = ⇑(f^ℕ (σ)) (n)$:
     \begin{itemize}
     \item at $0$, we have
       \begin{align*}
         f^ℕ (⇑σ)(0) &= f (⇑σ(0)) \\
                     &= f(v_X(0)) \\
                     &= v_Y(0) \\
                     &= ⇑(f^ℕ (σ)) (0)\rlap{,}
       \end{align*}
       
     \item and at any $n+1$, we have
       \begin{align*}
         f^ℕ (⇑σ)(n+1) &= f (⇑σ(n+1)) \\
                       &= f(σ (n)[↑_X]) \\
                       &= f (σ (n))[f ∘ ↑_X] \\
                       &= f (σ (n))[↑_Y] \\
                       &= ⇑(f^ℕ (σ)) (n+1)\rlap{.}
       \end{align*}
     \end{itemize}
   \item The \alert{$n$th derivative} $U^{(n)}$ of a parametric De
     Bruijn module $U$ is defined by induction: $U^{(0)} = U$ and
     $U^{(n+1)} = (U^{(n)})^{(1)}$.
   \item Given parametric De Bruijn modules $U$ and $V$, their
     \alert{binary product} maps any $X$ to the $X$-module product
     $U(X)×V(X)$.  This extends straightforwardly to small products.
   \item The parametric De Bruijn module $Mₐ$ induced by a binding
     arity $a = (n₁,…,nₚ)$ is the product $∏_{i ∈ p} θ^{(nᵢ)}$ of
     derivatives of the tautological parametric De Bruijn module.
   \item Given parametric De Bruijn modules $U$ and $V$, their
     \alert{coproduct} maps any $X$ to the $X$-module coproduct
     $U(X)+V(X)$.  This extends straightforwardly to small coproducts.
   \item The parametric De Bruijn module $M_S$ induced by
     a binding signature $S = (O,\ar)$ is the coproduct
     $∑_{o ∈ O} M_{\ar(o)}$ of the parametric De Bruijn modules
     induced by the arities of all operations.
 \end{itemize}
 \end{defi}

  \subsection{Interpreting the binding conditions}\label{ss:modulebinding}
  In the previous subsection, we have interpreted binding signatures
  as parametric modules, but we have not yet defined the models of a
  parametric module. Let us do this now, and prove that, for any
  binding signature $S$, the category of De Bruijn $S$-algebras is
  isomorphic to the category of models of the induced parametric De
  Bruijn module $M_S$.

  \begin{defi}\hfill 
    \begin{itemize}
    \item Given a parametric De Bruijn module $U$, a
    \alert{$U$-algebra} is a De Bruijn monad $X$, equipped with an
    $X$-module morphism $α∶ U(X) → X$.
  \item For any $U$, given $U$-algebras $(X,α)$ and $(Y,β)$, a
    \alert{$U$-algebra morphism} is a De Bruijn monad morphism
    $f∶ X → Y$ making the following diagram commute,
    \begin{center}
      \diag{%
        U(X) \& U(Y) \\
        X \& Y %
      }{%
        (m-1-1) edge[labela={U(f)}] (m-1-2) %
        edge[labell={α}] (m-2-1) %
        (m-2-1) edge[labelb={f}] (m-2-2) %
        (m-1-2) edge[labelr={β}] (m-2-2) %
      }
    \end{center}
    or equivalently $f(α(u)) = β(U(f)(u))$, for all $u ∈ U(X)$.
  \item For any $U$, $U$-algebras and morphisms between them form a
    category, which we denote by $U\MAlg$.
  \item The \alert{forgetful functor} $𝒰ᴹ∶ U\MAlg → 𝐃𝐁𝐌𝐧𝐝$ maps any
    $(X,α)$ to $X$.
  \end{itemize}
\end{defi}

As announced, let us prove
\begin{prop}\label{prop:sndorder}
  For any binding signature $S$, the categories $S\DBAlg$ and
  $M_S\MAlg$ are isomorphic over $𝐃𝐁𝐌𝐧𝐝$.
\end{prop}
\begin{proof}
  The key point is that for any binding arity $a= (n₁,…,nₚ)$, a map
  $o∶ Xᵖ → X$ is an operation of binding arity $a$ iff it is an
  $X$-module morphism $∏_{i=1}ᵖ X^{(nᵢ)} → X$.
  Indeed, the latter condition unfolds to the fact that,
  for any assignment $σ∶ ℕ → X$ and tuple
  $(e₁,…,eₚ) ∈ Xᵖ$, we have
  $$o(e₁,…,eₚ)[σ] = o (e₁[⇑^{n₁}σ],…,eₚ[⇑^{nₚ}σ])\rlap{,}$$
  which is exactly the $a$-binding condition~\eqref{eq:binding}.
\end{proof}

We readily obtain the following (bundled) reformulation of
Theorems~\ref{thm:initiality:I} and~\ref{thm:initiality:II}.
  \begin{cor}
  Consider any binding signature $S= (O,\ar)$, and let ${\DB}$ denote the
  initial $(ℕ+Σ_S)$-algebra, with structure maps $v∶ ℕ → {\DB}$ and
  $a∶   Σ_S({\DB}) → {\DB}$. Then:
  \begin{enumerati}
  \item \label{item:exists-subst:module} There exists a unique substitution map
    $s∶ {\DB} ⊗ {\DB} → {\DB}$ such that
    \begin{itemize}
    \item the map $ℕ⊗{\DB} \xrightarrow{v⊗{\DB}} {\DB}⊗{\DB} \xrightarrow{s} {\DB}$ coincides
      with the left unit of the skew monoidal structure $(n,f)↦ f(n)$, and
    \item $a$ is an $X$-module morphism.
    \end{itemize}
  \item This substitution map turns $({\DB},v,s,a)$ into an
    $M_S$-algebra.
  \item This $M_S$-algebra is initial in $M_S\MAlg$.
  \end{enumerati}
\end{cor}

 \section{Simply-typed extension}\label{s:types}
 In this section, we extend the framework
 of~§\ref{s:dbmonads}--\ref{s:elementary}, which is untyped, to the
 simply-typed case. The development essentially follows the same
 pattern, replacing sets with families.

 We fix in the whole section a set $𝕋$ of \alert{types}, and call
 \alert{$𝕋$-sets} the objects of $𝐒𝐞𝐭^𝕋$.  A morphism $X → Y$ is a
 family $(X(τ) → Y(τ))_{τ ∈ 𝕋}$ of maps.

 \subsection{De Bruijn \texorpdfstring{$𝕋$}{T}-monads}
 In this subsection, we define the typed analogue of De Bruijn monads. 

The role of $ℕ$ will be played in the typed context by the following
$𝕋$-set.
\begin{defi}
  Let $𝐍 ∈ 𝐒𝐞𝐭^𝕋$ be defined by $𝐍(τ) = ℕ$.
\end{defi}


\begin{defi}
  Given a $𝕋$-set $X$, an \alert{$X$-assignment} is a morphism of indexed sets $𝐍 → X$. We
  sometimes merely use ``assignment'' when $X$ is clear from context.
\end{defi}

\pagebreak   
\begin{nota}\hfill
  \begin{itemize}
  \item We observe that $𝕋$-sets form a cartesian closed category,
    where the exponential object $X^Y$ is given by
    $(X^Y)(τ) = X(τ)^{Y(τ)}$.
  \item We distinguish it from the hom-set by writing the latter
    $[Y,X]$.
  \item For any set $A$ and $𝕋$-set $X$, let $A·X = ∑_{a ∈ A} X$
    denote the $A$-fold coproduct of $X$.
\end{itemize}
\end{nota}

The analogue of the tensor product $X⊗Y = X×Y^ℕ$ will be played by
$[𝐍,Y]·X$, i.e., the iterated self-coproduct of $X$, with one copy
per $Y$-assignment (see Definition~\ref{def:tensor-prod-simply-typed} below).

\begin{exa} Consider arbitrary $𝕋$-sets $X$, $Y$, and $Z$.
  \begin{itemize}
  \item The $𝕋$-set $[X,Y]·Z$ is such that for all types $τ$, we have
    $$([X,Y]·Z) (τ) = [X,Y] · Z(τ) = [X,Y] × Z(τ).$$
  \item The $𝕋$-set $Y^X×Z$ is such that
    for all types $τ$, we have $$(Y^X×Z) (τ) = Y(τ)^{X(τ)} × Z(τ).$$
  \end{itemize}
  We will use the former for generalising substitution to the typed
  case.
\end{exa}

\begin{nota}
  For coherence with the untyped case, we tend to write an element of
  $([𝐍,Y]·X)(τ)$ as $(x,f)$, with $x ∈ X(τ)$ and $f∶ 𝐍 → Y$.

\end{nota}

Furthermore, Notation~\ref{not:subst} straightforwardly adapts to
the typed case as follows.

\begin{nota} Consider any map $s∶ [𝐍,Y]·X → Z$.
  \begin{itemize}
  \item For all $τ ∈ 𝕋$, $x ∈ X(τ)$, and $g∶ 𝐍 → Y$, we write
    $x[g]_{s,τ}$ for $s_τ(x,g)$, or even $x[g]$ when $s$ and $τ$ are
    clear from context.
  \item Furthermore, $s$ gives rise to the map
    \[
    \begin{array}[t]{rcl}
      [𝐍,Y]·X^𝐍 & → & Z^𝐍 \\
      τ ↦ (g, f: ℕ → X(τ)) & ↦ & n ↦ f(n)[g]_{s,τ}.
      \end{array}
    \]

    We use notation similar to Notation~\ref{not:subst} for this map, i.e.,
    $f[g]_{s,τ}(n) ≔ f(n)[g]_{s,τ}$, or $f[g](n) = f(n)[g]$ when $s$
    and $τ$ are clear from context.
  \item We use the same notation for the map
    \[
    \begin{array}[t]{rcl}
        [𝐍,Y]×[𝐍,X] & → & [𝐍,Z] \\
      (g,f) & ↦ & τ,n ↦ f(n)[g]_{s,τ}.
      \end{array}
    \]
\end{itemize}
\end{nota}

The definition of De Bruijn monads generalises almost \emph{mutatis
  mutandis}:
\begin{defi}\label{def:Tdbmonad}
  A \alert{De Bruijn $𝕋$-monad} is a $𝕋$-set $X$, equipped with
  \begin{itemize}
  \item a \alert{substitution} morphism $s∶ [𝐍,X]·X → X$, which
    takes an element $x ∈ X$ and an assignment $f∶ 𝐍 → X$, and returns
    an element $x[f]$, and

  \item a \alert{variables} morphism $v∶ 𝐍 → X$,
  \end{itemize}
  such that for all $x ∈ X$, and $f,g∶ 𝐍 → X$, we have
  \begin{mathpar}
    x[f][g] = x[f[g]]
    \and
    v(n)[f] = f(n)
    \and
    x[v] = x\rlap{.}
  \end{mathpar}
\end{defi}

  \begin{exa}
    The $𝕋$-set $𝐍$ itself is clearly a De Bruijn $𝕋$-monad, with
    variables given by the identity and substitution $[𝐍,𝐍]·𝐍 → 𝐍$
    given by evaluation. It is in fact initial in $𝐃𝐁𝐌𝐧𝐝(𝕋)$.
  \end{exa}
  \begin{exa}\label{ex:stlc}
    The set $Λ_{\ST}$ of simply-typed $λ$-terms with free variables of
    type $τ$ in $ℕ×\{τ\}$, considered equivalent modulo $α$-renaming,
    forms a De Bruijn monad.  Variables $𝐍 → Λ_{\ST}$ are given by
    mapping, at any $τ$, any $n ∈ ℕ$ to the variable $(n,τ)$.
    Substitution $[𝐍,Λ_{\ST}] · Λ_{\ST} → Λ_{\ST}$ is standard,
    capture-avoiding substitution.  One main purpose of this section
    is to characterise $Λ_{\ST}$ by a universal property, and
    reconstruct it categorically.
  \end{exa}

  \begin{rem}\label{rem:values}
    Untyped languages with multiple syntactic categories form De
    Bruijn monads. Indeed, it suffices to take $𝕋$ to be the set of
    syntactic categories, and, for each $c ∈ 𝕋$, let $X(c)$ be the set
    of terms of syntactic category $c$.  This should be taken with a
    grain of salt, though, as this assumes that each syntactic
    category has its kind of variables. E.g., let us consider a
    $λ$-calculus in which we wish to distinguish values from
    terms. Syntax then goes as follows:
    \begin{align*}
      e &\Coloneqq v ｜ e\ e & \text{(terms)} \\
      v &\Coloneqq x ｜ λx.e & \text{(values)}\rlap{.}
    \end{align*}
    Attempting to organise this as a (simply-typed) De Bruijn monad
    $X$, we take $𝕋 = 2 = \{ 𝐭, 𝐯 \}$, and let $X(𝐭)$ be the set of
    terms, while $X(𝐯)$ is the set of values.  However, $X$ fails to
    be a De Bruijn monad because there are no term variables.  One way
    of understanding this is that values form an untyped De Bruijn
    monad, and terms form a module over it~\cite{HHL,HHLlong}.  We
    present a simply-typed version of this approach in detail below
    in~§\ref{ss:values}.  Another way out consists in adding term
    variables $α$ to the syntax, which thus becomes:
    \begin{align*}
      e &\Coloneqq α ｜ v ｜ e\ e & \text{(terms)} \\
      v &\Coloneqq x ｜ λx.e & \text{(values)}\rlap{.}
    \end{align*}
  \end{rem}

  \subsection{Morphisms of De Bruijn \texorpdfstring{$𝕋$}{T}-monads}
  \begin{defi}
    A morphism $(X,s,v) → (Y,t,w)$ between De Bruijn $𝕋$-monads is a
    morphism $f∶ X → Y$ of $𝕋$-sets commuting with substitution and
    variables, in the sense that for all $τ ∈ 𝕋$, $x ∈ X(τ)$, and
    $g∶ 𝐍 → X$ we have $f_τ(x[g]) = f_τ(x)[f∘g]$ and $f∘v = w$.
  \end{defi}
  \begin{rem}
    More explicitly, the first axiom says: $f_τ(s_τ(x,g)) = t_τ(f_τ(x),f∘g)$.
  \end{rem}
  \begin{prop}
    De Bruijn $𝕋$-monads and morphisms between them form a category
    $𝐃𝐁𝐌𝐧𝐝(𝕋)$.
  \end{prop}

    

  \subsection{De Bruijn \texorpdfstring{$𝕋$}{T}-monads as relative monads}
  The presentation based on relative monads extends to the typed
  setting, by replacing the functor $ℕ∶ 1 → 𝐒𝐞𝐭$ with $𝐍∶ 1 → 𝐒𝐞𝐭^𝕋$,
  so that $\Lan_𝐍(X)(Y) ≅ [𝐍,Y]·X$.  Thus, $𝐃𝐁𝐌𝐧𝐝(𝕋)$ is equivalently
  the category of monads relative to the functor $1 → 𝐒𝐞𝐭^𝕋$ picking
  $𝐍$.  For the record, let us explicitly introduce the corresponding
  tensor product.
\begin{defi}
  \label{def:tensor-prod-simply-typed}
  For any $𝕋$-sets $X$ and $Y$, let $X⊗Y = [𝐍,Y]·X$.
\end{defi}
\begin{nota}
  For coherence with the untyped case, we tend to write an element of
  $(X⊗Y)(τ)$ as $(x,f)$, with $x ∈ X(τ)$ and $f∶ 𝐍 → Y$.
\end{nota}

  \subsection{Initial-algebra semantics}
  We now adapt the initial-algebra semantics of~§\ref{s:elementary} to
  the typed case.
  \subsubsection{Assignment  lifting}
  Let us start by generalising lifting to the typed case. This relies on a
  typed form of lifting, which acts on all
  variables of a given type, leaving all other variables untouched.
  \begin{defi}
    Let $(X,s,v)$ denote any De Bruijn $𝕋$-monad.
    We first define a typed analogue $↑^{τ}$ of the $↑$
    of Definition~\ref{def:prime}, as below left, and then
    the \alert{lifting} of any assignment  $σ∶ 𝐍 → X$
    as below right.
    \begin{center}
      $\begin{array}[t]{rcll}
          (↑^{τ})_τ(n) & = & v_τ(n+1) \\
          (↑^{τ})_{τ'}(n) & = & v_{τ'}(n) & \mbox{(if $τ≠τ'$)}
       \end{array}$
       \hfill $\begin{array}[t]{rcll}
                 (⇑^{τ}σ)_τ(0) & = & v_τ(0) \\
                 (⇑^{τ}σ)_τ(n+1) & = & σ_τ(n)[↑^{τ}] \\
                 (⇑^{τ}σ)_{τ'}(n) & = & σ_{τ'}(n)[↑^{τ}] 
                              & \mbox{(if $τ≠τ'$).}
        \end{array}$
    \end{center}
    Finally, for any sequence $γ = (τ₁,…,τₙ)$ of types, we define
    $⇑^{γ}σ$ inductively, by $⇑^{ε}σ = σ$ and
    $⇑^{γ,τ}σ = ⇑^{τ}(⇑^{γ}σ)$, where $ε$ denotes the empty sequence.
  \end{defi}

  \subsubsection{Binding arities and binding conditions}
We may now generalise binding arities and the binding conditions to
the typed setting.

  \begin{defiC}[{\cite{FioreHur}}]\hfill
    \begin{itemize}
    \item A \alert{first-order arity} is a pair
      $((τ₁,…,τₚ),τ) ∈ 𝕋^*×𝕋$ of a list of types and a type.
    \item A \alert{binding arity} is a tuple
      $a = (((γ₁,τ₁),…,(γₚ,τₚ)),τ)$, where each $γᵢ ∈ 𝕋^*$ is a list
      of types, and each $τᵢ$, as well as $τ$, are types. In other
      words,  $a ∈ (𝕋^*×𝕋)^*×𝕋$.
    \item The \alert{first-order arity $|a|$ associated} to $a$ is
      $((τ₁,…,τₚ),τ) ∈ 𝕋^*×𝕋$.
    \end{itemize}
  \end{defiC}

  \begin{rem}
    An arity $(((γ₁,τ₁),…,(γₚ,τₚ)),τ)$ may be
    understood as follows:
    \begin{itemize}
    \item $τ$ is the return type;
    \item $(τ₁,…,τₚ)$ are the argument types;
    \item each list $γᵢ = (τⁱ₁,…,τⁱ_{qᵢ})$ specifies that the $i$th
      argument should be considered as binding $qᵢ$ variables, of
      respective types $τⁱ₁$,…,$τⁱ_{qᵢ}$.
    \end{itemize}
  \end{rem}
  \begin{nota}
    We write any arity $(((γ₁,τ₁),…,(γₚ,τₚ)),τ)$ as an inference rule
    \begin{mathpar}
    \inferrule{γ₁ ⊢ τ₁ \\ … \\ γₚ ⊢ τₚ}{⊢ τ}~· 
  \end{mathpar}
\end{nota}

\begin{exa}\label{ex:stl}
  The binding signature for simply-typed $λ$-calculus has two
  operations $\lam_{τ,τ'}$ and $\app_{τ,τ'}$ for each pair 
  $(τ,τ')$ of types, of respective arities
    \begin{center}
     $\inferrule{τ ⊢ τ'}{ ⊢ τ → τ'}$ \hfil and \hfil
    $\inferrule{ ⊢ τ → τ' \\ ⊢ τ}{ ⊢ τ'}~·$
  \end{center}
  \end{exa}

This allows us to generalise the binding conditions, as follows.

\begin{defi}
  Let $a = (((γ₁,τ₁),…,(γₚ,τₚ)),τ)$ be any binding arity, and $M$ be
  any $𝕋$-set equipped with morphisms $s∶ [𝐍,M]·M → M$ and $v∶ 𝐍 → M$.  An
  \alert{operation of binding arity $a$} is a map
    $o∶ M(τ₁)×…×M(τₚ) → M(τ)$ satisfying the following
    \alert{$a$-binding condition} w.r.t.\ $(s,v)$:
  \begin{equation}
    \begin{array}{l}
        ∀ σ∶ 𝐍 → M,
    x₁,…,xₚ ∈ M(τ₁)×…×M(τₚ), \\
    o(x₁,…,xₚ)[σ] = o(x₁[⇑^{γ₁}σ],…,xₚ[⇑^{γₚ}σ]).
    \end{array}
    \label{eq:bindingtyped}
    \end{equation}
\end{defi}

  \subsubsection{Binding signatures and algebras}
  Finally, we generalise signatures and their models to the typed
  setting, and state a typed initiality theorem.
  \begin{defi}
    A \alert{first-order typed signature} consists of a set $O$ of
    \alert{operations}, equipped with an \alert{arity} map
    $\ar∶ O → 𝕋^*× 𝕋$.
  \end{defi}
\begin{defi}
  Consider a first-order signature $S ≔ (O,\ar)$.
  \begin{itemize}
  \item An \alert{$S$-algebra} is a set $X$, together with, for each
    operation $o ∈ O$ with arity $((τ₁,…,τₚ),τ)$, a map
    $o_X(τ₁)×…×X(τₚ) → X(τ)$.
  \item A morphism $X → Y$ of $S$-algebras is a map between underlying
    sets commuting with operations, in the sense that for each $o ∈ O$,
    letting $((τ₁,…,τₚ),τ) ≔ \ar(o)$, we have
    for all $x₁,…,xₚ ∈ X(τ₁)×…×X(τₚ)$, 
    $f_τ(o_X(x₁,…,xₚ)) = o_Y(f_{τ₁}(x₁),…,f_{τₚ}(xₚ))$.
  \end{itemize}
  We denote by  $S\alg$ the category of $S$-algebras and morphisms
  between them.
\end{defi}

  \begin{defi}\hfill
    \begin{itemize}
    \item A \alert{$𝕋$-binding signature} consists of a set $O$ of
    \alert{operations}, equipped with an arity map $O → (𝕋^*×𝕋)^*×𝕋$.
    \item The first-order signature $|S|$ associated with a binding
      signature $S ≔ (O,\ar)$ is $|S| ≔ (O,|\ar|)$, where
      $|\ar|∶ O → 𝕋^*×𝕋$ maps any $o ∈ O$ to $|\ar(o)|$.
  \end{itemize}
  \end{defi}

  Let us now present the notion of De Bruijn $S$-algebra:
  \begin{defi}
    Consider any $𝕋$-binding signature $S ≔ (O, \ar)$.
 \begin{itemize}
 \item A \alert{De Bruijn $S$-algebra} consists of a De Bruijn
   $𝕋$-monad $(X,s,v)$, together with, for all $o ∈ O$, an operation
   of binding arity $\ar(o)$.
\item A morphism of De Bruijn $S$-algebras is a map $f∶ X → Y$ between
  underlying sets, which is a morphism both of De Bruijn monads and of
  $|S|$-algebras.
  \end{itemize}
  We denote by $S\DBAlg$ the category of De Bruijn $S$-algebras and
  morphisms between them.
\end{defi}

In order to extend the initiality theorem to the typed case, we need
to define the endofunctor induced by a $𝕋$-binding signature $S$, which only depends 
on $|S|$, as in the untyped case.
  \begin{defi}[Induced endofunctor]\hfill
    \begin{itemize}
    \item For any $τ ∈ 𝕋$, let $𝐲_τ$ denote the $𝕋$-set
      defined by
      \[ \begin{array}[t]{rcll}
          𝐲_τ(τ) & = & 1 \\
          𝐲_τ(τ') & = & ∅ & \mbox{(if $τ' ≠ τ$)}.
        \end{array}\]
    \item The endofunctor $Σₐ$ induced by any arity
      \[a = \inferrule{τ¹₁,…,τ¹_{q₁}⊢τ₁ \\ … \\ τᵖ₁,…,τᵖ_{qₚ} ⊢ τₚ}{⊢
          τ}\] is defined by 
       $Σₐ(X) = (X(τ₁)×…×X(τₚ))·𝐲_{τ}.$
      Thus, a $Σₐ$-algebra is a $𝕋$-set $X$ equipped with a morphism
      $(X(τ₁)×…×X(τₚ))·𝐲_{τ} → X$,  
      or equivalently a map 
      $X(τ₁)×…×X(τₚ) → X(τ).$
    \item The endofunctor $Σ_S$ induced by any $𝕋$-binding signature
      $S = (O,\ar)$ is defined by
      $$Σ_S(X) = ∑_{o∈O} Σ_{\ar(o)}(X).$$
    \end{itemize}
  \end{defi}

We have the following typed extension of the initiality theorem.
\begin{thm}\label{thm:typedinitiality}
  For any $𝕋$-binding signature $S$, let ${\DB}$ denote the initial
  $(𝐍+Σ_S)$-algebra, with structure morphisms $v∶ 𝐍 → {\DB}$ and
  $a∶ Σ_S({\DB}) → {\DB}$, inducing maps \[
    o_{\DB}∶ {\DB}(τ₁)×…×{\DB}(τₚ) → {\DB}(τ)\] for all
  $o ∈ O$ with $\ar(o) = (((γ₁,τ₁),…,(γₚ,τₚ)),τ)$. Then:
  \begin{enumerati}
  \item \label{item:Tsubst} There exists a unique morphism
    $s∶ [𝐍,{\DB}] · {\DB} → {\DB}$ such that
    \begin{itemize}
    \item  for all $τ ∈ 𝕋$, $n ∈ ℕ$, and $f∶ 𝐍 → {\DB}$,
      $s_τ(v_τ(n),f) = f_τ(n)$, and
    \item for all $o ∈ O$, the map $o_{{\DB}}$ satisfies the
      $\ar(o)$-binding condition w.r.t.\ $(s,v)$.
    \end{itemize}
  \item This morphism $s$ turns $({\DB},v,s,a)$ into a De Bruijn $S$-algebra.
  \item This De Bruijn $S$-algebra is initial in $S\DBAlg$.
  \end{enumerati}
\end{thm}

\subsection{Application: values in simply-typed \tmlambda-calculus}
  \label{ss:values}
  We saw in Example~\ref{ex:stl} that the De Bruijn monad of
  simply-typed $λ$-calculus terms admits a simple signature.  But we
  also mentioned in Remark~\ref{rem:values} that (untyped) values may
  be organised as a monad, with terms forming a module over it.  In
  this subsection, as announced, we present a signature for a
  simply-typed version of this.

  We want elements of our De Bruijn monad at any type to be
  \alert{values} of that type.  (Indeed, values are closed under value
  substitution.)

  However, in order to define a signature for this De Bruijn
  monad, we cannot use application. Indeed, application returns terms
  which are not values.

  In order to solve this problem, we need to introduce the following
  auxiliary construction.  The idea is to pack up all occurrences of
  application between layers of value operations (abstraction and
  variable), into a single operation.
  
  We do this by introducing \alert{application binary trees},
    which are proof derivations generated by the following rules,
    \begin{mathpar}
      \inferrule{ }{σ ⊢_{BT} σ} \and
      \inferrule{Γ⊢_{BT} σ → τ \\ Δ⊢_{BT} σ}{Γ,Δ⊢_{BT}τ}
    \end{mathpar}
    \noindent where $Γ,Δ$ denotes concatenation of lists of simple
    types.  Thus a proof of $Γ ⊢_{BT} τ$ is essentially a simply-typed
    term involving only application, with one, linearly used free
    variable for each type in $Γ$, in the same order.  Linearity is
    here used to keep track of all leaves in the typing context, which
    we will now use to define the desired binding signature.

    \begin{defi}
      Let $BT^Γ_σ$ denote the set of such application binary trees
      with conclusion $Γ⊢_{BT}σ$.
  \end{defi}
  We then take as binding signature $S_{λᵥ}$ for simply-typed values
  the one with one operation $L_{π,σ}$ of arity
  \[
  \inferrule{σ ⊢ τ₁ \\ … \\ σ ⊢ τₙ}{⊢ σ → τ} \]
for each simple type $σ$
  and application binary tree $π ∈ BT^{τ₁,…,τₙ}_τ$.

    \begin{exa}
    For a simple example, if $π$ is merely the axiom $τ ↦ τ$,
    then $L_{π,σ}$ has the arity
    $\inferrule{σ ⊢ τ}{⊢ σ → τ}$
    of $λ$-abstraction.
    In this case, $L_{π,σ}$ is thought of as forming $λx:σ.v^τ$ from
    any value $v$ of type $τ$ with an additional variable of type $σ$.
  \end{exa}
  
  \begin{exa}
    For a less trivial, yet basic example, if $π$ is
    \begin{mathpar}
      \inferrule{τ₁ → τ₂ ⊢_{BT} τ₁ → τ₂ \\ τ₁ ⊢_{BT} τ₁}{τ₁ → τ₂, τ₁
        ⊢_{BT} τ₂}~\rlap{,}
    \end{mathpar}
    then $L_{π,σ}$ has arity
    \[
    \inferrule{σ ⊢ τ₁ → τ₂ \\ σ ⊢ τ₁}{⊢ σ → τ₂}~·
    \]
    This operation is thought of as forming
    $λx:σ. (f^{τ₁→τ₂}\ a^{τ₁})$ from values $f$ and $a$.
  \end{exa}
  
    By Theorem~\ref{thm:typedinitiality}, the initial De Bruijn
    $S_{λᵥ}$-algebra has as carrier the initial algebra for the
    induced endofunctor, which is by construction the subset of values.
  \section{Equations}\label{s:equations}
  In this section, we introduce a notion of equational theory for
  specifying (typed) De Bruijn monads, following ideas from
  \cite{FioreHurEquational}.

  \begin{defi}
    \label{def:equational-theory}
    A \alert{De Bruijn equational theory} consists of
    \begin{itemize}
    \item two binding signatures $S$ and $T$, and
    \item two functors $L,R∶ S\DBAlg → T\DBAlg$ over $𝐃𝐁𝐌𝐧𝐝(𝕋)$, i.e.,
      making the following diagram commute serially,
      where $U^S$ and $Uᵀ$ denote the forgetful functors.
      \begin{center}
        \diag{%
          S\DBAlg \& \& T\DBAlg \\
          \& 𝐃𝐁𝐌𝐧𝐝(𝕋) %
        }{%
          (m-1-1.5) edge[labela={L}] (m-1-3.175) %
          (m-1-1.-3) edge[labelb={R}] (m-1-3.183) %
          (m-1-1) edge[labelbl={U^S}] (m-2-2) %
          (m-1-3) edge[labelbr={Uᵀ}] (m-2-2) %
        }
      \end{center}
    \end{itemize}
  \end{defi}

\begin{exa}\label{ex:eqsig}
  Recalling the binding signature $S_Λ$ for $λ$-calculus from
  Example~\ref{ex:siglambda}, let us define a De Bruijn equational
  theory for $β$-equivalence. We take $Tᵦ = (1,0)$, and for any De
  Bruijn $S_Λ$-algebra $X$,
  \begin{itemize}
  \item $L(X)$ has as structure map
    \[\begin{array}[t]{rcl}
        X² & → & X \\
                  (e₁,e₂) & ↦ & \app(\lam(e₁),e₂)
      \end{array}\] while
  \item $R(X)$ has as structure map 
    \[ \begin{array}[t]{rcl}
        X² & → & X \\
        (e₁,e₂) & ↦ & e₁[e₂·\id].
       \end{array}\]
     
    \noindent (Here $e₂·\id$ denotes the assignment $0 ↦ e₂$, $n+1↦v(n)$.)
  \end{itemize}
\end{exa}

  \begin{defi}
    Given an equational theory $E = (S,T,L,R)$, a De Bruijn
    \alert{$E$-algebra} is a De Bruijn $S$-algebra $X$ such that
    $L(X) = R(X)$.

    Let $E\DBAlg$ denote the category of $E$-algebras, with morphisms of
    De Bruijn $S$-algebras between them.
  \end{defi}

  \begin{rem}
    The category $E\DBAlg$ is an equaliser of $L$ and $R$ in $𝐂𝐀𝐓$.
  \end{rem}

  Let us now turn to characterising the initial De Bruijn $E$-algebra,
  for any De Bruijn equational theory $E$.
  For this, we introduce the following relation.
  \begin{defi}
    \label{def:eqrel-ini}
    For any De Bruijn equational theory $E = (S,T,L,R)$, with
    $S = (O,\ar)$ and $T = (O',\ar')$, let $\DB$ denote the initial
    $(𝐍+Σ_S)$-algebra.  We define $∼_E$ to be the smallest equivalence
    relation on $\DB$ satisfying the following rules,
    \begin{mathpar}
      \inferrule{ }{o'_{L(\DB)}(e₁,…,eₚ) ∼_E o'_{R(\DB)}(e₁,…,eₚ)}
      \and
      \inferrule{e₁ ∼_E e'₁ \\ … \\ e_q ∼_E e'_q}{o_{\DB}(e₁,…,e_q)
        ∼_E o_{\DB}(e'₁,…,e'_q)}
    \end{mathpar}
    for all $e,e₁,…$ in $\DB$, $o' ∈ O'$ with $|\ar'(o')| = p$, and
    $o ∈ O$ with $|\ar(o)| = q$.
  \end{defi}

\begin{exa}
  For the equational theory of Example~\ref{ex:eqsig}, the first rule
  instantiates precisely to the $β$-rule, while the second enforces
  congruence.
\end{exa}

\begin{thm}\label{thm:eqn}
  For any equational theory $E = (S,T,L,R)$, $E\DBAlg$ admits an
  initial object, whose carrier set is the quotient $\DB/{∼_E}$.
\end{thm}
\begin{proof}
  We formalised the proof in Coq~\cite{DBCoq}, where
  \begin{itemize}
  \item 
    existence is called
    \coqident{quotsyntax}{ini_morE_model_mor},
   \item uniqueness is called
     \coqident{quotsyntax}{ini_morE_unique}.
  \end{itemize}
  The specific case of $λ$-calculus modulo $βη$-equation has also 
  been mechanised in HOL~\cite{DBHOL}.
\end{proof}

\begin{exa}
  The initial model for the equational theory of
  Example~\ref{ex:eqsig} is the quotient of $λ$-terms in De Bruijn
  representation by $β$-equivalence.
\end{exa}

\section{Mechanised proofs}\label{s:mechanised}
Our theoretical framework is in particular meant to help specifying and
reasoning mechanically about binding syntax using De Bruijn
representation.  To give practical examples, we describe in this
section two implementations, in HOL Light and Coq, that cover several
crucial parts of the theory of this paper, and, in particular
Theorems~\ref{thm:initiality:I} and~\ref{thm:initiality:II}.  The full
source code is available on github~\cite{DBCoq,DBHOL}.

\subsection{HOL Light proof}

Let us start by discussing the HOL Light formalisation.
One key fact is that, despite being much weaker than ZFC set theory or
the Calculus of Inductive Constructions, HOL is expressive enough to
program the De Bruijn encoding and to reason about it.  For instance,
other representations, such as those based on monads over
sets~\cite{DBLP:journals/iandc/HirschowitzM10} (possibly via the nested datatype
technique~\cite{bird_paterson_1999,DBLP:journals/jar/HirschowitzM12}) are
not directly implementable in HOL.

The reader does not need to be familiar with higher-order
logic (HOL) to follow the essential ideas of this section.  HOL is
based on simply-typed $λ$-calculus.  Following a standard
denotational semantics in Zermelo-Fraenkel set theory,
types in HOL can be thought of as non-empty sets.

Our HOL Light implementation is divided into two main parts.
The first part treats the specific case of $λ$-calculus and can be useful to illustrate the essential ideas of this paper in a simple---yet
paradigmatic---setting.
\begin{nota}
In the following, we refer to the theorems and definitions in the HOL code by indicating their name in parentheses 
in \verb|teletype| font, e.g., (\verb|INITIAL_MORPHISM_UNIQUE|).  HOL terms and formulas are enclosed in
backquotes as in \verb|`2 + 2`|, types are prefixed by a colon as in \verb|`:bool`|.
\end{nota}
The type \verb|:dblambda| of $λ$-calculus is simply
defined as the following inductive type
\begin{verbatim}
  let dblambda_INDUCT, dblambda_RECURSION = define_type
    "dblambda = REF num | APP dblambda dblambda | ABS dblambda";;  
\end{verbatim}
The lifting function (\verb|`DERIV`|) and substitution function \verb|`SUBST`| are defined in the obvious way.  The equations that characterise them are summarised in the following theorem (\verb|SUBST_CLAUSES|).
\begin{verbatim}
  |- (!f i. SUBST f (REF i) = f i) /\
     (!f x y. SUBST f (APP x y) = APP (SUBST f x) (SUBST f y)) /\
     (!f x. SUBST f (ABS x) = ABS (SUBST (DERIV f) x)) /\
     (!f. DERIV f 0 = REF 0) /\
     (!f i. DERIV f (SUC i) = SUBST (REF o SUC) (f i))
\end{verbatim}
The names \verb|`SUC`| and \verb|`o`| respectively denote successor
and function composition.  The symbols \verb|`/\`| and \verb|`!`| are
HOL notation for conjunction and universal quantification.  We
recognise:
\begin{itemize}
\item in the first line, the variables map,
\item in the next two lines, the binding conditions for application
  and abstraction, and
\item in the final two lines, the equations defining the
  lifting of an assignment.
\end{itemize}

We show that the functions \verb|`SUBST`| and \verb|`DERIV`| satisfying the above identities are unique (\verb|SUBST_DERIV_UNIQUE|); thus, we have formalised the first point of Theorem~\ref{thm:initiality:I} for the case of $λ$-calculus.

The second point of the Theorem translates in a law for associativity (\verb|SUBST_SUBST|)
and a second law for unitality
(\verb|SUBST_REF|) complementing the first equation of (\verb|SUBST_CLAUSES|) above.

We also provide the classical definition of unary substitution
(\verb|SUBST1|, as found e.g., in~\cite{Huet}) and then show how the
latter is an instance of the former (\verb|SUBST1_EQ_SUBST|).  As
shown by other authors
\cite{DBLP:conf/popl/AbadiCCL90,DBLP:conf/cpp/SchaferST15}, reasoning
on parallel substitution can be significantly easier.  Here for
instance, we prove the associativity of unary substitution in a few
lines (\verb|SUBST1_SUBST1|) by reducing to parallel substitution.
In contrast, proving the same result directly for  unary substitution
is less intuitive: the mere statement of the property to be proved by
induction is tricky to devise.

Next, we introduce the category of De Bruijn monads (\verb|MONAD|,
\verb|MONAD_MOR|) and their associated modules (\verb|MODULE|,
\verb|MODULE_MOR|).
Our Definition~\ref{def:dbmonad} presents De Bruijn monads as triples consisting of a set, an associative substitution operator and a two-sided unit.
In HOL, this is implemented as a type \verb|`:A`| together with a substitution operation \verb|`op:(num->A)->A->A`| and a unit \verb|`e:num->A`|.\footnote{We warn the reader that in this code \text{op} is used for substitution operation and should not be confused with the operations of the syntax $o\in O$ of the previous sections.}
However, the unit is uniquely determined by the substitution operator and it is denoted \verb|`UNIT op`| in our implementation.  Moreover, the type \verb|`:A`| is automatically inferred.  Thus, we simply identify a De Bruijn monad by its substitution operator, that is, we write \verb|`op IN MONAD`| to indicate that \verb|`op`| is (the substitution operator of) a monad.
\begin{verbatim}
  |- !op. op IN MONAD <=>
          (!f g x:A. op g (op f x) = op (op g o f) x) /\
          (!f n. op f (UNIT op n) = f n) /\
          (!x. op (UNIT op) x = x)
\end{verbatim}
The set of morphisms between two De Bruijn monads \verb|`op1`| and \verb|`op2`| is then defined as follows.
\begin{verbatim}
  |- MONAD_MOR (op1,op2) =
     {h:A->B | op1 IN MONAD /\ op2 IN MONAD /\
               (!n. h (UNIT op1 n) = UNIT op2 n) /\
               (!f x. h (op1 f x) = op2 (h o f) (h x))}
\end{verbatim}
Modules are implemented using a similar style.  We implemented the
constructions on modules needed for interpreting binding signatures:
product (\verb|MPROD|) and derivation (\verb|DMOP|).

In this setup, we can state and prove Theorem~\ref{thm:initiality:II} for the $λ$-calculus.
The models of our syntax are De Bruijn monads endowed with functions \verb|`app`| and \verb|`abs`| that are module morphisms (\verb|DBLAMBDA_MODEL|).
\begin{verbatim}
  app IN MODULE_MOR op (MPROD op op, op)
  lam IN MODULE_MOR op (DMOP op op, op)
\end{verbatim}
\begin{rem}
  Both product \verb|MPROD| and derivation \verb|DMOP| expect two
  arguments, for different reasons.  Product expects the two modules
  it takes the product of.  Derivation takes a monad and a module over
  it, and it derives the latter. The monad is needed because
  derivation relies on monadic substitution (in the definition of
  $⇑$).
\end{rem}
Model morphisms (\verb|DBLAMBDA_MODEL_MOR|) are  De Bruijn
monad morphisms that commute with \verb|`app`| and \verb|`abs`|.  We then
obtain the universal property of $λ$-calculus by giving an initial model
morphism (\verb|DBLAMBDAINIT|,
\verb|DBLAMBDAINIT_IN_DBLAMBDA_MODEL_MOR|),
\begin{verbatim}
  |- !op app lam.
        (op,app,lam) IN DBLAMBDA_MODEL
        ==> DBLAMBDAINIT (op,app,lam) IN
            DBLAMBDA_MODEL_MOR ((SUBST,UNCURRY APP,ABS),(op,app,lam))
\end{verbatim}
and by proving its uniqueness (\verb|DBLAMBDAINIT_UNIQUE|).

This part closes with the analogous theorem for the initial-algebra semantics
of $λ$-calculus modulo $βη$-equivalence
(\verb|EXP_MONAD_MOR_LC_EXPMAP|, \verb|LC_EXPMAP_UNIQUE|).  The style
is the one proposed in~\cite{DBLP:journals/iandc/HirschowitzM10} which uses exponential monads, that
is, a monads $M$ endowed with a module isomorphism $\mathsf{abs} : M' \stackrel{\simeq}{\longrightarrow} M$.

The second part of the HOL Light code implements Theorems~\ref{thm:initiality:I}
and~\ref{thm:initiality:II} for arbitrary signatures.

The increased generality comes at a cost in this implementation: since
HOL does not feature dependent types, it is impossible to implement
the term algebra of a given binding signature as a mere type: one has
to resort to a ``well-formedness'' predicate.  From this perspective,
it may be instructive to compare with our Coq implementation, which takes advantage of dependent types.

The above difficulty is solved in the standard way in HOL Light.
First, we build a type \verb|rterm| for \alert{raw terms} over a ``full'' signature, i.e., one
with countably many operations of each arity.  We then introduce an
inductive set (i.e., an inductive predicate) of well-formed terms
(\verb|WELLFORMED_RULES|) that selects the terms respecting a given
signature.  Besides this technical difficulty, the formal development
follows the same pattern as for $λ$-calculus.

The substitution operator is specified by two equations (\verb|TMSUBST_CALUSES|).
\begin{verbatim}
  |- (!f i. TMSUBST f (TMREF i) = f i) /\
     (!f c args. TMSUBST f (FN c args) =
                 FN c (MAP (\(k,x). k,TMSUBST (TMDERIV k f) x) args))
\end{verbatim}
where \verb|`TMREF`| denotes variables and \verb|`FN`| denotes operations from the signature.
The latter takes two arguments, the \alert{name} of the construction \verb|`c`| (a natural number) and a list of pairs \verb|`(k,x)'| where \verb|`k`| is a natural number denoting the number of bound variables and \verb|`x`| is a term.

We formulate the appropriate notion of category of models in this setting (\verb|MODEL|, \verb|MODEL_MOR|) of which the above data constitutes an object (\verb|RTERM_IN_MODEL|).
Then we prove the universal property by giving the initial morphism (\verb|INITIAL_MORPHISM_IN_MODEL_MOR|) and show its uniqueness (\verb|INITIAL_MORPHISM_UNIQUE|).


Finally, to tie the knot, we derive again the universal property for
$λ$-calculus as an instance of this new, more general, framework
(\verb|DBLAMBDA_UNIVERSAL|).

\subsection{Coq proof}
We now discuss the Coq formalisation~\cite{DBCoq}.  Our implementation
addresses only the general case of arbitrary signatures since, as
mentioned before, we do not gain anything by treating a particular
case separately, thanks to dependent types.  As also mentioned
above, the formalisation covers signatures with equations, in the
untyped case.

The reader who wants to skim through the main definitions and
constructions of this implementation can look at the file
\verb|Summary.v|, which reviews the main constructions and results.  The
formalisation has an idiomatic style: the minute details of
implementation pose no significant problem.  Therefore, we just point
the reader to the most relevant definitions.

We start with the file \verb|syntaxdb.v|.
The syntax is defined as an inductive type parameterised by a signature.
\begin{coqdoccode}
\coqdocemptyline
\coqdocnoindent
\coqdockw{Record} \coqdef{syntaxdb.signature}{signature}{\coqdocrecord{signature}} :=\coqdoceol
\coqdocindent{1.00em}
\{ \coqdef{syntaxdb.O}{O}{\coqdocprojection{O}}  : \coqdockw{Type};\coqdoceol
\coqdocindent{2.00em}
\coqdef{syntaxdb.ar}{ar}{\coqdocprojection{ar}} : \coqref{syntaxdb.O:2}{\coqdocmethod{O}} \coqexternalref{::type scope:x '->' x}{http://coq.inria.fr/distrib/V8.12.0/stdlib//Coq.Init.Logic}{\coqdocnotation{\ensuremath{\rightarrow}}} \coqexternalref{list}{http://coq.inria.fr/distrib/V8.12.0/stdlib//Coq.Init.Datatypes}{\coqdocinductive{list}} \coqexternalref{nat}{http://coq.inria.fr/distrib/V8.12.0/stdlib//Coq.Init.Datatypes}{\coqdocinductive{$ℕ$}}\}.\coqdoceol
\coqdocemptyline
\coqdocnoindent
\coqdockw{Inductive} \coqdef{syntaxdb.Z}{Z}{\coqdocinductive{Z}} (\coqdef{syntaxdb.S:4}{S}{\coqdocbinder{S}} : \coqref{syntaxdb.signature}{\coqdocrecord{signature}}) : \coqdockw{Type} :=\coqdoceol
\coqdocindent{1.00em}
\coqdef{syntaxdb.Var}{Var}{\coqdocconstructor{Var}} : \coqexternalref{nat}{http://coq.inria.fr/distrib/V8.12.0/stdlib//Coq.Init.Datatypes}{\coqdocinductive{$ℕ$}} \coqexternalref{::type scope:x '->' x}{http://coq.inria.fr/distrib/V8.12.0/stdlib//Coq.Init.Logic}{\coqdocnotation{\ensuremath{\rightarrow}}} \coqref{syntaxdb.Z:5}{\coqdocinductive{Z}} \coqref{syntaxdb.S:4}{\coqdocvariable{S}}\coqdoceol
\coqdocnoindent
\ensuremath{|} \coqdef{syntaxdb.Op}{Op}{\coqdocconstructor{Op}} : \coqdockw{\ensuremath{\forall}} (\coqdef{syntaxdb.o:7}{o}{\coqdocbinder{o}} : \coqref{syntaxdb.O}{\coqdocprojection{O}} \coqref{syntaxdb.S:4}{\coqdocvariable{S}}), \coqdocinductive{vec} (\coqref{syntaxdb.Z:5}{\coqdocinductive{Z}} \coqref{syntaxdb.S:4}{\coqdocvariable{S}}) (\coqref{syntaxdb.ar}{\coqdocprojection{ar}} \coqref{syntaxdb.o:7}{\coqdocvariable{o}})  \coqexternalref{::type scope:x '->' x}{http://coq.inria.fr/distrib/V8.12.0/stdlib//Coq.Init.Logic}{\coqdocnotation{\ensuremath{\rightarrow}}} \coqref{syntaxdb.Z:5}{\coqdocinductive{Z}} \coqref{syntaxdb.S:4}{\coqdocvariable{S}}.\coqdoceol
\end{coqdoccode}
\smallskip
\noindent Here, \coqdocinductive{vec} $A$ $\ell$ is defined as the
inductive type of vectors of elements $A$, whose length is that of the
list $\ell$.  Assuming a type $X$ equipped with a substitution map
$-[-]$ and variable embedding, the binding condition is defined as
\begin{coqdoccode}
  \coqdocemptyline
  \coqdocnoindent
\coqdockw{Definition} \coqdef{syntaxdb.binding condition}{binding\_condition}{\coqdocdefinition{binding\_condition}} (\coqdef{syntaxdb.a:217}{a}{\coqdocbinder{a}} : \coqexternalref{list}{http://coq.inria.fr/distrib/V8.12.0/stdlib//Coq.Init.Datatypes}{\coqdocinductive{list}} \coqexternalref{nat}{http://coq.inria.fr/distrib/V8.12.0/stdlib//Coq.Init.Datatypes}{\coqdocinductive{$ℕ$}}) (\coqdef{syntaxdb.op:218}{op}{\coqdocbinder{op}} : \coqdocinductive{vec} \coqref{syntaxdb.binding condition.X}{\coqdocvariable{X}} \coqref{syntaxdb.a:217}{\coqdocvariable{a}} \coqexternalref{::type scope:x '->' x}{http://coq.inria.fr/distrib/V8.12.0/stdlib//Coq.Init.Logic}{\coqdocnotation{\ensuremath{\rightarrow}}} \coqref{syntaxdb.binding condition.X}{\coqdocvariable{X}}) :=\coqdoceol
\coqdocindent{2.00em}
\coqdockw{\ensuremath{\forall}} (\coqdef{syntaxdb.f:219}{f}{\coqdocbinder{f}} : \coqexternalref{nat}{http://coq.inria.fr/distrib/V8.12.0/stdlib//Coq.Init.Datatypes}{\coqdocinductive{$ℕ$}} \coqexternalref{::type scope:x '->' x}{http://coq.inria.fr/distrib/V8.12.0/stdlib//Coq.Init.Logic}{\coqdocnotation{\ensuremath{\rightarrow}}} \coqref{syntaxdb.binding condition.X}{\coqdocvariable{X}})(\coqdef{syntaxdb.v:220}{v}{\coqdocbinder{v}} : \coqdocinductive{vec} \coqref{syntaxdb.binding condition.X}{\coqdocvariable{X}} \coqref{syntaxdb.a:217}{\coqdocvariable{a}}),\coqdoceol
\coqdocindent{3.00em}
\coqref{syntaxdb.op:218}{\coqdocvariable{op}} \coqref{syntaxdb.v:220}{\coqdocvariable{v}} \coqref{syntaxdb.binding condition.:::x '[' x ']'}{\coqdocnotation{[}} \coqref{syntaxdb.f:219}{\coqdocvariable{f}} \coqref{syntaxdb.binding condition.:::x '[' x ']'}{\coqdocnotation{]}} \coqexternalref{::type scope:x '=' x}{http://coq.inria.fr/distrib/V8.12.0/stdlib//Coq.Init.Logic}{\coqdocnotation{=}} \coqref{syntaxdb.op:218}{\coqdocvariable{op}} (\coqdocdefinition{vec\_map} (\coqdockw{fun} \coqdef{syntaxdb.n:221}{n}{\coqdocbinder{n}} \coqdef{syntaxdb.x:222}{x}{\coqdocbinder{x}} \ensuremath{\Rightarrow} \coqref{syntaxdb.x:222}{\coqdocvariable{x}} \coqref{syntaxdb.binding condition.:::x '[' x ']'}{\coqdocnotation{[}} \coqref{syntaxdb.f:219}{\coqdocvariable{f}} \coqref{syntaxdb.binding condition.:::x 'x5E' '(' x ')'}{\coqdocnotation{\^{}}} \coqref{syntaxdb.binding condition.:::x 'x5E' '(' x ')'}{\coqdocnotation{(}} \coqref{syntaxdb.n:221}{\coqdocvariable{n}} \coqref{syntaxdb.binding condition.:::x 'x5E' '(' x ')'}{\coqdocnotation{)}}\coqref{syntaxdb.binding condition.:::x '[' x ']'}{\coqdocnotation{]}}) \coqref{syntaxdb.v:220}{\coqdocvariable{v}}).\coqdoceol
\coqdocemptyline
\coqdocnoindent
\end{coqdoccode}%
where
\begin{itemize}
\item $-[-]$ denotes substitution,
\item 
 \begin{coqdoccode}
  \coqdocvar{f} \^{} ( \coqdocvar{n} )
\end{coqdoccode}
denotes $⇑^{n}f$, and
\item  \coqdocdefinition{vec\_map} $f$ maps a vector
$v = (x₁,\dots,xₙ)$ of type \coqdocinductive{vec} $A$ $(a₁,…,aₙ)$
to $(f\ a₁\ x₁, …,f\ aₙ\ xₙ)$.
\end{itemize}
The definition of models is split into two parts: the data, and the properties.

\begin{coqdoccode}
  \coqdocnoindent
\coqdockw{Record} \coqdef{syntaxdb.model data}{model\_data}{\coqdocrecord{model\_data}} (\coqdef{syntaxdb.S:223}{S}{\coqdocbinder{S}} : \coqref{syntaxdb.signature}{\coqdocrecord{signature}}) := \coqdoceol
\coqdocindent{1.00em}
\{ \coqdef{syntaxdb.carrier}{carrier}{\coqdocprojection{carrier}} $:>$ \coqdockw{Type};\coqdoceol
\coqdocindent{2.00em}
\coqdef{syntaxdb.variables}{variables}{\coqdocprojection{variables}} : \coqexternalref{nat}{http://coq.inria.fr/distrib/V8.12.0/stdlib//Coq.Init.Datatypes}{\coqdocinductive{$ℕ$}} \coqexternalref{::type scope:x '->' x}{http://coq.inria.fr/distrib/V8.12.0/stdlib//Coq.Init.Logic}{\coqdocnotation{\ensuremath{\rightarrow}}} \coqref{syntaxdb.carrier:225}{\coqdocmethod{carrier}};\coqdoceol
\coqdocindent{2.00em}
\coqdef{syntaxdb.ops}{ops}{\coqdocprojection{ops}} : \coqdockw{\ensuremath{\forall}} (\coqdef{syntaxdb.o:227}{o}{\coqdocbinder{o}} : \coqref{syntaxdb.O}{\coqdocprojection{O}} \coqref{syntaxdb.S:223}{\coqdocvariable{S}}), \coqdocinductive{vec} \coqref{syntaxdb.carrier:225}{\coqdocmethod{carrier}} (\coqref{syntaxdb.ar}{\coqdocprojection{ar}} \coqref{syntaxdb.o:227}{\coqdocvariable{o}}) \coqexternalref{::type scope:x '->' x}{http://coq.inria.fr/distrib/V8.12.0/stdlib//Coq.Init.Logic}{\coqdocnotation{\ensuremath{\rightarrow}}} \coqref{syntaxdb.carrier:225}{\coqdocmethod{carrier}};\coqdoceol
\coqdocindent{2.00em}
\coqdef{syntaxdb.substitution}{substitution}{\coqdocprojection{substitution}} : \coqexternalref{::type scope:x '->' x}{http://coq.inria.fr/distrib/V8.12.0/stdlib//Coq.Init.Logic}{\coqdocnotation{(}}\coqexternalref{nat}{http://coq.inria.fr/distrib/V8.12.0/stdlib//Coq.Init.Datatypes}{\coqdocinductive{$ℕ$}} \coqexternalref{::type scope:x '->' x}{http://coq.inria.fr/distrib/V8.12.0/stdlib//Coq.Init.Logic}{\coqdocnotation{\ensuremath{\rightarrow}}} \coqref{syntaxdb.carrier:225}{\coqdocmethod{carrier}}\coqexternalref{::type scope:x '->' x}{http://coq.inria.fr/distrib/V8.12.0/stdlib//Coq.Init.Logic}{\coqdocnotation{)}} \coqexternalref{::type scope:x '->' x}{http://coq.inria.fr/distrib/V8.12.0/stdlib//Coq.Init.Logic}{\coqdocnotation{\ensuremath{\rightarrow}}} \coqexternalref{::type scope:x '->' x}{http://coq.inria.fr/distrib/V8.12.0/stdlib//Coq.Init.Logic}{\coqdocnotation{(}}\coqref{syntaxdb.carrier:225}{\coqdocmethod{carrier}} \coqexternalref{::type scope:x '->' x}{http://coq.inria.fr/distrib/V8.12.0/stdlib//Coq.Init.Logic}{\coqdocnotation{\ensuremath{\rightarrow}}} \coqref{syntaxdb.carrier:225}{\coqdocmethod{carrier}}\coqexternalref{::type scope:x '->' x}{http://coq.inria.fr/distrib/V8.12.0/stdlib//Coq.Init.Logic}{\coqdocnotation{)}}\coqdoceol
\coqdocindent{1.00em}
\}.\coqdoceol
\end{coqdoccode}

\begin{coqdoccode}
\coqdocnoindent
\coqdockw{Record} \coqdef{syntaxdb.is model}{is\_model}{\coqdocrecord{is\_model}} \{\coqdef{syntaxdb.S:235}{S}{\coqdocbinder{S}} : \coqref{syntaxdb.signature}{\coqdocrecord{signature}}\}(\coqdef{syntaxdb.m:236}{m}{\coqdocbinder{m}} : \coqref{syntaxdb.model data}{\coqdocrecord{model\_data}} \coqref{syntaxdb.S:235}{\coqdocvariable{S}}) := \{\coqdoceol
\coqdocindent{0.50em}
\coqdef{syntaxdb.substitution ext}{substitution\_ext}{\coqdocprojection{substitution\_ext}} : \coqdockw{\ensuremath{\forall}} (\coqdef{syntaxdb.f:238}{f}{\coqdocbinder{f}} \coqdef{syntaxdb.g:239}{g}{\coqdocbinder{g}} : \coqexternalref{nat}{http://coq.inria.fr/distrib/V8.12.0/stdlib//Coq.Init.Datatypes}{\coqdocinductive{$ℕ$}} \coqexternalref{::type scope:x '->' x}{http://coq.inria.fr/distrib/V8.12.0/stdlib//Coq.Init.Logic}{\coqdocnotation{\ensuremath{\rightarrow}}} \coqref{syntaxdb.m:236}{\coqdocvariable{m}}),  \coqexternalref{::type scope:x '->' x}{http://coq.inria.fr/distrib/V8.12.0/stdlib//Coq.Init.Logic}{\coqdocnotation{(}}\coqdockw{\ensuremath{\forall}} \coqdef{syntaxdb.n:241}{n}{\coqdocbinder{n}}, \coqref{syntaxdb.f:238}{\coqdocvariable{f}} \coqref{syntaxdb.n:241}{\coqdocvariable{n}} \coqexternalref{::type scope:x '=' x}{http://coq.inria.fr/distrib/V8.12.0/stdlib//Coq.Init.Logic}{\coqdocnotation{=}} \coqref{syntaxdb.g:239}{\coqdocvariable{g}} \coqref{syntaxdb.n:241}{\coqdocvariable{n}}\coqexternalref{::type scope:x '->' x}{http://coq.inria.fr/distrib/V8.12.0/stdlib//Coq.Init.Logic}{\coqdocnotation{)}} \coqexternalref{::type scope:x '->' x}{http://coq.inria.fr/distrib/V8.12.0/stdlib//Coq.Init.Logic}{\coqdocnotation{\ensuremath{\rightarrow}}} \coqdockw{\ensuremath{\forall}} \coqdef{syntaxdb.x:240}{x}{\coqdocbinder{x}}, \coqref{syntaxdb.x:240}{\coqdocvariable{x}} \coqref{syntaxdb.:::x '[' x ']'}{\coqdocnotation{[}} \coqref{syntaxdb.f:238}{\coqdocvariable{f}} \coqref{syntaxdb.:::x '[' x ']'}{\coqdocnotation{]}} \coqexternalref{::type scope:x '=' x}{http://coq.inria.fr/distrib/V8.12.0/stdlib//Coq.Init.Logic}{\coqdocnotation{=}} \coqref{syntaxdb.x:240}{\coqdocvariable{x}} \coqref{syntaxdb.:::x '[' x ']'}{\coqdocnotation{[}} \coqref{syntaxdb.g:239}{\coqdocvariable{g}} \coqref{syntaxdb.:::x '[' x ']'}{\coqdocnotation{]}};\coqdoceol
\coqdocindent{0.50em}
\coqdef{syntaxdb.variables subst}{variables\_subst}{\coqdocprojection{variables\_subst}} : \coqdockw{\ensuremath{\forall}} \coqdef{syntaxdb.x:243}{x}{\coqdocbinder{x}} \coqdef{syntaxdb.f:244}{f}{\coqdocbinder{f}}, \coqref{syntaxdb.:::x '[' x ']'}{\coqdocnotation{(}}\coqref{syntaxdb.variables}{\coqdocprojection{variables}} \coqref{syntaxdb.m:236}{\coqdocvariable{m}} \coqref{syntaxdb.x:243}{\coqdocvariable{x}}\coqref{syntaxdb.:::x '[' x ']'}{\coqdocnotation{)}} \coqref{syntaxdb.:::x '[' x ']'}{\coqdocnotation{[}} \coqref{syntaxdb.f:244}{\coqdocvariable{f}} \coqref{syntaxdb.:::x '[' x ']'}{\coqdocnotation{]}} \coqexternalref{::type scope:x '=' x}{http://coq.inria.fr/distrib/V8.12.0/stdlib//Coq.Init.Logic}{\coqdocnotation{=}} \coqref{syntaxdb.f:244}{\coqdocvariable{f}} \coqref{syntaxdb.x:243}{\coqdocvariable{x}};\coqdoceol
\coqdocindent{0.50em}
\coqdef{syntaxdb.ops subst}{ops\_subst}{\coqdocprojection{ops\_subst}} : \coqdockw{\ensuremath{\forall}} (\coqdef{syntaxdb.o:246}{o}{\coqdocbinder{o}} : \coqref{syntaxdb.O}{\coqdocprojection{O}} \coqref{syntaxdb.S:235}{\coqdocvariable{S}}), \coqref{syntaxdb.binding condition}{\coqdocdefinition{binding\_condition}} (\coqref{syntaxdb.variables}{\coqdocprojection{variables}} \coqref{syntaxdb.m:236}{\coqdocvariable{m}}) (\coqref{syntaxdb.substitution}{\coqdocprojection{substitution}} (\coqdocvar{m} := \coqref{syntaxdb.m:236}{\coqdocvariable{m}})) (\coqref{syntaxdb.ops}{\coqdocprojection{ops}} \coqref{syntaxdb.o:246}{\coqdocvariable{o}});\coqdoceol
\coqdocindent{0.50em}
\coqdef{syntaxdb.assoc}{assoc}{\coqdocprojection{assoc}} : \coqdockw{\ensuremath{\forall}} (\coqdef{syntaxdb.f:248}{f}{\coqdocbinder{f}} \coqdef{syntaxdb.g:249}{g}{\coqdocbinder{g}} : \coqexternalref{nat}{http://coq.inria.fr/distrib/V8.12.0/stdlib//Coq.Init.Datatypes}{\coqdocinductive{$ℕ$}} \coqexternalref{::type scope:x '->' x}{http://coq.inria.fr/distrib/V8.12.0/stdlib//Coq.Init.Logic}{\coqdocnotation{\ensuremath{\rightarrow}}} \coqref{syntaxdb.m:236}{\coqdocvariable{m}}) (\coqdef{syntaxdb.x:250}{x}{\coqdocbinder{x}} : \coqref{syntaxdb.m:236}{\coqdocvariable{m}}), \coqref{syntaxdb.x:250}{\coqdocvariable{x}} \coqref{syntaxdb.:::x '[' x ']'}{\coqdocnotation{[}} \coqref{syntaxdb.g:249}{\coqdocvariable{g}} \coqref{syntaxdb.:::x '[' x ']'}{\coqdocnotation{]}} \coqref{syntaxdb.:::x '[' x ']'}{\coqdocnotation{[}} \coqref{syntaxdb.f:248}{\coqdocvariable{f}} \coqref{syntaxdb.:::x '[' x ']'}{\coqdocnotation{]}} \coqexternalref{::type scope:x '=' x}{http://coq.inria.fr/distrib/V8.12.0/stdlib//Coq.Init.Logic}{\coqdocnotation{=}} \coqref{syntaxdb.x:250}{\coqdocvariable{x}} \coqref{syntaxdb.:::x '[' x ']'}{\coqdocnotation{[}} \coqref{syntaxdb.:::x '[' x ']'}{\coqdocnotation{(}}\coqdockw{fun} \coqdef{syntaxdb.n:251}{n}{\coqdocbinder{n}} \ensuremath{\Rightarrow} \coqref{syntaxdb.:::x '[' x ']'}{\coqdocnotation{(}}\coqref{syntaxdb.g:249}{\coqdocvariable{g}} \coqref{syntaxdb.n:251}{\coqdocvariable{n}}\coqref{syntaxdb.:::x '[' x ']'}{\coqdocnotation{)}} \coqref{syntaxdb.:::x '[' x ']'}{\coqdocnotation{[}} \coqref{syntaxdb.f:248}{\coqdocvariable{f}} \coqref{syntaxdb.:::x '[' x ']'}{\coqdocnotation{])}} \coqref{syntaxdb.:::x '[' x ']'}{\coqdocnotation{]}} ;\coqdoceol
\coqdocindent{0.50em}
\coqdef{syntaxdb.id neutral}{id\_neutral}{\coqdocprojection{id\_neutral}} : \coqdockw{\ensuremath{\forall}} (\coqdef{syntaxdb.x:253}{x}{\coqdocbinder{x}} : \coqref{syntaxdb.m:236}{\coqdocvariable{m}}), \coqref{syntaxdb.x:253}{\coqdocvariable{x}} \coqref{syntaxdb.:::x '[' x ']'}{\coqdocnotation{[}} \coqref{syntaxdb.variables}{\coqdocprojection{variables}} \coqref{syntaxdb.m:236}{\coqdocvariable{m}} \coqref{syntaxdb.:::x '[' x ']'}{\coqdocnotation{]}} \coqexternalref{::type scope:x '=' x}{http://coq.inria.fr/distrib/V8.12.0/stdlib//Coq.Init.Logic}{\coqdocnotation{=}} \coqref{syntaxdb.x:253}{\coqdocvariable{x}}\coqdoceol
\coqdocnoindent
\}.\coqdoceol
\coqdocemptyline
\coqdocnoindent
\coqdockw{Record} \coqdef{syntaxdb.model}{model}{\coqdocrecord{model}} (\coqdef{syntaxdb.S:255}{S}{\coqdocbinder{S}} : \coqref{syntaxdb.signature}{\coqdocrecord{signature}}) := \{\coqdoceol
\coqdocindent{0.50em}
\coqdef{syntaxdb.mod carrier}{mod\_carrier}{\coqdocprojection{mod\_carrier}} $:>$ \coqref{syntaxdb.model data}{\coqdocrecord{model\_data}} \coqref{syntaxdb.S:255}{\coqdocvariable{S}};\coqdoceol
\coqdocindent{0.50em}
\coqdef{syntaxdb.mod laws}{mod\_laws}{\coqdocprojection{mod\_laws}} : \coqref{syntaxdb.is model}{\coqdocrecord{is\_model}} \coqref{syntaxdb.mod carrier:257}{\coqdocmethod{mod\_carrier}}\coqdoceol
\coqdocnoindent
\}.\coqdoceol
\coqdocemptyline
\end{coqdoccode}
The symbol $:>$ declares an implicit coercion. For example, given a term $m$ of type
\coqdocrecord{model\_data},
we can then just write $m$ when we actually mean $\textrm{carrier}\ m$. Coq implicitly inserts
the field getter \texttt{carrier} whenever it is necessary, based on the typing constraints.

The first property \coqdocprojection{substitution\_ext} looks superfluous: it intuitively follows from 
the fact that pointwise equal functions are equal.
We however explicitly require this property because the latter general fact,
called function
extensionality, is
not built-in in Coq.\footnote{Alternatively, we could have explicitly assumed
  function extensionality as an (unrestricted) axiom, as we do anyway when
  axiomatising quotient types.}

The remaining of the file \verb|syntaxdb.v| consists of the definition of model morphisms and
the proofs of Theorems~\ref{thm:initiality:I}
and~\ref{thm:initiality:II}.
The rest of the formalisation focuses on initiality for
signatures with equations (Theorem~\ref{thm:eqn}), in the untyped
setting.
Since quotients are not built-in in Coq, we axiomatise a quotient type
$X /\!/ R$ in the file \verb|Quot.v| for each equivalence relation
$R$ on a type $X$, that is, for each $R$ of type $\coqdocrecord{Eqv}\ X$.
The canonical projection is then denoted by $-/R ∶ X → X /\!/ R$.

A De Bruijn equational theory (Definition~\ref{def:equational-theory}) is
defined as a record with four fields.
\begin{coqdoccode}
\coqdocemptyline
\coqdocnoindent
\coqdockw{Record} \coqdef{quotsyntax.equational theory}{equational\_theory}{\coqdocrecord{equational\_theory}} :=\coqdoceol
\coqdocindent{1.00em}
\{ \coqdef{quotsyntax.metavariables}{metavariables}{\coqdocprojection{metavariables}} : \coqdocrecord{signature} ;\coqdoceol
\coqdocindent{2.00em}
\coqdef{quotsyntax.main signature}{main\_signature}{\coqdocprojection{main\_signature}} : \coqdocrecord{signature} ;\coqdoceol
\coqdocindent{2.00em}
\coqdef{quotsyntax.left handside}{left\_handside}{\coqdocprojection{left\_handside}} : \coqref{quotsyntax.half equation}{\coqdocrecord{half\_equation}} \coqref{quotsyntax.main signature:18}{\coqdocmethod{main\_signature}} \coqref{quotsyntax.metavariables:17}{\coqdocmethod{metavariables}} ;\coqdoceol
\coqdocindent{2.00em}
\coqdef{quotsyntax.right handside}{right\_handside}{\coqdocprojection{right\_handside}} : \coqref{quotsyntax.half equation}{\coqdocrecord{half\_equation}} \coqref{quotsyntax.main signature:18}{\coqdocmethod{main\_signature}} \coqref{quotsyntax.metavariables:17}{\coqdocmethod{metavariables}} \coqdoceol
\coqdocindent{1.00em}
\}.\coqdoceol
\coqdocemptyline
\coqdocnoindent
\end{coqdoccode}%
A \coqdocrecord{half-equation} is a functor between two categories of models
preserving the underlying De Bruijn monad. In other words, it provides any model
of the first signature with an algebra structure for
the second signature, and this assignment is compatible with model morphisms.
\begin{minipage}{\textwidth}
\begin{coqdoccode}
\coqdocemptyline
\coqdocnoindent
\coqdockw{Record} \coqdef{quotsyntax.half equation}{half\_equation}{\coqdocrecord{half\_equation}} (\coqdef{quotsyntax.S1:1}{S1}{\coqdocbinder{S1}} : \coqdocrecord{signature})(\coqdef{quotsyntax.S2:2}{S2}{\coqdocbinder{S2}} : \coqdocrecord{signature}) :=\coqdoceol
\coqdocindent{1.00em}
\{\coqdoceol
\coqdocindent{2.00em}
\coqdef{quotsyntax.lift ops}{lift\_ops}{\coqdocprojection{lift\_ops}} $:>$ \coqdockw{\ensuremath{\forall}} (\coqdef{quotsyntax.m:4}{m}{\coqdocbinder{m}} : \coqdocrecord{model} \coqref{quotsyntax.S1:1}{\coqdocvariable{S1}}), \coqdockw{\ensuremath{\forall}} (\coqdef{quotsyntax.o:5}{o}{\coqdocbinder{o}} : \coqdocprojection{O} \coqref{quotsyntax.S2:2}{\coqdocvariable{S2}}), \coqdocinductive{vec} \coqref{quotsyntax.m:4}{\coqdocvariable{m}} (\coqdocprojection{ar} \coqref{quotsyntax.o:5}{\coqdocvariable{o}}) \coqexternalref{::type scope:x '->' x}{http://coq.inria.fr/distrib/V8.12.0/stdlib//Coq.Init.Logic}{\coqdocnotation{\ensuremath{\rightarrow}}} \coqref{quotsyntax.m:4}{\coqdocvariable{m}};\coqdoceol
\coqdocindent{2.00em}
\coqdef{quotsyntax.lift ops subst}{lift\_ops\_subst}{\coqdocprojection{lift\_ops\_subst}} :\coqdoceol
\coqdocindent{3.00em}
\coqdockw{\ensuremath{\forall}} (\coqdef{quotsyntax.m:7}{m}{\coqdocbinder{m}} : \coqdocrecord{model} \coqref{quotsyntax.S1:1}{\coqdocvariable{S1}}) (\coqdef{quotsyntax.o:8}{o}{\coqdocbinder{o}} : \coqdocprojection{O} \coqref{quotsyntax.S2:2}{\coqdocvariable{S2}}),\coqdoceol
\coqdocindent{4.00em}
\coqdocdefinition{binding\_condition} (\coqdocprojection{variables} \coqref{quotsyntax.m:7}{\coqdocvariable{m}}) (\coqdocprojection{substitution} (\coqdocvar{m} := \coqref{quotsyntax.m:7}{\coqdocvariable{m}}))\coqdoceol
\coqdocindent{13.00em}
(@\coqref{quotsyntax.lift ops:6}{\coqdocmethod{lift\_ops}} \coqref{quotsyntax.m:7}{\coqdocvariable{m}} \coqref{quotsyntax.o:8}{\coqdocvariable{o}}) ;\coqdoceol
\coqdocindent{2.00em}
\coqdef{quotsyntax.lift ops natural}{lift\_ops\_natural}{\coqdocprojection{lift\_ops\_natural}} : \coqdockw{\ensuremath{\forall}} (\coqdef{quotsyntax.m1:10}{m1}{\coqdocbinder{m1}} \coqdef{quotsyntax.m2:11}{m2}{\coqdocbinder{m2}} : \coqdocrecord{model} \coqref{quotsyntax.S1:1}{\coqdocvariable{S1}}) (\coqdef{quotsyntax.f:12}{f}{\coqdocbinder{f}} : \coqdocrecord{model\_mor} \coqref{quotsyntax.m1:10}{\coqdocvariable{m1}} \coqref{quotsyntax.m2:11}{\coqdocvariable{m2}})\coqdoceol
\coqdocindent{12.50em}
(\coqdef{quotsyntax.o:13}{o}{\coqdocbinder{o}} : \coqdocprojection{O} \coqref{quotsyntax.S2:2}{\coqdocvariable{S2}})(\coqdef{quotsyntax.v:14}{v}{\coqdocbinder{v}} : \coqdocinductive{vec} \coqref{quotsyntax.m1:10}{\coqdocvariable{m1}} (\coqdocprojection{ar} \coqref{quotsyntax.o:13}{\coqdocvariable{o}})),\coqdoceol
\coqdocindent{4.00em}
\coqref{quotsyntax.lift ops:6}{\coqdocmethod{lift\_ops}} (\coqdocdefinition{vec\_map} (\coqdockw{fun} \coqdocvar{\_} \ensuremath{\Rightarrow} \coqref{quotsyntax.f:12}{\coqdocvariable{f}}) \coqref{quotsyntax.v:14}{\coqdocvariable{v}})  \coqexternalref{::type scope:x '=' x}{http://coq.inria.fr/distrib/V8.12.0/stdlib//Coq.Init.Logic}{\coqdocnotation{=}} \coqref{quotsyntax.f:12}{\coqdocvariable{f}} (\coqref{quotsyntax.lift ops:6}{\coqdocmethod{lift\_ops}} \coqref{quotsyntax.v:14}{\coqdocvariable{v}})\coqdoceol
\coqdocindent{1.00em}
\}.\coqdoceol
\coqdocnoindent
\end{coqdoccode}
\end{minipage}
A model of an equational theory is a model of the main signature equalising
both half-equations, in the sense that they yield equal algebra structures.
\begin{coqdoccode}
\coqdocemptyline
\coqdocnoindent
\coqdockw{Record} \coqdef{quotsyntax.model equational}{model\_equational}{\coqdocrecord{model\_equational}} (\coqdef{quotsyntax.E:43}{E}{\coqdocbinder{E}} : \coqref{quotsyntax.equational theory}{\coqdocrecord{equational\_theory}}) :=\coqdoceol
\coqdocindent{1.00em}
\{ \coqdef{quotsyntax.main model}{main\_model}{\coqdocprojection{main\_model}} $:>$ \coqdocrecord{model} (\coqref{quotsyntax.main signature}{\coqdocprojection{main\_signature}} \coqref{quotsyntax.E:43}{\coqdocvariable{E}}) ;\coqdoceol
\coqdocindent{2.00em}
\coqdef{quotsyntax.model eq}{model\_eq}{\coqdocprojection{model\_eq}} : \coqdockw{\ensuremath{\forall}} \coqdef{quotsyntax.o:46}{o}{\coqdocbinder{o}} (\coqdef{quotsyntax.v:47}{v}{\coqdocbinder{v}} : \coqdocinductive{vec} \coqref{quotsyntax.main model:45}{\coqdocmethod{main\_model}} (\coqdocprojection{ar} \coqref{quotsyntax.o:46}{\coqdocvariable{o}})),\coqdoceol
\coqdocindent{4.00em}
\coqref{quotsyntax.left handside}{\coqdocprojection{left\_handside}} \coqref{quotsyntax.E:43}{\coqdocvariable{E}} \coqref{quotsyntax.main model:45}{\coqdocmethod{main\_model}} \coqref{quotsyntax.o:46}{\coqdocvariable{o}} \coqref{quotsyntax.v:47}{\coqdocvariable{v}} \coqexternalref{::type scope:x '=' x}{http://coq.inria.fr/distrib/V8.12.0/stdlib//Coq.Init.Logic}{\coqdocnotation{=}} \coqref{quotsyntax.right handside}{\coqdocprojection{right\_handside}} \coqref{quotsyntax.E:43}{\coqdocvariable{E}} \coqref{quotsyntax.main model:45}{\coqdocmethod{main\_model}} \coqref{quotsyntax.o:46}{\coqdocvariable{o}} \coqref{quotsyntax.v:47}{\coqdocvariable{v}}\coqdoceol
\coqdocindent{1.00em}
\}.\coqdoceol
\coqdocemptyline
\coqdocnoindent
\end{coqdoccode}%
Following Definition~\ref{def:eqrel-ini}, the initial algebra of an equational
theory is obtained by quotienting the
initial algebra of the main signature by the smallest congruent equivalence
relation relating the images by the algebra structures induced by the two half-equations.

\begin{coqdoccode}
  \coqdocemptyline
\coqdocnoindent
\coqdockw{Inductive} \coqdef{quotsyntax.rel Z}{rel\_Z}{\coqdocinductive{rel\_Z}} (\coqdef{quotsyntax.E:49}{E}{\coqdocbinder{E}} : \coqref{quotsyntax.equational theory}{\coqdocrecord{equational\_theory}}) : \coqdocinductive{Z} (\coqref{quotsyntax.main signature}{\coqdocprojection{main\_signature}} \coqdocvar{E}) \coqexternalref{::type scope:x '->' x}{http://coq.inria.fr/distrib/V8.12.0/stdlib//Coq.Init.Logic}{\coqdocnotation{\ensuremath{\rightarrow}}} \coqdocinductive{Z} (\coqref{quotsyntax.main signature}{\coqdocprojection{main\_signature}} \coqdocvar{E}) \coqexternalref{::type scope:x '->' x}{http://coq.inria.fr/distrib/V8.12.0/stdlib//Coq.Init.Logic}{\coqdocnotation{\ensuremath{\rightarrow}}} \coqdockw{Prop} :=\coqdoceol
\coqdocnoindent
\ensuremath{|} \coqdef{quotsyntax.eqE}{eqE}{\coqdocconstructor{eqE}} : \coqdockw{\ensuremath{\forall}} \coqdef{quotsyntax.o:52}{o}{\coqdocbinder{o}} \coqdef{quotsyntax.v:53}{v}{\coqdocbinder{v}}, \coqref{quotsyntax.rel Z:50}{\coqdocinductive{rel\_Z}} (\coqref{quotsyntax.left handside}{\coqdocprojection{left\_handside}} \coqref{quotsyntax.E:49}{\coqdocvariable{E}} (\coqdocdefinition{ZModel} \coqdocvar{\_}) \coqref{quotsyntax.o:52}{\coqdocvariable{o}} \coqref{quotsyntax.v:53}{\coqdocvariable{v}}) (\coqref{quotsyntax.right handside}{\coqdocprojection{right\_handside}} \coqref{quotsyntax.E:49}{\coqdocvariable{E}} (\coqdocdefinition{ZModel} \coqdocvar{\_}) \coqref{quotsyntax.o:52}{\coqdocvariable{o}} \coqref{quotsyntax.v:53}{\coqdocvariable{v}}) \coqdoceol
\coqdocnoindent
\ensuremath{|} \coqdef{quotsyntax.reflE}{reflE}{\coqdocconstructor{reflE}} : \coqdockw{\ensuremath{\forall}} \coqdef{quotsyntax.z:54}{z}{\coqdocbinder{z}}, \coqref{quotsyntax.rel Z:50}{\coqdocinductive{rel\_Z}} \coqref{quotsyntax.z:54}{\coqdocvariable{z}} \coqref{quotsyntax.z:54}{\coqdocvariable{z}}\coqdoceol
\coqdocnoindent
\ensuremath{|} \coqdef{quotsyntax.symE}{symE}{\coqdocconstructor{symE}} : \coqdockw{\ensuremath{\forall}} \coqdef{quotsyntax.a:55}{a}{\coqdocbinder{a}} \coqdef{quotsyntax.b:56}{b}{\coqdocbinder{b}}, \coqref{quotsyntax.rel Z:50}{\coqdocinductive{rel\_Z}} \coqref{quotsyntax.b:56}{\coqdocvariable{b}} \coqref{quotsyntax.a:55}{\coqdocvariable{a}} \coqexternalref{::type scope:x '->' x}{http://coq.inria.fr/distrib/V8.12.0/stdlib//Coq.Init.Logic}{\coqdocnotation{\ensuremath{\rightarrow}}} \coqref{quotsyntax.rel Z:50}{\coqdocinductive{rel\_Z}} \coqref{quotsyntax.a:55}{\coqdocvariable{a}} \coqref{quotsyntax.b:56}{\coqdocvariable{b}}\coqdoceol
\coqdocnoindent
\ensuremath{|} \coqdef{quotsyntax.transE}{transE}{\coqdocconstructor{transE}} : \coqdockw{\ensuremath{\forall}} \coqdef{quotsyntax.a:57}{a}{\coqdocbinder{a}} \coqdef{quotsyntax.b:58}{b}{\coqdocbinder{b}} \coqdef{quotsyntax.c:59}{c}{\coqdocbinder{c}}, \coqref{quotsyntax.rel Z:50}{\coqdocinductive{rel\_Z}} \coqref{quotsyntax.a:57}{\coqdocvariable{a}} \coqref{quotsyntax.b:58}{\coqdocvariable{b}} \coqexternalref{::type scope:x '->' x}{http://coq.inria.fr/distrib/V8.12.0/stdlib//Coq.Init.Logic}{\coqdocnotation{\ensuremath{\rightarrow}}} \coqref{quotsyntax.rel Z:50}{\coqdocinductive{rel\_Z}} \coqref{quotsyntax.b:58}{\coqdocvariable{b}} \coqref{quotsyntax.c:59}{\coqdocvariable{c}} \coqexternalref{::type scope:x '->' x}{http://coq.inria.fr/distrib/V8.12.0/stdlib//Coq.Init.Logic}{\coqdocnotation{\ensuremath{\rightarrow}}} \coqref{quotsyntax.rel Z:50}{\coqdocinductive{rel\_Z}} \coqref{quotsyntax.a:57}{\coqdocvariable{a}} \coqref{quotsyntax.c:59}{\coqdocvariable{c}}\coqdoceol
\coqdocnoindent
\ensuremath{|} \coqdef{quotsyntax.congrE}{congrE}{\coqdocconstructor{congrE}} : \coqdockw{\ensuremath{\forall}} (\coqdef{quotsyntax.o:60}{o}{\coqdocbinder{o}} : \coqdocprojection{O} (\coqref{quotsyntax.main signature}{\coqdocprojection{main\_signature}} \coqref{quotsyntax.E:49}{\coqdocvariable{E}})) (\coqdef{quotsyntax.v:61}{v}{\coqdocbinder{v}} \coqdef{quotsyntax.v':62}{v'}{\coqdocbinder{v'}} : \coqdocinductive{vec} \coqdocvar{\_} (\coqdocprojection{ar} \coqref{quotsyntax.o:60}{\coqdocvariable{o}})),\coqdoceol
\coqdocindent{2.00em}
\coqdocinductive{rel\_vec} (@\coqref{quotsyntax.rel Z:50}{\coqdocinductive{rel\_Z}} \coqref{quotsyntax.E:49}{\coqdocvariable{E}})  \coqref{quotsyntax.v:61}{\coqdocvariable{v}} \coqref{quotsyntax.v':62}{\coqdocvariable{v'}} \coqexternalref{::type scope:x '->' x}{http://coq.inria.fr/distrib/V8.12.0/stdlib//Coq.Init.Logic}{\coqdocnotation{\ensuremath{\rightarrow}}} \coqref{quotsyntax.rel Z:50}{\coqdocinductive{rel\_Z}} (\coqdocconstructor{Op} \coqref{quotsyntax.o:60}{\coqdocvariable{o}} \coqref{quotsyntax.v:61}{\coqdocvariable{v}}) (\coqdocconstructor{Op} \coqref{quotsyntax.o:60}{\coqdocvariable{o}} \coqref{quotsyntax.v':62}{\coqdocvariable{v'}}).\coqdoceol
\coqdocemptyline
\coqdocnoindent
  \coqdocnoindent
  \coqdockw{Definition} \coqdef{quotsyntax.ZEr}{ZEr}{\coqdocdefinition{ZEr}} (\coqdef{quotsyntax.E:63}{E}{\coqdocbinder{E}} : \coqref{quotsyntax.equational theory}{\coqdocrecord{equational\_theory}}) : \coqdocrecord{Eqv} (\coqdocinductive{Z} (\coqref{quotsyntax.main signature}{\coqdocprojection{main\_signature}} \coqref{quotsyntax.E:63}{\coqdocvariable{E}})) :=\coqdoceol
  \coqdocindent{1.00em}
  \coqdocconstructor{Build\_Eqv} (@\coqref{quotsyntax.rel Z}{\coqdocinductive{rel\_Z}} \coqref{quotsyntax.E:63}{\coqdocvariable{E}}) (@\coqref{quotsyntax.reflE}{\coqdocconstructor{reflE}} \coqref{quotsyntax.E:63}{\coqdocvariable{E}}) (@\coqref{quotsyntax.symE}{\coqdocconstructor{symE}} \coqref{quotsyntax.E:63}{\coqdocvariable{E}})(@\coqref{quotsyntax.transE}{\coqdocconstructor{transE}} \coqref{quotsyntax.E:63}{\coqdocvariable{E}}) .\coqdoceol
  \coqdocnoindent
  \coqdocemptyline
  \coqdocnoindent
  \coqdockw{Definition} \coqdef{quotsyntax.ZE}{ZE}{\coqdocdefinition{ZE}} (\coqdef{quotsyntax.E:64}{E}{\coqdocbinder{E}} : \coqref{quotsyntax.equational theory}{\coqdocrecord{equational\_theory}}) := \coqdocinductive{Z} (\coqref{quotsyntax.main signature}{\coqdocprojection{main\_signature}} \coqref{quotsyntax.E:64}{\coqdocvariable{E}}) \coqdocnotation{//} \coqdocnotation{(}\coqref{quotsyntax.ZEr}{\coqdocdefinition{ZEr}} \coqref{quotsyntax.E:64}{\coqdocvariable{E}}\coqdocnotation{)}.\coqdoceol
  \coqdocemptyline
  \coqdocnoindent
\end{coqdoccode}%
The congruence case \coqdocconstructor{congrE} involves the pointwise
relation $\coqdocinductive{rel\_vec}\  R$ induced on vectors by a relation $R$.

The rest of the file consists in showing that this definition indeed induces an initial model
of the given equational theory. We also provide the instantiation on the equational signature  \verb|LC|$\beta\eta$\verb|_sig| of $\lambda$-calculus modulo $\beta\eta$-equivalence.

  \section{Conclusion}\label{s:conclu}
  We have proposed a simple, set-based theory of syntax with variable
  binding, which associates a notion of model (or algebra) to each binding
  signature, and constructs a term model following De Bruijn
  representation. The notion of model features a substitution
  operation.
 We have experienced the simplicity of this theory by
  implementing it in both Coq and HOL Light.

  We have furthermore equipped the construction with an
  initial-algebra semantics, organising the models of any binding
  signature into a category, and proving that the term model is
  initial therein.  

  We have then studied this initial-algebra semantics in a bit more depth, in two
  directions.
  \begin{itemize}
    \item We have first established a formal link with the
    presheaf-based approach~\cite{fiore:presheaf}, proving
    that well-behaved models (in a suitable sense on each side of the
    correspondence) agree up to an equivalence of categories.
    \item We have then recast the whole initial-algebra
      semantics into two established, abstract frameworks for syntax
      with variable binding, one based on
      strengths~\cite{fiore:presheaf,DBLP:conf/lics/Fiore08}, the
      other on
      modules~\cite{hirscho:lam,DBLP:journals/iandc/HirschowitzM10}.
    \end{itemize}

  Finally, we have shown that our theory extends easily to a
  simply-typed setting, and smoothly incorporates equations and
  transitions.
\subsection*{Funding acknowledgement}
This work was supported in part by a European Research Council (ERC) Consolidator Grant for the project “TypeFoundry”, funded under the European Union’s Horizon 2020 Framework Programme (grant agreement no. 101002277).


\bibliographystyle{alphaurl}
\bibliography{bib}

\end{document}

%% file: macros.tex
\usepackage[toc,page]{appendix}
\usepackage{savesym}
\savesymbol{widering}
\usepackage{xspace,mathpartir,xparse,url,enumitem,multirow,wrapfig,adjustbox,rotating}
\usepackage{thmtools,thm-restate}
\usepackage{ebmath,ebproof,ebutf8}
\usepackage{listings}

\newcommand{\isempty}[1]%
{
  \ifthenelse{\equal{#1}{}}%
    {EMPTY}
    {FULL, it contains the string '#1'}
}
\newcommand{\coqident}[2]{\texttt{\lstinline{#1.#2}}}

\usepackage{environ}

\usepackage{tikz-cats}
\usetikzlibrary{matrix,cd}
\tikzcdset{arrow style=tikz}
\tikzcdset{every arrow/.append style = -{Computer Modern Rightarrow[]}}
\tikzset{every arrow/.append style = -{Computer Modern Rightarrow[]}}
\tikzset{
  labelslt/.style n args={2}{%
    labelat={[left]{$\scriptstyle #1$}}{.45},%
    labelat={[right]{$\scriptstyle #2$}}{.55}%
    } %
  }
\tikzset{
  labelsltat/.style n args={4}{%
    labelat={[left]{$\scriptstyle #1$}}{#3},%
    labelat={[right]{$\scriptstyle #2$}}{#4}%
    } %
  }
\tikzset{twoof/.style={twocenter={#1}}}

\renewcommand{\diaggrandhauteur}{.6}
\renewcommand{\diaggrandlargeur}{1}

\newcommand{\xto}[1]{\xrightarrow{#1}}

\newcommand{\xonto}[1]{\xarrow[onto]{#1}}

\newcommand{\DB}[1][S]{\mathrm{DB}_{#1}}

\DeclareMathOperator{\alg}{-\,\mathbf{alg}}
\DeclareMathOperator{\Mod}{-\,\mathbf{Mod}}
\DeclareMathOperator{\MAlg}{-\,\mathbf{MAlg}}

\DeclareMathOperator{\Mon}{-\,\mathbf{Mon}}
\DeclareMathOperator{\DBAlg}{-\,\mathbf{DBAlg}}

\DeclareMathOperator{\id}{id}

\DeclareMathOperator{\scc}{succ}

\DeclareMathOperator{\lam}{lam}
\DeclareMathOperator{\app}{app}

\newcommand{\Lan}{\operatorname{Lan}}

\DeclareTextMath\tmlambda[lambda]{\lambda}

\newcommand{\ST}{\mathrm{ST}}

\newcommand{\its}{\mathit{int}}
\newcommand{\fin}{\mathit{fin}}

\newcommand{\ar}{{\mathit{ar}}}

\newcommand{\alert}[1]{\textbf{#1}}

\NewDocumentCommand{\GammaInv}{}{\mathpalette\doGammaInv\relax}
\makeatletter
\NewDocumentCommand{\doGammaInv}{mm}{%
  \reflectbox{$\m@th#1\Gamma$}%
}
\makeatother
\renewcommand{\daleth}{\GammaInv}

\newcommand{\ajustedroit}[2][1]{\adjustbox{max width=#1\columnwidth,max height=.95\textheight}{#2}}%